\documentclass[11pt,a4paper]{article}

\usepackage[margin=1in]{geometry}
\usepackage{amsmath}
\usepackage{amssymb}
\usepackage{amsthm}
\usepackage{booktabs}
\usepackage[numbers]{natbib}
\usepackage[hidelinks,unicode]{hyperref}
\newcommand{\doi}[1]{\href{https://doi.org/#1}{\nolinkurl{doi:#1}}}

\hypersetup{
  pdftitle={Certified Split Windows for Parallel Lexing: Recovering Boundaries Where No Byte Certifies},
  pdfauthor={Nicklas Nidhögg},
  pdfsubject={Parallel lexical analysis; deterministic finite automata; maximal munch tokenization},
  pdfkeywords={parallel lexing, tokenization, deterministic finite automata, maximal munch, split windows}
}

\theoremstyle{definition}
\newtheorem{definition}{Definition}
\newtheorem{remark}{Remark}

\theoremstyle{plain}
\newtheorem{lemma}{Lemma}
\newtheorem{theorem}{Theorem}
\newtheorem{corollary}{Corollary}
\newtheorem{proposition}{Proposition}
\newtheorem{invariant}{Invariant}

\newcommand{\code}[1]{\texttt{#1}}
\newcommand{\before}{\mathsf{before}}

\title{Certified Split Windows for Parallel Lexing:\\Recovering Boundaries Where No Byte Certifies}
\author{Nicklas Nidh\"ogg\\[0.4em] \normalsize Independent Researcher\\ \normalsize
\texttt{nicklas.nidhogg@gmail.com}\\ \normalsize ORCID: 0009-0006-6161-6150}
\date{August 10, 2026}

\begin{document}

\maketitle

\begin{abstract}
A certified split point lets a parallel lexer cut unlexed input at one byte with the serial token stream provably
preserved, but several conventional token sets in the predecessor's study certify no byte once string, comment, or
whitespace-run forms are included (arXiv:2608.03473v2). We generalize from a byte to a bounded window, a byte string
such that, whatever the context, the token covering its final byte begins at a known origin. Since the certificate is
conditional on occurrence and may be vacuous, every applicability figure counts only windows carrying an asserted
completely tokenizable occurrence witness. We give a conservative model of a maximal-munch scanner's possible histories
across a window, prove it sound, and decide reachability in it exactly by exhausting a finite quotient of its reachable
configurations. Within the flat, completely tokenizable scope, model-positive answers are semantic certificates;
negatives are relative to the model, which refuses some windows a greedy scanner would allow. For both model and
semantic certificates, existence is PSPACE-complete with or without a witness requirement, and a shortest window has
$2^{\Theta(n)}$ bytes in the worst case over at most $n$ live states. Two smaller searches decide plain model existence,
the second shipped in release 2.2.0 of munch. Of 400 random token sets, 322 of the 337 certifying no byte gain witnessed
windows and four more gain model certificates with occurrence unresolved. No model search reaches its threshold. The
predecessor's exact-empty rows all gain witnessed windows of two to four bytes. A boundary profile marks each gap of an
occurring window forced, crossed or open, and a window-local cut is sound exactly at forced gaps. Whether a given gap is
ever a boundary and whether it is open are PSPACE-complete, while occurrence, the forced verdict and a given origin's
certificate are NL-complete.
\end{abstract}

\section{Introduction}
\label{sec:introduction}

The predecessor of this paper derives, from a compiled token set, the complete set of bytes at which unlexed input can
be cut with the serial token stream preserved, and proves its condition necessary as well as
sufficient~\cite{nidhogg2026splitpoints}. Its sharpest limitation is its own applicability table, in which several of
the studied conventional token sets certify \emph{no} byte because a single string form, comment form, or whitespace run
gives some non-initial live state a transition on every candidate. A fresh 400-grammar random sweep generated for this
study reproduces the pattern, 337 of its 400 sets certifying none. This paper is about the object that refusal leaves
standing, a byte \emph{string} after which the start of the token covering the string's final byte is pinned, regardless
of surrounding context.

Concretely, we ask: for a token set compiled to a DFA and scanned by maximal munch, is there a window
$W = w_0 \cdots w_{k-1}$ and an offset $o$ with $0 \leq o < k$ such that in \emph{every} completely tokenizable input
containing $W$, the token covering the occurrence's final byte begins exactly $o$ bytes into it? A worker that finds $W$
in a completely tokenizable input may then begin scanning at the recovered boundary with no speculation, no state
enumeration, and no DFA-state-recovery pass over the input, as at a certified byte, which is the $k = 1$ case.

The contributions are:

\begin{itemize}
\item
A conservative model of the scanner's possible token-prefix histories across a window, replacing a refuted natural
predecessor, and a soundness proof by a representation invariant (Section~\ref{sec:model}, Section~\ref{sec:soundness}).
\item
A finite quotient of the model's reachable configurations, exact for them by a stated single-occupancy invariant,
turning breadth-first search over windows into a terminating decision procedure for the model
(Section~\ref{sec:quotient}).
\item
Two specializations (Section~\ref{sec:specialization}). At length one the model certifies exactly the bytes the
published predicate reports, so the multi-byte construction is that predicate's conservative continuation. Over a prefix
code, where the certificate is code synchronization under another name, an occurring window is certified at a given
origin exactly when the split there is synchronizing and its right half is a prefix of a codeword.
\item
Two smaller exact searches for the model (Section~\ref{sec:smaller}), a quotient of at most $3^n$ keys over $n$ live
states and an unambiguous automaton of at most $(n + 2)\, 2^{n-1}$ states that follows one chosen origin. The second
returns a shortest certified window, bounds its length, and is the search release 2.2.0 of munch ships as
\code{shortest\_split\_window()}.
\item
A proof that existence of a model certificate, plain or witnessed, is PSPACE-complete (Section~\ref{sec:complexity}), by
a reduction from DFA intersection non-emptiness with one token kind and five symbols that fixes the minimum window
length at $m + 3$ plus the length of the shortest word that all $m$ source DFAs accept. A family of such token sets has
shortest windows of length $2^{\Omega(\sqrt n)}$, with exact minima 4, 6, 11, 18, 67, 68, 429, 850 and 2{,}531 for its
first nine members, searched by a figure program, and sources that share one transition table give shortest windows of
$2^{\Omega(n)}$ bytes, the largest order possible.
\item A characterization of the cloud (Section~\ref{sec:factorization}), whose hypotheses are exactly those of arbitrary
token factorizations of concatenations of tokens, greedy or not. From this the vacuous window follows, token kinds may
be erased before any search, and the model is exact wherever every factorization is the scan's, a class decided in
polynomial time that holds the prefix codes strictly. A shortest witnessed model certificate is found exactly with its
occurrence, by a product with the verifier of the companion paper on certified splitting~\cite{nidhogg2026splitting},
called the splitting companion below, and whether the model misses any witnessed semantic certificate at all, decidable
already from the regularity Section~\ref{sec:related} notes, is decided by an explicit construction that returns a
shortest missed window.
\item An evaluation over all six exact-empty rows of the prior study's applicability table, those certifying no useful
byte under the exact condition, one new cumulative variant, and 400 random token sets generated for this paper: the
witnessed rescue rate where no byte certifies, the window lengths that suffice for the studied C-like and JSON
tokenizations, and rewind-stress checks of 1{,}079{,}392 generated executions that scanned through the window and
contained at least one rewind, 418{,}466 of them completely tokenizable, with zero disagreements against the shipped
scanner (Section~\ref{sec:evaluation}).
\item
A boundary profile of a window, the verdict forced, crossed or open at each of its gaps over every occurrence, with the
certificate as one of its patterns, and a proof that a window-local cut rule that reconciles nothing is sound exactly at
forced gaps (Section~\ref{sec:profile}). Deciding whether a given gap of a given window is a boundary at some occurrence
is PSPACE-complete, and so is deciding whether it is open, by a second reduction from DFA intersection non-emptiness
with one token kind and five symbols whose window always occurs. Whether the window occurs, whether a token crosses a
given gap at some occurrence, the verdict forced and Definition~\ref{def:window} for a given pair, plain or witnessed,
are NL-complete, and an occurrence and a crossing are each found in an input of linear length when one exists, an
occurrence refuting an origin in one of quadratic length, which is the order some token sets need, while a shortest
input refuting a crossed verdict can need $2^{\Omega(n)}$ bytes over $n$ states, the largest order possible for a window
of at most $n$ bytes. Whether a token set has a certified split window at all is PSPACE-complete, and so is whether it
has a witnessed one, even with one token kind over two symbols and every state live. Either way a shortest one has at
most $(n + 1)\, 7^n + 4^n + 1$ bytes over $n$ states and can need $2^{\Omega(n)}$, so its length is $2^{\Theta(n)}$ in
the worst case. Over a token set with a longest token, a window containing a forced gap forces every boundary the scan
from that gap commits with its lookahead inside the window, while no single width decides every gap of the JSON
tokenization, since no bounded window sees whether an unmatched quote lies to its left. Release 2.1.0 of munch decides
the profile exactly within its search cap, and an exploratory campaign over JSON, C-like and byte-pair token sets
measures what it adds to the certificate.
\end{itemize}

The probe decides certificates and model-search results offline from the compiled tables after automaton construction;
the occurrence witnesses and scanner checks are generated executions against the shipped scanner. Executable checks in
the munch repository~\cite{munch} print and assert every figure of the evaluation (Section~\ref{sec:evaluation}) and of
the long-window family (Table~\ref{tab:family}), so a drifted number fails the test suite. Section~\ref{sec:evaluation}
names the programs and the evaluation versions. The boundary-profile measurements of Section~\ref{sec:profile} are an
exploratory campaign that no test asserts.

\section{Preliminaries}
\label{sec:preliminaries}

We use the scanner model of the predecessor~\cite{nidhogg2026splitpoints}. A token set compiles to a DFA $A$ with
initial state $q_0$ and is scanned by maximal munch, in which the scanner runs $A$ from the current position as far as a
transition exists, emits the token of the last accepting configuration passed, resumes immediately after it, and
restarts in $q_0$. An input is \emph{completely tokenizable} when this process consumes it exactly. The setting is
\emph{flat}, meaning one fixed token set compiled to one DFA scanned from one fixed $q_0$, with no lexical modes, no
mode stack, and no semantic scanner state. $A^+$ denotes the live subautomaton, whose states are those both reachable
from $q_0$ and co-accessible to acceptance and whose transitions are those between live states. Write $Q^+$ for its
state set. A token may match the empty string. The scan never emits it, and the next lemma lets every statement below
assume a start state that does not accept.

\begin{lemma}[Positive-width equivalent]
\label{lem:unroll}
Let $A$ be the DFA of a token set whose start state $q_0$ accepts, and let $A'$ be $A$ with a fresh start state $q_0'$
that carries $q_0$'s outgoing transitions and does not accept, $q_0$ itself retained. Then the maximal-munch scans of
$A$ and $A'$ agree on every input, $A'$ accepts exactly the nonempty words $A$ accepts, and no transition of $A'$ enters
$q_0'$.
\end{lemma}

\begin{proof}
Under the scan's convention only an accepting configuration reached after consuming a byte determines an emitted token,
so every emitted token has positive width and a scan records an acceptance only after consuming a byte. From $q_0$ and
from $q_0'$ the first byte leads to the same state, since $q_0'$ carries $q_0$'s transitions, and from there the two
scans traverse the same states and record the same accepting positions. The emitted token and the restart offset agree,
and both scans restart in their own start state. The accepted words of $A'$ are the words of length at least one
accepted by $A$, the empty word being lost because $q_0'$ does not inherit $q_0$'s acceptance. No transition of $A$
targets $q_0'$, which is fresh, and the copied transitions target $q_0$'s successors, so nothing enters $q_0'$.
\end{proof}

Throughout, a token set whose start state accepts is read through $A'$, and the artifact compiles every such set this
way before deciding anything about it, while a set whose start does not accept passes through unchanged. The unrolled
automaton has one state more and the same scan, so nothing below loses generality by assuming that $q_0$ does not
accept. From here on $A$ and $q_0$ name the automaton so read and its start, which is $q_0'$ for a nullable set. For a
nullable set the fresh start is never re-entered, and for any other set the re-entrancy of $q_0$ remains the separate
condition that Section~\ref{sec:specialization} treats.

\begin{definition}[Certified split window]
\label{def:window}
Let $W = w_0 \cdots w_{k-1}$ with $k \geq 1$ and let $o \in \{0, \ldots, k-1\}$. The pair $(W, o)$ is a
\emph{certified split window} for a token set when, for every completely tokenizable input $x$ containing $W$ at offset
$t$, the token of the maximal-munch tokenization of $x$ that contains the byte at offset $t + k - 1$ begins at $t + o$.
\end{definition}

\begin{definition}[Witnessed window]
\label{def:witnessed}
A certified split window $(W, o)$ is \emph{witnessed} when some completely tokenizable input contains $W$.
\end{definition}

Witnessed windows are what the splitting companion calls \emph{occurring}. Definition~\ref{def:window} quantifies over
the inputs containing $W$, so it holds vacuously when no completely tokenizable input contains $W$, and the vacuity is
not hypothetical. Over $\{\code{0}, \code{00}, \code{01}\}$ the model below certifies the window \code{1001} at origin
2, yet no completely tokenizable input contains \code{1001}. Every \code{1} is the tail of a greedily chosen \code{01},
so a token boundary follows the window's first byte, maximal munch must then consume \code{00}, and the final \code{1}
sits at a boundary where no token starts. That argument proves non-occurrence, and the certificate holds because
\code{1001} occurs in a factorization of \code{01001} into tokens that the scan does not take, as
Section~\ref{sec:factorization} explains. The artifact asserts the model prediction and the empty result of its bounded
targeted witness search. Every applicability figure in this paper counts only witnessed certificates. The named rows pin
their witness inputs in the artifact, and for each of the 322 witnessed random grammars the probe constructs and
verifies a concrete completely tokenizable input, with the 322 aggregate asserted; a certificate the bounded witness
search cannot witness is reported as unresolved, never as a rescue. Whenever a statement specifies an alphabet, its
inputs and its windows are words over that alphabet, since a byte that no token contains occurs in no completely
tokenizable input and is a certified split window of one byte by vacuity alone.

For $k = 1$ Definition~\ref{def:window} is the boundary guarantee of a certified split point, since the token containing
the single byte begins at it. Three guarantees should be kept apart, and this paper certifies the strongest. Model
unanimity implies that the covering token's origin is fixed, which in turn implies that some fixed safe boundary exists
inside the window, and neither converse holds. Over $\{\code{a}, \code{ab}, \code{b}\}$ at \code{ab} the covering origin
is fixed yet the model refuses (Section~\ref{sec:strictness}); over $\{\code{a}, \code{abx}, \code{b}, \code{x}\}$ at
\code{ab} the byte \code{a} always begins a token, so offset $0$ is a fixed safe boundary, yet the covering token's
origin is not fixed, since the final \code{b} belongs to \code{b} in the input \code{ab} and to \code{abx} in
\code{abx}. The weakest guarantee is already usable, since a worker can cut at the known boundary and rescan the
window's suffix, and release 2.1.0 of munch decides it exactly within its search cap (Section~\ref{sec:profile}). Its
\code{boundary\_profile()} classifies every gap of an occurring window as \code{must}, \code{never} or \code{may}, a
boundary at every occurrence, at none, or at some. The covering-token form is the stronger, sufficient property this
forward origin-recovery model certifies, and it hands the worker its resumption point directly. The model therefore has
two distinct sources of false negatives. Covering-origin certification is stricter than locating some safe boundary, and
the cloud conservatively represents extra segmentations and refuses some true covering-origin certificates,
$\{\code{a}, \code{ab}, \code{b}\}$ at \code{ab} among them. Note also that the definition says nothing about the tokens
overlapping the window's earlier bytes, and that it does not require the boundary to be the only one inside the window.

\section{The model}
\label{sec:model}

A worker cutting blind knows only that in the final segmentation, the window's first byte is consumed from \emph{some}
live token-prefix state, either inside a token that began at some unknown earlier offset or exactly at a boundary. The
acceptance-gated seed below carries the boundary case, including a window at the start of the input. The model tracks a
\emph{cloud} of hypotheses $(q, \omega)$, each a live state $q$ paired with an origin $\omega$, which is either
$\before$ for a token that began before the window or an in-window offset. The initial cloud $C_0$ is every live state
paired with $\before$.

Reading window byte $w_j$ maps $C_j$ to $C_{j+1}$ by two rules:

\begin{itemize}
\item \textbf{Direct step.} A pair $(q, \omega)$ whose state consumes $w_j$ into a live state moves there with its
origin unchanged. A pair whose state cannot consume $w_j$ into a live state is an impossible history and is dropped,
never restarted.
\item \textbf{Acceptance-gated seed.} If some pair in $C_j$ is accepting, one fresh pair $(\delta(q_0, w_j),\, j)$ is
seeded, provided that target is live. The gate is sound because tokens end only where the automaton accepts and a token
can begin at offset $j$ only where the previous token could have ended or at the input start, the case the initially
open gate represents.
\end{itemize}

One rename accompanies the direct step. Where $q_0$ is not re-entrant in $A^+$, a pair stepping \emph{from} $q_0$ is
beginning a token, so its origin becomes the current offset. Where a non-empty live path returns to $q_0$, the rename is
disabled, since reaching $q_0$ then no longer identifies a boundary. This is the same re-entrancy condition the
length-one certificate carries~\cite{nidhogg2026splitpoints}.

Clouds are \emph{sets} of pairs; the set semantics is load-bearing below, where two rules inserting the same pair yield
one element. The window is \emph{certified at origin $o$} when $C_k \neq \varnothing$ and there is an
$o \in \{0, \ldots, k-1\}$ with $\omega = o$ for every $(q, \omega) \in C_k$. An empty cloud certifies nothing, matching
the implementation, which refuses it, and the empty cloud is absorbing, since every step of it is empty. Non-emptiness
does not guarantee that the window occurs in any completely tokenizable input, since a certified window may be vacuous
under Definition~\ref{def:window}, and Definition~\ref{def:witnessed} sets the witnessed windows apart. Agreement on the
state alone is not enough, since learning that the scan is inside a string literal is knowledge, but not a boundary.
Throughout, $q_0$ is \emph{re-entrant} when some live transition of $A^+$ targets it; non-re-entrancy means no live edge
enters $q_0$ at all.

\paragraph{The discarded predecessor.} A natural variant restarts a trajectory that cannot consume the byte, treating
the failure point as a token boundary, in place of the acceptance-gated seed. That variant is refuted. Over the token
set $\{\code{a}, \code{abc}, \code{bx}, \code{x}\}$ and window \code{abx}, its cloud after \code{ab} is the single pair
carrying origin 0 toward \code{abc}. The byte \code{x} kills that trajectory, the restart replaces it with a token
beginning at offset 2, nothing else survives, and the variant certifies $(\code{abx}, 2)$. Yet the input \code{abx}
itself is completely tokenizable as \code{a} followed by \code{bx}, so the token covering the final byte begins at
offset 1, and the certificate is false at a witnessed occurrence. A failing trajectory is an impossible history, not a
boundary, and restarting it manufactures support for origins no execution justifies. The repaired model certifies
\code{abx} at origin 1, where the scanner does cut, and the case is carried as an asserted row of the artifact.

\section{Soundness}
\label{sec:soundness}

Fix a completely tokenizable input $x$ containing $W$ at offset $t$. For $j \in 1..k$ let $\sigma_j$ be the start of the
token that the final maximal-munch segmentation of $x$ assigns to the byte at $t+j-1$, let
$\rho_j = \delta^*(q_0,\, x[\sigma_j \,..\, t+j))$ be that token's prefix state after consuming through byte $t+j-1$,
and let $\omega_j = \sigma_j - t$, or $\before$ when $\sigma_j < t$.

\begin{lemma}[Representation]
\label{lem:representation}
For every $j \in 1..k$, the pair $(\rho_j, \omega_j)$ is in $C_j$.
\end{lemma}

\begin{proof}
Every $\rho_j$ is live. It is reachable through the actual token prefix, and co-accessible because its token ends at
some $e \geq t+j$ with $\delta^*(\rho_j, x[t+j \,..\, e)) = \delta^*(q_0, x[\sigma_j \,..\, e))$, and the right side is
accepting.

\emph{Base, $j = 1$.} If $\sigma_1 < t$, the state $p = \delta^*(q_0,\, x[\sigma_1 \,..\, t))$ is live and lies in $C_0$
paired with $\before$, and the direct step carries $p$ to $(\rho_1, \before)$. The rename cannot interfere, since $p$
has consumed at least one byte, which makes $p = q_0$ possible only if a non-empty live path returns to $q_0$, and that
disables the rename. If $\sigma_1 = t$, then $\rho_1 = \delta(q_0, w_0)$ and the seed fires, because $C_0$, which is all
of the live states, contains a live accepting state. The fixed input is completely tokenizable and non-empty, so its
first token traces an accepting path from $q_0$, making $q_0$ live with an accepting state reachable from it, and a
reachable accepting state is trivially co-accessible.

\emph{Step, $j$ to $j+1$.} If the byte at $t+j$ continues its token, then $\sigma_{j+1} = \sigma_j$ and the direct step
carries $(\rho_j, \omega_j)$ to $(\delta(\rho_j, w_j), \omega_j)$. In this case the rename does not overwrite the
origin, by the same re-entrancy argument as in the base. If the byte at $t+j$ begins a token, the previous token ended
at $t+j$, so $\delta^*(q_0,\, x[\sigma_j \,..\, t+j))$ accepts. That state is $\rho_j$, in $C_j$ by hypothesis, so the
seed fires and emits exactly $(\delta(q_0, w_j),\, j) = (\rho_{j+1}, \omega_{j+1})$.
\end{proof}

The lemma states containment, not equality, so the cloud may carry hypotheses no execution realizes. The surplus is
harmless in one direction and load-bearing in the other. At a realized occurrence in a completely tokenizable input,
surplus hypotheses can only preserve the actual origin's unanimity or destroy unanimity; they cannot manufacture
unanimity at a false origin.

\begin{theorem}[Soundness]
\label{thm:soundness}
If $C_k \neq \varnothing$ and every pair of $C_k$ carries the same origin $o \neq \before$, then $(W, o)$ is a certified
split window.
\end{theorem}

\begin{proof}
By Lemma~\ref{lem:representation}, $(\rho_k, \omega_k) \in C_k$, so $\omega_k = o$, so $\sigma_k = t + o$, a token start
of the final segmentation. The input was an arbitrary completely tokenizable one containing $W$ at $t$.
\end{proof}

\paragraph{Backup never appears.} Maximal-munch lookahead that a later rewind discards occupies states the lemma says
nothing about; the invariant tracks where the \emph{final} segmentation's tokens begin, not where the read head wanders.
A boundary that a rewind later exposes was seeded by the model step reading the first byte after the accepting position
justifying it. This holds when lookahead crosses the window's left edge, when several accepting positions are passed,
and across chains of consecutive rewinds, each case being an instance of the step case.

\section{The finite quotient and the decision procedure}
\label{sec:quotient}

Acceptance gating alone does not terminate, since over $\code{a}^+$ the automaton accepts after every byte and origins
accumulate without bound. The search deduplicates on a quotient of the cloud, which records the set $B$ of states
carrying $\before$ and, for each state, the number of distinct in-window origins it carries, saturated at two.

\begin{invariant}[Single occupancy]
\label{inv:single}
In every reachable cloud, each surviving in-window origin occupies at most one state.
\end{invariant}

An origin enters the cloud at most once per rule application, into a single state, and the deterministic step moves it
to at most one successor, so it continues in one state or dies. At offset $j$ the acceptance seed and the
non-re-entrant-$q_0$ rename can both propose the fresh origin $j$. Both target $\delta(q_0, w_j)$, and set semantics
coalesces the identical pair, so single occupancy survives the collision. Origins may die, which is why the invariant
says at most one. The pre-window origin differs by design, since it may occupy many states, and it is held apart as the
Boolean support set $B$. An arbitrary cloud could violate the invariant, placing one origin in two states, and there the
saturated counts could not decide unanimity; no such cloud is reachable, and the restriction is load-bearing for
everything below. The key of a cloud is the pair $\kappa(C) = (B,\, m)$ with
$m(q) = \min(2,\, \text{number of in-window origins at } q)$.

\begin{lemma}[Congruence, depth-aligned] \label{lem:congruence} Let $C_u$ and $C_v$ be clouds reached by the search
after words $u$ and $v$. If $\kappa(C_u) = \kappa(C_v)$, then for every byte $b$ the successors
$\mathrm{step}(C_u, b, |u|)$ and $\mathrm{step}(C_v, b, |v|)$ are either both empty or both non-empty, and when
non-empty, \[ \kappa(\mathrm{step}(C_u, b, |u|)) = \kappa(\mathrm{step}(C_v, b, |v|)), \] with the certification verdict
agreeing on both sides. For successor-key equality alone, it is enough that the inserted offsets be \emph{fresh},
exceeding every in-window origin of their respective clouds. The certification verdict additionally needs legal depths,
and the search supplies both, since it steps with the word lengths $|u|$ and $|v|$ and every origin in a length-$\ell$
cloud lies below $\ell$. Freshness is not decorative. Over $\{\code{a}^+, \code{b}\}$ the cloud after \code{a} carries
origin $0$, and stepping on \code{a} with offset $0$ merges the seed into that origin while offset $1$ creates a second
one, so the same key would have two different successors under arbitrary offsets. \end{lemma}

\begin{proof}
The acceptance gate reads only which occupied states accept, and the occupied states are $B \cup \{q : m(q) > 0\}$, a
function of the key. The new $\before$-support is the live deterministic image of $B$, excluding the contribution from
$q_0$ when $q_0$ is non-re-entrant; that excluded contribution joins the fresh-origin rename handled below. Over the
single-token set $\{\code{a}\}$, stepping $C_0$ on \code{a} renames $q_0$'s $\before$ origin to $0$, so
$B' = \varnothing$, where the unexcluded image would wrongly retain $\delta(q_0, \code{a})$. Consider the in-window
counts at a successor state $q'$. The direct step carries each origin from its unique state
(Invariant~\ref{inv:single}), so origins arriving at $q'$ from distinct predecessors are distinct. The exact update at
each target sums the old-origin contributions arriving there and adds one fresh origin if the seed or the
non-re-entrant-$q_0$ rename fires toward it. The seed and the rename are combined by logical or, since they denote the
same fresh origin and the cloud is a set. The update then saturates at two. When $q_0$ is non-re-entrant, no live edge
enters $q_0$, so after the first step no reachable cloud contains $q_0$ with an in-window origin. For a non-re-entrant
start state the rename decision is readable from $q_0 \in B$. Whether the rename fires is determined by $q_0 \in B$,
which is key-readable, and by re-entrancy, which is fixed automaton data outside the key. Whether the seed fires is
determined by the key's accepting support, and the freshness hypothesis says that the fresh origin collides with no
surviving origin. Note that saturation commutes with this update in the only direction needed, since a saturated state
maps its surviving origins together, and a successor's saturated sum computed from saturated counts equals the
saturation of the exact sum. A state's count can also drop to zero when its transition is missing, which the update
handles as an empty contribution. Fresh origins are equivariant under renaming of origin values, and no rule ever reads
an origin's value, only equality between origins, so keys computed on either side agree. Emptiness is key-readable, and
certification of the full window is the key predicate $B = \varnothing$ with total count exactly one, using
Invariant~\ref{inv:single} to identify one total count with one origin in one state.
\end{proof}

One tempting simplification is false and worth flagging, namely that a state holding two origins can never contribute to
unanimity later. Over $\{\code{a}^+, \code{b}\}$, after \code{aa} the \code{a}-state carries two in-window origins, as
well as $\before$. On \code{b} all of them die while their pre-step acceptance opens the gate for one fresh
\code{b}-origin, and the cloud is unanimous. The true statement, and the one the lemma uses, is that co-located origins
either follow the same live successor together or all die, and that origins which die leave no trace the key must
distinguish.

\begin{theorem}[Decision procedure]
\label{thm:decision}
Breadth-first search over windows, deduplicated on keys, decides whether the model certifies any window for a given
token set, and returns the minimum certified length; it does not enumerate every certified word. Each state contributes
a factor of $2 \times 3$ to the key space, so at most $6^{|Q^+|}$ keys are ever retained; each key expands into at most
$256$ successors, each computable from the key in $O(|Q^+|)$ time, for a worst case of
$O(256 \cdot |Q^+| \cdot 6^{|Q^+|})$ time and $O(|Q^+| \cdot 6^{|Q^+|})$ space.
\end{theorem}

By Lemma~\ref{lem:congruence} the walk cannot miss the existence of a certifying continuation, nor change the minimum
certified length. After one matched step both successor clouds sit at fresh depths again, so the lemma inducts over
every common suffix. Breadth-first order retains a shallowest representative of each key, and a certifying suffix from a
discarded deeper occurrence yields an equal or shorter certificate from the retained representative. Deduplication may
still skip individual certified words, but exhaustion proves the model admits none at any length. Exhaustion does not
prove semantic non-existence, because the model is conservative (Section~\ref{sec:strictness}).

The displayed bounds describe an abstract key-to-key search, while the probe realizes the procedure with one concrete
representative cloud and word per key, using the quotient for deduplication, so its cost additionally depends on
representative size. It is budgeted, and its three outcomes are certified, exhausted, and inconclusive once the retained
key count exceeds a fixed threshold of 200{,}000 keys. The evaluation below reports zero inconclusive model searches,
and the retained key counts it reports are one indicator of the search footprint rather than a complete cost model.
Section~\ref{sec:smaller} replaces the $6^{|Q^+|}$ bound by $3^{|Q^+|}$ and by $(|Q^+| + 2)\, 2^{|Q^+| - 1}$, and shows
that no polynomial algorithm decides the question unless $\mathrm{P} = \mathrm{PSPACE}$.

\section{Smaller searches and the complexity of the model}
\label{sec:smaller}

Write $n = |Q^+|$ for the number of live states. The quotient of Section~\ref{sec:quotient} is exact and finite, and the
evaluation's retained-key counts stand against its $6^n$ bound. Two smaller presentations decide the same question. The
first coarsens the key to three statuses per state; the second guesses the one origin that survives and is an
unambiguous automaton over the certified windows. The second is the search release 2.2.0 of munch ships, and it is also
the route to the model's exact complexity. Existence of a model certificate is PSPACE-complete, and a shortest window
can need $2^{\Omega(n)}$ bytes.

\subsection{A three-status quotient}
\label{sec:three}

Fix a reachable cloud. A live state is \emph{empty} when no pair occupies it, \emph{single} when exactly one pair
occupies it and that pair carries an in-window origin, and \emph{blocked} otherwise, when it carries $\before$ or at
least two in-window origins. The three-status key of a cloud is the pair $\kappa_3(C) = (S, B)$ of its single and its
blocked states. A state carrying $\before$ beside an in-window origin is blocked, so $\kappa_3$ identifies keys that the
six-status $\kappa$ keeps apart.

\begin{theorem}[Three-status quotient]
\label{thm:three}
Let $C_u$ and $C_v$ be clouds reached by the search after words $u$ and $v$ with $\kappa_3(C_u) = \kappa_3(C_v)$. For
every byte $b$ the successors $\mathrm{step}(C_u, b, |u|)$ and $\mathrm{step}(C_v, b, |v|)$ are both empty or both
non-empty, and when non-empty their three-status keys agree. A reachable cloud is certified exactly when
$B = \varnothing$ and $|S| = 1$. Breadth-first search over three-status keys therefore decides whether the model
certifies any window and returns the minimum certified length. It retains at most $3^n$ keys, in
$O(256 \cdot n \cdot 3^n)$ time and $O(n \cdot 3^n)$ space.
\end{theorem}

\begin{proof}
The occupied states are $S \cup B$, so the key says whether the acceptance gate is open. When $q_0$ is not re-entrant,
no live edge enters it, so $q_0$ is occupied only in $C_0$, where it carries $\before$ and is blocked; the key says
whether the rename fires, through $q_0 \in B$. When it fires, the contribution of $q_0$ to the direct step is set aside
and counted as the fresh origin, as in Lemma~\ref{lem:congruence}. Now fix a successor state $q'$ and collect what
arrives there. If a blocked predecessor $p$ has $\delta(p, b) = q'$, it brings $\before$ or two distinct in-window
origins, all of which move to $q'$ together, so $q'$ is blocked whatever else arrives. Otherwise every predecessor of
$q'$ is single and brings one in-window origin. These origins are pairwise distinct by Invariant~\ref{inv:single}, since
one origin occupies one state, and the fresh origin, when the seed or the rename targets $q'$, is distinct from all of
them by freshness. Hence $q'$ is blocked, single or empty as the number of arriving origins is at least two, one or
zero. Every quantity used is a function of $(S, B)$, of $b$ and of fixed automaton data, hence the successor keys agree,
and the successor is empty exactly when every state is empty, which the key says too.

For the verdict, a non-empty cloud is certified when no pair carries $\before$ and all pairs share one in-window origin.
No $\before$ and no two origins at one state means $B = \varnothing$, and then every occupied state is single; by single
occupancy the number of distinct origins is the number of occupied states, so one origin means $|S| = 1$. The congruence
inducts over every continuation as in Theorem~\ref{thm:decision}, with the same depth alignment, and the bounds follow
from the three statuses per state.
\end{proof}

The coarsening loses nothing, since the key forgets only distinctions that the transition and certification rules never
read. At a blocked state the six-status key additionally separates whether it carries $\before$ and how many in-window
origins it holds up to two, and the proof above reads neither. The figure program of Section~\ref{sec:evaluation} runs
this search beside the gate's and finds the same shortest length on every named row, and Table~\ref{tab:windows} gives
the retained keys of both.

\subsection{The chosen-origin search}
\label{sec:chosen}

A certified window has one surviving origin. Guessing that origin at its birth, and carrying its state beside the
support of every other hypothesis, turns the search into a nondeterministic automaton over windows. Its states have two
forms: $U(S)$, for $S \subseteq Q^+$, means that no origin has been chosen yet and that $S$ is the set of states the
cloud occupies, and $V(S, q)$, with $q \in Q^+$ and $S \subseteq Q^+ \setminus \{q\}$, means that the chosen origin sits
at $q$ and that the pairs carrying any other origin occupy exactly the states of $S$. The initial state is $U(Q^+)$, and
the accepting states are the $V(\varnothing, q)$. The state $U(\varnothing)$ accepts nothing and has no successor, and
no transition enters it, so it is reached only as the initial state of a token set with no live state.

Take a byte $b$ at $U(S)$. The rename fires when $q_0 \in S$ and $q_0$ is not re-entrant. Let $T$ be the live image of
$S$ under $b$, with the contribution of $q_0$ set aside when the rename fires. A fresh origin is born at
$r = \delta(q_0, b)$ when $r$ is live and either $S$ contains an accepting state or the rename fired. Without a birth
the successor is $U(T)$, when $T$ is non-empty. With a birth there are two successors. Declining the birth moves to
$U(T \cup \{r\})$, and choosing it moves to $V(T, r)$, which is allowed only when $r \notin T$. Now take a byte $b$ at
$V(S, q)$. Let $q' = \delta(q, b)$; there is no successor when $q'$ is not live. Let $T$ be the live image of $S$,
together with $\delta(q_0, b)$ when that state is live and $S \cup \{q\}$ contains an accepting state. There is no
successor when $q' \in T$, and otherwise the successor is $V(T, q')$. No rename arises in a $V$ state, since a
non-re-entrant $q_0$ is unoccupied after the first byte and a re-entrant one disables the rename.

\begin{theorem}[Chosen-origin automaton]
\label{thm:chosen}
The automaton accepts exactly the windows the model certifies. It is unambiguous: a certified window has exactly one
accepting run, which chooses the window's origin at that origin's birth. It has at most
\[ K \;=\; 2^n + n\, 2^{n-1} \;=\; (n + 2)\, 2^{n-1} \]
states, and at most two successors per state and byte.
\end{theorem}

\begin{proof}
Two invariants hold along every run, by induction on the bytes read. In $U(S)$ after a word $u$, the set $S$ is the
support of $C_u$, since the gate reads only which occupied states accept, the direct step maps the support to its live
image, the rename is read from $q_0 \in S$ as in Theorem~\ref{thm:three}, and a fresh origin occupies $\delta(q_0, b)$.
In $V(S, q)$ after $u$, the cloud $C_u$ contains a pair $(q, o)$, where $o$ is the origin chosen when the run left $U$,
and $S$ is the set of states occupied by pairs with an origin other than $o$. The pair $(q, o)$ moves to
$(\delta(q, b), o)$ when that state is live. Otherwise $o$ dies, and no later cloud contains it, since every origin born
later carries a larger offset. The other pairs move to the live image of $S$, and the fresh origin, when born, is
another origin and lands at $\delta(q_0, b)$, while the gate reads the support $S \cup \{q\}$ of $C_u$. The successor
$V(T, q')$ thus records exactly the state of $o$ and the support of the others, as long as $q' \notin T$. When
$q' \in T$ the run stops, and nothing is lost. The state $q'$ then carries $o$ and a second origin, and by determinism
the two move together or die together at every later step, so no later cloud has $o$ as its only origin.

A run in $V(\varnothing, q)$ after $w$ thus says $C_w = \{(q, o)\}$ with $o$ in the window, which is certification.
Conversely let $w$ be certified with origin $o$. At the byte $w_o$ the origin $o$ is born at $\delta(q_0, w_o)$, by the
seed or by the rename, and that state carries no other origin, for otherwise $o$ could never become the only one. Choose
$o$ there and decline every other birth. The chosen pair survives to the end, since a pair carrying $o$ is in $C_w$, and
it never shares a state with another origin, since both would then reach $C_w$; its final $S$ is empty, because every
pair of $C_w$ carries $o$. For unambiguity, an accepting run chooses the origin that is alone at the end, which is the
unique origin of $C_w$ and has one birth, and the run is deterministic before the choice, its transitions deterministic
after it. For the count, there are $2^n$ sets $S$ for the $U$ states, counting the empty support, which no transition
enters, and for the $V$ states there are $n$ choices of $q$, with $S$ a subset of the remaining $n - 1$ states. Only a
birth before the choice branches.
\end{proof}

\begin{corollary}[Shortest window]
\label{cor:chosen}
Breadth-first search over the automaton decides whether the model certifies any window and returns a shortest one, in
$O(256 \cdot n \cdot K)$ time and $O(n \cdot K)$ space. The origin is the depth at which the returned path leaves $U$.
If a certified window exists, one of length at most $K - 1$ exists. Existence of a model certificate is in PSPACE in the
size of the compiled automaton.
\end{corollary}

\begin{proof}
A path from $U(Q^+)$ to an accepting state spells a certified window, and conversely. Breadth-first order finds a
shortest one, and the origin is read off the one $U$-to-$V$ edge of the path. A shortest accepting path repeats no
state, so its length is below $K$. A state is written in $O(n)$ bits and a successor is computed in polynomial time, so
a path of length below $K$ is guessed holding one state and a counter of $O(n)$ bits, which places existence in
nondeterministic polynomial space, and so in PSPACE.
\end{proof}

\begin{remark}
\label{rem:determinize}
Determinizing the chosen-origin automaton gives back the three-status quotient once equivalent dead states are
identified. After a word whose cloud is non-empty with three-status key $(S, B)$, the set of automaton states reached is
$U(S \cup B)$, where the run that chose nothing survives, together with $V((S \cup B) \setminus \{q\},\, q)$ for every
$q \in S$. A $V$ run survives exactly while its origin is alone at its state, which is the single status, and its other
support is then the rest of the occupied states. A non-empty word that empties the cloud leaves no run alive, so an
emptied cloud corresponds to the empty subset. The subset construction stays within these special subsets, which is why
no second exponential appears.
\end{remark}

Release 2.2.0 of munch ships the search of Corollary~\ref{cor:chosen} as \code{shortest\_split\_window()}, beside the
window decision \code{is\_split\_window()} of release 1.4.0. The search holds one node per automaton state reached,
packed as a chosen state and a bit per live state, and reads the window and its origin back along the trail of parents.
A node budget, $2^{20}$ by default, bounds the nodes held, and running out is reported as such, never as an answer. The
library's test reproduces the first six minima of the family of Section~\ref{sec:long}, and the figure program of
Section~\ref{sec:evaluation} runs the same search over the named token sets. Table~\ref{tab:windows} reports the nodes
the search holds when its shortest window appears.

\subsection{PSPACE-completeness}
\label{sec:complexity}

The upper bound of Corollary~\ref{cor:chosen} is matched by a lower bound that holds with one token kind and five
symbols. The source problem is intersection non-emptiness: given complete DFAs $A_1, \ldots, A_m$ over a common
alphabet, is some word accepted by all of them? Kozen proved it PSPACE-complete~\cite[Lemma~3.2.3,
p.~261]{kozen1977lower}. We take the source DFAs over the alphabet $\{\code{a}, \code{b}\}$, which loses nothing, since
a fixed-length binary code of the letters, decoded by a polynomial number of added states with a rejecting sink for
malformed codes, carries any alphabet to this one. We also keep in each $A_i$ only the states reachable from its start,
which changes no language.

\begin{theorem}[PSPACE-completeness]
\label{thm:pspace}
Deciding whether a compiled token set admits a model certificate of any length is PSPACE-complete. So is deciding
whether it admits a witnessed one. Both remain PSPACE-hard for token sets with one token kind over five symbols, whose
start state neither accepts nor is re-entered and whose every other state accepts.
\end{theorem}

Membership is Corollary~\ref{cor:chosen} for the plain problem and Theorem~\ref{thm:witnessed} below for the witnessed
one, and the rest of this section is the reduction. Let $A_i = (Q_i, \{\code{a}, \code{b}\}, \delta_i, s_i, F_i)$ for
$i = 1, \ldots, m$, with $m \geq 1$. The token DFA has the symbols $\code{a}, \code{b}, \code{r}, \code{t}, \code{\#}$,
and its states are the root $q_0$, an error state $E$, a guard $G$, a result state $Y$, timers $T_0, \ldots, T_m$,
headers $H_0, \ldots, H_m$, and a disjoint copy $(i, q)$ of every state $q$ of every $A_i$. The root is the only state
that does not accept, and nothing enters it; every other state accepts the one token kind. The root has the two
transitions $q_0 \xrightarrow{\code{r}} H_0$ and $q_0 \xrightarrow{\code{\#}} Y$, and the other transitions are in
Table~\ref{tab:reduction}, where an entry marked undefined is absent. On the remaining bytes of the alphabet no state
has a transition.

\begin{table}[!htbp]
\centering
\caption{The transitions of the reduction's token DFA at every state but the root, which has
$q_0 \xrightarrow{\code{r}} H_0$ and $q_0 \xrightarrow{\code{\#}} Y$ only. The letter $x$ ranges over the source
alphabet $\{\code{a}, \code{b}\}$.}
\label{tab:reduction}
\begin{tabular}{lcccc}
\toprule
State & \code{r} & \code{t} & \code{\#} & $x$ \\
\midrule
$E$ & $T_m$ & $E$ & $E$ & $E$ \\
$G$ & $T_m$ & $E$ & undefined & $G$ \\
$Y$ & $T_m$ & $E$ & $E$ & $E$ \\
$T_j$, $j > 0$ & $T_{j-1}$ & $E$ & $E$ & $E$ \\
$T_0$ & $T_0$ & $G$ & $E$ & $E$ \\
$H_{i-1}$, $i \leq m$ & $H_i$ & $(i, s_i)$ & $E$ & $E$ \\
$H_m$ & $H_m$ & $G$ & $E$ & $E$ \\
$(i, q)$ & $T_m$ & $E$ & undefined if $q \in F_i$, else $E$ & $(i, \delta_i(q, x))$ \\
\bottomrule
\end{tabular}
\end{table}

The construction has $\sum_i |Q_i| + 2m + 6$ states and takes polynomial time. Every state is live. The root reaches
$H_j$ on $\code{r}^{j+1}$, $Y$ on \code{\#}, $E$ on \code{\#\#}, $T_m$ on \code{\#r}, each lower timer on further
\code{r}'s, $G$ on $\code{r}^{m+1}\code{t}$, the start $(i, s_i)$ on $\code{r}^{i}\code{t}$, and every other $(i, q)$ on
the word that reaches $q$ in $A_i$ from there. Every state but the root accepts, and the root has an accepting
successor, so all are co-accessible. The root is not re-entrant, since no transition targets it, and the rename is in
force. Two facts about the table carry the proof. First, on $\code{r}, \code{t}, \code{a}, \code{b}$ every state but the
root has a transition, so on these letters no pair at a non-root state ever dies. Second, every occupied state after the
first byte is a non-root state and accepts, so the gate is open at every step after the first, and it is open at the
first because $C_0$ holds accepting states.

\begin{lemma}[Common word to window]
\label{lem:reduction-forward}
If every $A_i$ accepts $w$, then $W = \code{r}^{m+1}\, \code{t}\, w\, \code{\#}$ is certified at origin $|W| - 1$.
Moreover $W$ is itself completely tokenizable, with its final byte a token of its own.
\end{lemma}

\begin{proof}
Follow the cloud. During the $m + 1$ bytes \code{r}, every pair at $E$, $G$, $Y$, a timer or a source copy moves into
the timers and down to $T_0$, where it stays, every pair at a header moves up to $H_m$ and stays, and the root's pair is
renamed to origin $0$ at $H_0$. The gate is open at every byte, so a fresh origin $j$ is born at $H_0$ at each byte
$j = 0, \ldots, m$; the rename and the seed coincide at $j = 0$. After the block, origin $j$ sits at $H_{m-j}$, one
origin per header, with $(T_0, \before)$ and $(H_m, \before)$ beside them. The byte \code{t} sends $T_0$ and $H_m$ to
$G$, and $H_{i-1}$ to $(i, s_i)$ for $i = 1, \ldots, m$, and since the root has no \code{t} edge no origin is born. The
letters of $w$ keep $G$ fixed and drive each $(i, s_i)$ to $(i, \delta_i(s_i, w))$, with no birth. On the final
\code{\#} every pair at $G$ dies, and every source pair dies too, because its state is in $F_i$, while the open gate
seeds $(Y, |W| - 1)$; the final cloud is that singleton. For the scan, the run of $W$ from the root passes
$H_0, \ldots, H_m$, enters $G$ on \code{t} and stays there through $w$, accepting at every byte. Since $G$ has no
\code{\#} edge, the first token is $W$ without its last byte, and the scan restarts at \code{\#}, which the root
consumes into $Y$.
\end{proof}

\begin{lemma}[Window to common word]
\label{lem:reduction-back}
Suppose the model certifies some window. Then it certifies a prefix of it of the form
$P\, \code{r}^{L}\, \code{t}\, w\, \code{\#}$ at its last byte, with $L \geq m + 1$ and $w \in \{\code{a}, \code{b}\}^*$
accepted by every $A_i$.
\end{lemma}

\begin{proof}
Take a certified window and cut it at the first byte after which the cloud is certified. Call the result
$W' = w_0 \cdots w_{k-1}$, so $C_k$ is certified and $C_{k-1}$ is not. A byte outside the five symbols empties the
cloud, which is absorbing and never certified, so every byte of $W'$ is one of the five, and $C_{k-1}$ is non-empty. The
last byte is \code{\#}. Were it one of $\code{r}, \code{t}, \code{a}, \code{b}$, every pair of $C_{k-1}$ at a non-root
state would survive with its origin intact, since the rename touches only the root. Those pairs carry $\before$ for
$k = 1$, and for $k > 1$ whatever made $C_{k-1}$ uncertified, $\before$ or two distinct origins, and a fresh origin only
adds a pair, which cannot certify an uncertified cloud; so $C_k$ would not be certified. On \code{\#}, every pair at a
non-root state moves to $E$ or dies, and the open gate seeds $(Y, k - 1)$. The root's pair, present only in $C_0$,
becomes that same seeded pair by the rename. Certification at origin $k - 1$ needs every pair of $C_{k-1}$ at a non-root
state to die, so those pairs occupy only $G$ and accepting source copies.

It remains to show that $W'$ ends in $\code{r}^{L}\, \code{t}\, w\, \code{\#}$ with $L \geq m + 1$ and $w$ a binary
word, and that every $A_i$ accepts $w$. For the suffix, follow one pair, the sentinel, through the window. If some byte
of $w_0 \cdots w_{k-2}$ is \code{\#}, let $j$ be the position of the last one. The cloud $C_j$ is non-empty, since the
empty cloud is absorbing and $C_{k-1}$ is not empty, and the gate is open, so $(Y, j) \in C_{j+1}$. Otherwise take
$(E, \before) \in C_0$ and set $j = -1$. The bytes $w_{j+1}, \ldots, w_{k-2}$ are not \code{\#}, so they are among
$\code{r}, \code{t}, \code{a}, \code{b}$, on which no pair at a non-root state dies, and the sentinel survives to
$C_{k-1}$ with its origin, following the table from $Y$ or $E$. From $Y$, $E$, the timers and $G$, these four letters
lead only to $Y$, $E$, the timers and $G$, so the sentinel never reaches a header or a source copy, and to die on the
final \code{\#} it must sit at $G$ in $C_{k-1}$.

Call $z = w_{j+1} \cdots w_{k-2}$ the word the sentinel reads. Among these states, $G$ is entered only from $T_0$ on
\code{t}; it is left on \code{r} and \code{t}, and kept on \code{a} and \code{b}. Hence the last \code{t} of $z$ is read
from $T_0$, and only the letters \code{a} and \code{b} follow it; that is, $z = z' \code{t} w$ with
$w \in \{\code{a}, \code{b}\}^*$, and the sentinel is at $T_0$ after $z'$.

The timer $T_0$ is entered only from $T_1$ or from itself, on \code{r}, and every other letter from a timer leads to $E$
or $G$. Let $\code{r}^L$ be the longest run of \code{r}'s ending $z'$. The state before that run is reached by a letter
other than \code{r}, or is the sentinel's initial state, so it is among $Y$, $E$, $G$. If $L = 0$ the sentinel is there
after $z'$, not at $T_0$. Otherwise the run enters $T_m$ on its first byte and needs $m$ more to reach $T_0$, so
$L \geq m + 1$ and $W' = P\, \code{r}^{L}\, \code{t}\, w\, \code{\#}$.

During the last $m + 1$ bytes of the run $\code{r}^L$ the gate is open and the root's \code{r} edge is live, so an
origin is born at $H_0$ at each of them. The one born $j$ bytes before the run's end sits at $H_j$ when the run ends,
for $j = 0, \ldots, m$, so these origins are distinct, one per header. The \code{t} sends $H_{i-1}$ to $(i, s_i)$ for
$i = 1, \ldots, m$, and $w$ drives each to $(i, \delta_i(s_i, w))$. Every pair of $C_{k-1}$ dies on the final \code{\#},
and a source copy dies there exactly when its state is in $F_i$, so $\delta_i(s_i, w) \in F_i$ for every $i$.
\end{proof}

\begin{proof}[Proof of Theorem~\ref{thm:pspace}]
Lemma~\ref{lem:reduction-forward} maps a common word to a certified window, which is witnessed, since $W$ is its own
completely tokenizable occurrence, and Lemma~\ref{lem:reduction-back} maps any certified window, witnessed or not, to a
common word. The construction is polynomial, so both problems are PSPACE-hard, and they are in PSPACE by
Corollary~\ref{cor:chosen} and Theorem~\ref{thm:witnessed}.
\end{proof}

The reduction fixes the minimum length exactly, at the source's minimum plus $m + 3$.

\begin{corollary}[Minimum length]
\label{cor:reduction-length}
Let $\ell$ be the length of a shortest word every $A_i$ accepts. Then a shortest model-certified window of the
constructed token set has length $\ell + m + 3$, and so does a shortest witnessed one.
\end{corollary}

\begin{proof}
Lemma~\ref{lem:reduction-forward} gives a witnessed window of length $\ell + m + 3$. By Lemma~\ref{lem:reduction-back},
any certified window has a certified prefix ending in $\code{r}^{L}\, \code{t}\, w\, \code{\#}$ with $L \geq m + 1$ and
$|w| \geq \ell$, so its length is at least $\ell + m + 3$; a witnessed window is certified, so the same bound holds for
it.
\end{proof}

The theorem concerns the model, which only approximates Definition~\ref{def:window} from below. It also concerns the
compiled automaton, and since a token set given as regular expressions may compile to exponentially many states, no
bound in the expression size is claimed. Unless $\mathrm{P} = \mathrm{PSPACE}$ it rules out a polynomial-time algorithm
for the model's existence question, so replacing the quotient by another search cannot give one. Theorem~\ref{thm:long}
below shows separately that returning a certified window can require exponentially many bytes. For the definition
itself, the decision for a given pair $(W, o)$ is NL-complete by Theorem~\ref{thm:fixed-nl}, while
Theorem~\ref{thm:gap-pspace} proves two neighboring questions PSPACE-complete, whether a given gap of a given window is
a boundary at some occurrence and whether it is open. Whether a token set has a witnessed certified split window at all
is PSPACE-complete as well, by Proposition~\ref{prop:existence}, whose reduction is a different one, from the words of
rank one of a partial automaton, and Theorems~\ref{thm:semantic-length} and~\ref{thm:plain} give the length of a
shortest one, which is exponential in the worst case too. Theorem~\ref{thm:plain} proves both without the witness.

\subsection{Shortest windows can be long}
\label{sec:long}

Corollary~\ref{cor:chosen} bounds a shortest certified window by $K - 1$, which is exponential in $n$. The reduction
instantiated on cycles gives shortest windows of length $2^{\Omega(\sqrt n)}$, so no polynomial bound holds, and the
figure program searches their exact minima. Instantiated on sources that share one transition table it gives
$2^{\Omega(n)}$, which is the order of the upper bound.

For $m \geq 1$ let the source $A_i$, for $i = 1, \ldots, m$, be a cycle of $i$ states numbered $0, \ldots, i - 1$,
placed behind the header $H_{i-1}$. It starts at $0$, both letters advance by one modulo $i$, and it accepts at state
$i - 1$ alone, so it accepts exactly the words whose length is $-1$ modulo $i$. Let $L_m = \mathrm{lcm}(1, \ldots, m)$.
A word every source accepts has length $-1$ modulo $L_m$, and $\code{a}^{L_m - 1}$ is a shortest one. The reduction
gives a token set with
\[ n_m \;=\; \sum_{i=1}^{m} i + 2m + 6 \;=\; \frac{m^2 + 5m + 12}{2} \]
states, and by Corollary~\ref{cor:reduction-length} its shortest certified window, plain or witnessed, has length
$R_m = L_m + m + 2$. One such window is $\code{r}^{m+1}\, \code{t}\, \code{a}^{L_m - 1}\, \code{\#}$, certified at its
last byte and completely tokenizable as two tokens.

\begin{theorem}[Long windows]
\label{thm:long}
For every $m \geq 1$ the token set above has $n_m$ live states in its constructed automaton, which is not minimal, and a
shortest certified window of length $R_m \geq 2^m / (m + 1) + m + 2$. Hence $R_m = 2^{\Omega(\sqrt{n_m})}$. No
polynomial $p$ bounds the length of a shortest model-certified window by $p(n)$ over all compiled token sets, nor of a
shortest witnessed one. Other sources give more. For infinitely many $N$ there is a token set with one token kind over
five symbols whose automaton has $N$ states, all live, and whose shortest model-certified window has $2^{\Omega(N)}$
bytes, as has its shortest witnessed one. No larger order is possible, since over $n$ live states a shortest
model-certified window has fewer than $K$ bytes, by Corollary~\ref{cor:chosen}, and a shortest witnessed one fewer than
$K (n + 1)\, 2^n$, the size of the product of Theorem~\ref{thm:witnessed}. These are bounds for the certificates of the
model. The shortest windows of Definition~\ref{def:window}, witnessed or not, have bounds of their own, of the same
order, in Theorems~\ref{thm:semantic-length} and~\ref{thm:plain}.
\end{theorem}

\begin{proof}
It remains to show $L_m \geq 2^m / (m + 1)$. Fix a prime $p$ and write $v_p$ for the exponent of $p$ in an integer. By
the formula for the exponent of $p$ in a factorial,
\[ v_p\binom{m}{k} \;=\; \sum_{j \geq 1} \left( \left\lfloor \frac{m}{p^j} \right\rfloor - \left\lfloor
\frac{k}{p^j} \right\rfloor - \left\lfloor \frac{m - k}{p^j} \right\rfloor \right). \] Each term is $0$ or $1$, since
$m = k + (m - k)$ and the floor of a sum exceeds the sum of the floors by at most one. Terms with $p^j > m$ vanish, so
$v_p\binom{m}{k} \leq \lfloor \log_p m \rfloor = v_p(L_m)$. This holds for every prime, so every $\binom{m}{k}$ divides
$L_m$. In particular $L_m \geq \max_k \binom{m}{k} \geq 2^m / (m + 1)$, since the $m + 1$ binomial coefficients sum to
$2^m$. Since $n_m$ is quadratic in $m$, $m = \Theta(\sqrt{n_m})$ and the growth follows. For the claim about
polynomials, $R_m / n_m^d$ is unbounded for every fixed degree $d$.

For the larger order, Vorel~\cite[Theorem~8, p.~378]{vorel2014subset} gives, for infinitely many sizes $|Q|$, a strongly
connected complete deterministic automaton $B = (Q, \{\code{a}, \code{b}\}, \delta)$, its two letters written \code{a}
and \code{b} here, and a set $S \subseteq Q$ that some word sends to a single state, such that a shortest word doing so
has length $2^{\Omega(|Q|)}$. Let $z$ be a shortest one, let $\delta(S, z) = \{f\}$ and $S = \{s_1, \ldots, s_m\}$, and
take the sources $A_i = (Q, \{\code{a}, \code{b}\}, \delta, s_i, \{f\})$. Every source accepts $z$, and a word every
source accepts sends $S$ to $f$ and is no shorter than $z$. The sources have the same transitions and the same accepting
state, so the token automaton of Section~\ref{sec:complexity} can hold $Q$ once in place of its $m$ copies. The rows of
the states $(i, q)$ in Table~\ref{tab:reduction} become one row for each state $q$ of $Q$, with \code{\#} undefined at
$f$ alone, and \code{t} leads from $H_{i-1}$ to $s_i$. Lemmas~\ref{lem:reduction-forward} and~\ref{lem:reduction-back}
hold for this automaton, with $(i, \delta_i(s_i, w))$ read as $\delta(s_i, w)$. Their proofs use of a source state only
its row of the table. The pairs that \code{t} sends from the headers to $s_1, \ldots, s_m$ carry $m$ distinct origins,
so they stay $m$ pairs whatever states they come to share, and all of them die on the final \code{\#} exactly when $w$
sends every $s_i$ to $f$. Corollary~\ref{cor:reduction-length} follows as before, and a shortest model-certified window,
plain or witnessed, has $|z| + m + 3$ bytes. The automaton has $N = |Q| + 2m + 6 \leq 3 |Q| + 6$ states, so that length
is $2^{\Omega(N)}$. Every state is live, since every state but the root accepts and the states of $Q$ are reached from
$s_1$, $B$ being strongly connected. For the upper bounds, a shortest accepting path of the chosen-origin automaton
repeats no state, and neither does one of the product.
\end{proof}

\begin{table}
\centering
\caption{The long-window family: the constructed automaton's live states (the minimal one has three fewer from $m = 2$
on), $\mathrm{lcm}(1, \ldots, m)$, the shortest certified window's length, equal to $L_m + m + 2$ on every row, and the
nodes the chosen-origin search holds when it appears. Every window is certified at its last byte and scans as two
tokens.}
\label{tab:family}
\begin{tabular}{crrrr}
\toprule
$m$ & States $n_m$ & $L_m$ & Shortest window & Nodes held \\
\midrule
1 & 9 & 1 & 4 & 18 \\
2 & 13 & 2 & 6 & 41 \\
3 & 18 & 6 & 11 & 110 \\
4 & 24 & 12 & 18 & 253 \\
5 & 31 & 60 & 67 & 1{,}226 \\
6 & 39 & 60 & 68 & 1{,}781 \\
7 & 48 & 420 & 429 & 11{,}038 \\
8 & 58 & 840 & 850 & 27{,}017 \\
9 & 69 & 2{,}520 & 2{,}531 & 87{,}518 \\
\bottomrule
\end{tabular}
\end{table}

The family does not reach every size $n$, but its sizes are dense enough. Since $n_{m+1} - n_m = m + 3 \leq n_m$, for
every $n \geq 9$ the member with the largest $n_m \leq n$ has $n < 2 n_m$, so some token set with at most $n$ live
states has a shortest window of length $2^{\Omega(\sqrt n)}$. Table~\ref{tab:family} gives the exact minima for
$m = 1, \ldots, 9$. The figure program finds them by the chosen-origin search, replays each window through the plain
cloud at the claimed origin, scans it, and asserts every entry of the table. The result is explicit, a shortest window
that is long and not merely hard to find, though it does not exclude a compressed description of long windows. It
concerns the model, and for the shortest windows of Definition~\ref{def:window}, witnessed or not,
Theorems~\ref{thm:semantic-length} and~\ref{thm:plain} give the worst-case order $2^{\Theta(n)}$.

\section{Specializations}
\label{sec:specialization}

\subsection{Length one}
\label{sec:length-one}

At length one the model collapses to the published certificate, and the correspondence is exact at the level of the
\emph{shipped predicate}, the one that withholds vacuously certified bytes, rather than of the bare condition.
Throughout this section the token set is non-empty, read through its positive-width equivalent where a token matches the
empty string, and $q_0$ is live.

Call a byte $b$ \emph{useful} when $\delta(q_0, b)$ is defined and live. The predicate \code{is\_split\_point($b$)}
of~\cite{nidhogg2026splitpoints} holds exactly when $b$ is useful, no live state other than $q_0$ has a $b$-transition
into a live state, and $q_0$ is not re-entrant in $A^+$. Two quantifier readings coincide here without loss. A
transition target that is co-accessible makes its source co-accessible, so ``reachable state with a live-target
$b$-transition'' and ``live state with a live-target $b$-transition'' name the same states.

\begin{theorem}[Specialization]
\label{thm:specialization}
For every byte $b$, the model certifies $(b, 0)$, the only origin a length-one window admits, if and only if
\code{is\_split\_point($b$)} holds.
\end{theorem}

\begin{proof}
Compute $C_1$ explicitly. $C_0$ is every live state paired with $\before$, and it contains a live accepting state, since
$q_0$ is live and an accepting state reachable from it is trivially co-accessible, so the seed's acceptance gate is open
at offset 0. Reading $b$ yields exactly:

\begin{itemize}
\item
from the seed, the pair $(\delta(q_0, b),\, 0)$, present if and only if $b$ is useful;
\item
from the direct step on $(q_0, \before)$, present if and only if $b$ is useful, namely the pair $(\delta(q_0, b),\, 0)$
when $q_0$ is not re-entrant, by the rename, and $(\delta(q_0, b),\, \before)$ when it is;
\item
from the direct step on every other live state $q$ that has a live-target $b$-transition, the pair
$(\delta(q, b),\, \before)$.
\end{itemize}

If \code{is\_split\_point($b$)} holds, only the first two items contribute, the rename fires, and both surviving
contributions are the single pair $(\delta(q_0, b),\, 0)$, so $C_1$ is the singleton and the model certifies $(b, 0)$.
Conversely, suppose the model certifies $(b, 0)$, so $C_1$ is non-empty and unanimous at origin $0$. If $b$ were not
useful, the first two items would be absent and every pair of $C_1$ would carry $\before$ from the third, or $C_1$ would
be empty; either way certification fails, so $b$ is useful. If some live $q \neq q_0$ had a live-target $b$-transition,
the third item would place a $\before$-pair in $C_1$ and break unanimity, so none does. If $q_0$ were re-entrant, the
second item would place $(\delta(q_0, b),\, \before)$ in $C_1$ beside the seed's origin-$0$ pair and break unanimity, so
$q_0$ is not re-entrant. These are exactly the predicate's three clauses.
\end{proof}

The vacuous case lands on the predicate's side of the published condition-versus-predicate distinction by construction,
since a byte no live state consumes empties the cloud, which the model refuses, exactly as the shipped predicate
withholds a byte the bare condition certifies vacuously. At length one the model is therefore exactly the shipped
predicate, and the multi-byte construction is its conservative continuation. The artifact additionally asserts the
agreement byte for byte on every grammar of the evaluation before each search. After Theorem~\ref{thm:specialization},
that check verifies only the implementation.

\subsection{The prefix-code case}
\label{sec:codes}

The other specialization is by token set. A \emph{prefix code} $X \subseteq \Sigma^+$ has no word that is a proper
prefix of another, and read as a token set it is the setting where the certificate needs neither maximal-munch priority
nor the cloud's origin bookkeeping. At any boundary at most one codeword begins, so the completely tokenizable inputs
are exactly $X^*$, each with the one factorization the scan computes. Code theory has two notions of a boundary forced
by bounded context. A \emph{synchronizing pair} of $X$~\cite[p.~375]{berstel2010codes} is a pair
$(x, y) \in X^* \times X^*$ such that $u x y v \in X^*$ implies $u x \in X^*$ and $y v \in X^*$ for all
$u, v \in \Sigma^*$. The second comes from the synchronization pairs of Fici, Romana, Sciortino and
Urbina~\cite[Section~3.3]{fici2025morphisms}, stated there for circular factors of a morphic image. Specializing them to
a morphism $\mu$ sending the letters of an alphabet $\Gamma$ to the codewords, so that $X = \mu(\Gamma)$, and to
$L = \Gamma^*$, we call a split $w = w_1 w_2$ of a factor $w$ of $X^*$ \emph{synchronizing} when every occurrence of $w$
in a word of $X^*$ has a factorization boundary between $w_1$ and $w_2$. Neither is a certified window as it stands,
since neither says which codeword covers the window's final byte, and the certificate says nothing about the halves
being in $X^*$; the exact relation is the following.

\begin{proposition}[Prefix codes]
\label{prop:codes}
Let $X$ be a prefix code read as a token set, and let $W$ be a window occurring in some word of $X^*$, with origin $o$,
split as $W = W_{<o}\, W_{\geq o}$ at the origin.
\begin{enumerate}
\item $(W, o)$ is certified if and only if the split is synchronizing and $W_{\geq o}$ is a prefix of some codeword.
\item If $(x, y)$ is a synchronizing pair with $y$ nonempty and $y = c_1 \cdots c_m$ is its factorization into
codewords, then $(xy,\ |x| + |c_1 \cdots c_{m-1}|)$ is certified. Conversely, if $(W, o)$ is certified with
$W_{<o} \in X^*$ and $W_{\geq o} \in X$, then $(W_{<o}, W_{\geq o})$ is a synchronizing pair.
\end{enumerate}
\end{proposition}

\begin{proof}
Throughout, a completely tokenizable input is a word of $X^*$ and its token boundaries are the boundaries of its one
factorization, since at a boundary the codeword the text spells is the only codeword that begins there.

(1) If $(W, o)$ is certified, every occurrence has a boundary at offset $o$, so the split is synchronizing, and the
codeword covering the final byte begins at $o$, so $W_{\geq o}$ lies within that codeword from its start and is a prefix
of it. Conversely, take an occurrence of $W$ at offset $t$ in a word of $X^*$. The split being synchronizing puts a
boundary at $t + o$, where some codeword $c'$ begins. Let $c$ be a codeword with prefix $W_{\geq o}$. Were $c'$ to end
inside $W_{\geq o}$, it would be a proper prefix of $c$, which a prefix code forbids; so $c'$ extends through
$W_{\geq o}$, and the token covering the final byte begins at $t + o$.

(2) Let $(x, y)$ be a synchronizing pair and take an occurrence $u x y v \in X^*$. Then $u x \in X^*$ places a boundary
before $y$, and $y v \in X^*$ with $y \in X^*$ factors from that boundary through $y$'s own codewords, since at each of
their starts the codeword the text spells is the one $y$ contains. Hence $c_m$ begins at $|u x| + |c_1 \cdots c_{m-1}|$
and covers the final byte of $xy$. Conversely, if $(W, o)$ is certified with $W_{<o} \in X^*$ and $W_{\geq o} \in X$,
then every $u W v \in X^*$ has a boundary at $|u| + o$, so $u W_{<o}$ is a concatenation of codewords and so is
$W_{\geq o} v$, which is the synchronizing-pair condition.
\end{proof}

The proposition places the classical case. Over a prefix code the certified windows that occur are the synchronizing
splits whose right half sits inside one codeword. A synchronizing pair of Berstel, Perrin and Reutenauer with a nonempty
right component is a certified window on its concatenation, with the origin at the start of the right component's last
codeword, and a certified window whose left half is in $X^*$ and whose right half is a codeword is a synchronizing pair.
The occurrence hypothesis sets the vacuous case of Definition~\ref{def:window} aside, since a window no word of $X^*$
contains is certified whatever its right half spells. The named token sets of Section~\ref{sec:evaluation} fall outside
the code case. An identifier is a proper prefix of a longer identifier, so no C-like token set is a prefix code, and the
JSON number \code{1} is a proper prefix of \code{11}, so JSON is not one either; over
$\{\code{a}, \code{ab}, \code{b}\}$ the word \code{ab} has two factorizations and maximal munch picks one. The random
sample of Section~\ref{sec:evaluation} is drawn without that restriction and contains prefix codes, which the
proposition covers. Beyond the code case the certificate is defined through the scan, the origin bookkeeping of
Section~\ref{sec:model} tests it soundly, and the evaluation measures that reach.

\section{Strictness of the model}
\label{sec:strictness}

The model refuses windows a greedy scanner would allow. The conservatism is deliberate. The seed rule uses acceptance as
the only license a token needs to begin, which over-approximates greedy behavior. This section exhibits one source of
conservatism and shows that nonvacuous strictness begins at length two. Both examples below are asserted artifact rows,
whose assertion checks that the model refuses the window and that an exhaustive oracle over every completely tokenizable
input up to a length bound finds, at every occurrence, the token covering the window's final byte beginning at the
claimed origin, with the occurrence count pinned exactly. The covering-token check is the property of
Definition~\ref{def:window}; checking merely that some token begins at the origin passes false covering-origin claims,
such as $\{\code{a}, \code{abx}, \code{b}, \code{x}\}$ at \code{ab}, where the input \code{ab} tokenizes as
\code{a}$\mid$\code{b} while the fixed boundary at the occurrence remains a safe cut.

\emph{Example one.} Over $\{\code{a}, \code{ab}, \code{b}\}$ the window \code{ab} is semantically certified at origin 0,
and universally so. The byte \code{a} occurs only as a token's first byte, so a token begins at every occurrence offset
$t$, and maximal munch there prefers \code{ab} over \code{a}, so the covering token of the final byte begins at $t$. The
oracle confirms the argument over all inputs to length 14, with 98{,}305 occurrences and zero violations. The model
refuses the window. After \code{a}, the cloud is the single pair carrying origin 0, but \code{a} is a token, so reading
\code{b} seeds a competing trajectory at origin 1, the segmentation \code{a}$\mid$\code{b} that greedy scanning never
chooses, and unanimity is lost. Section~\ref{sec:factorization} shows that such segmentations are the whole of the
surplus.

\emph{Example two.} Over $\{\code{ab}, \code{abc}, \code{c}\}$ the window \code{abc} is semantically certified at origin
0, again universally, since \code{a} occurs only token-initially, so a token begins at $t$, and maximal munch prefers
\code{abc} over \code{ab} there. The oracle confirms it over all inputs to length 12 with 932 occurrences and zero
violations. Both examples instantiate the same $(u,\, uv,\, v)$ shape. They share the mechanism and differ in prefix
depth, the competing origin arriving one byte in for the first and two bytes in for the second, where the accepting
proper prefix \code{ab} seeds \code{c}, the segmentation \code{ab}$\mid$\code{c} that maximal munch forgoes.

Each has length at least two, and that is not an accident of the examples.

\begin{corollary}[Non-vacuous strictness begins at length two]
\label{cor:strictness}
Let $b$ be a byte occurring in some completely tokenizable input. If the model refuses $(b, 0)$, then $(b, 0)$ is not a
certified split window. The occurrence hypothesis is necessary, since a byte no completely tokenizable input contains
satisfies Definition~\ref{def:window} vacuously while the model refuses its emptied cloud, and such vacuous
disagreements are not strictness.
\end{corollary}

\begin{proof}
Since $b$ occurs, the final segmentation's covering token consumes $b$ at that occurrence, a transition from a live
state into a live state, so if only $q_0$ has a live-target $b$-transition and $q_0$ is not re-entrant, the predicate
reports $b$ and, by Theorem~\ref{thm:specialization}, the model certifies $(b, 0)$, contrary to assumption. The exact
condition of~\cite{nidhogg2026splitpoints} thus fails for $b$, and its necessity theorem constructs a completely
tokenizable input placing an occurrence of $b$ strictly inside a token. At that occurrence the token containing $b$
begins before $b$, so $(b, 0)$ is not certified.
\end{proof}

Non-vacuous conservatism is therefore a strictly multi-byte phenomenon. At length one the model is exact for occurring
bytes, and the shortest strict refusals have length two, a bound example one attains. The
$\{\code{a}, \code{abx}, \code{b}, \code{x}\}$ family of Section~\ref{sec:preliminaries} plays the opposite role in the
artifact, a negative control for the covering-origin property. Cutting at the family's fixed boundary is safe, since
\code{a} occurs only token-initially, yet the covering origin genuinely varies, and the artifact pins those violations
exactly; an oracle that misses them would accept false origins. Negatives in every table of this paper are claims about
the model, never about the language.

\section{What the cloud represents}
\label{sec:factorization}

The cloud was built as a conservative over-approximation, and Section~\ref{sec:strictness} exhibited the surplus; this
section says exactly what the surplus is. Let $X$ be the token language, the set of non-empty words $A$ accepts from
$q_0$, and let $\mathcal{G}$ be the set of completely tokenizable inputs, so $\mathcal{G} \subseteq X^*$. A
\emph{factorization} of a word $x \in X^*$ is a sequence of words of $X$ whose concatenation is $x$, with no greedy
choice imposed. The cloud is exactly the set of histories of all factorizations of all words of $X^*$ around the window.
This characterization explains why the vacuous window of Section~\ref{sec:preliminaries} is certified, locates a class
where the model is exact, finds a shortest witnessed window exactly, and decides whether the model misses any witnessed
certificate at all.

\subsection{Factorization semantics}
\label{sec:semantics}

Fix a window $W = w_0 \cdots w_{k-1}$, and take a word $x = u W v \in X^*$ and a factorization of $x$. Let $s$ be the
start of the factor covering the byte at $|u| + k - 1$, let $q = \delta^*(q_0,\, x[s \,..\, |u| + k))$ be that factor's
state after the window's last byte, and let $\omega$ be $\before$ when $s < |u|$, and $s - |u|$ otherwise. The set
$F(W)$ collects all pairs $(q, \omega)$ so obtained, over every $u$, $v$ and factorization.

\begin{theorem}[Factorization semantics]
\label{thm:factorization}
For every non-empty window, $C_k = F(W)$. Every pair of $C_k$ is realized by some $x = u W v$ with $|u| \leq n$ and
$|v| \leq n - 1$. In particular $C_k$ is non-empty exactly when $W$ is a factor of a word of $X^*$.
\end{theorem}

\begin{proof}
First $F(W) \subseteq C_k$. Fix $x = u W v$ and a factorization, and for $j \in 1..k$ let $s_j$ be the start of the
factor covering the byte at $|u| + j - 1$, with $\rho_j$ and $\omega_j$ defined from it as in
Section~\ref{sec:soundness}. The argument of Lemma~\ref{lem:representation} applies with ``factor'' in place of ``token
of the final segmentation'', since it used only that tokens are accepted words placed end to end. Each $\rho_j$ is live,
reached by a prefix of an accepted factor and completed by its rest. If $s_1 < |u|$, the state
$p = \delta^*(q_0, x[s_1 \,..\, |u|))$ is live and $(p, \before) \in C_0$, and the direct step gives
$(\rho_1, \before)$; the rename cannot fire, since $p$ has consumed a byte and $p = q_0$ only when $q_0$ is re-entrant.
If $s_1 = |u|$, the seed fires, because $C_0$ holds an accepting state, $X$ being non-empty, and it gives $(\rho_1, 0)$.
At a continuing byte the direct step carries the pair with its origin, and again the rename cannot fire. At a byte
beginning a factor the previous factor $x[s_j \,..\, |u| + j)$ is in $X$, so $\rho_j$ accepts and the seed gives
$(\rho_{j+1}, j)$.

Now $C_k \subseteq F(W)$, by induction on $j$. The claim is that every pair of $C_j$ is the pair of some
$x = u\, W[0 \,..\, j)\, v$ with a factorization, where $|u| \leq n$ and $v$ is a shortest accepting completion of the
covering factor's state, so that $|v| \leq n - 1$. For $j = 1$ the pairs of $C_1$ arise three ways. A direct step from
$(p, \before)$ with $p \neq q_0$ gives $(\delta(p, w_0), \before)$. Let $u$ be a shortest word with
$\delta^*(q_0, u) = p$, non-empty and of length at most $n - 1$, and let $v$ be a shortest accepting completion from
$\delta(p, w_0)$, which exists since that state is live; then $u w_0 v \in X$ is one factor, starting before the window.
A direct step from $(q_0, \before)$ that keeps $\before$ happens only when $q_0$ is re-entrant, and then a shortest
non-empty $u$ with $\delta^*(q_0, u) = q_0$ has length at most $n$, and the same completion gives one factor $u w_0 v$.
The fresh origin, from the seed or from the rename of $(q_0, \before)$, gives $(\delta(q_0, w_0), 0)$, realized with $u$
empty by the factor $w_0 v$ beginning at the window.

For the step, a pair of $C_{j+1}$ is either a direct successor $(\delta(q, b), \omega)$ of some $(q, \omega) \in C_j$,
or the fresh origin. In the first case take the realization $x = u\, W[0 \,..\, j)\, v$ of $(q, \omega)$, whose covering
factor reaches $q$ at the window's end, and replace its remainder $v$ by $b v'$, with $v'$ a shortest completion from
$\delta(q, b)$, which keeps every earlier factor and the origin. In the second case the seed fired, so some
$(q', \omega') \in C_j$ has $q'$ accepting. Take its realization, end its covering factor at the window's end, which is
legitimate since $q'$ accepts, and begin a new factor $b v'$, with $v'$ a shortest completion from $\delta(q_0, b)$,
whose origin is $j$. The rename does not arise at $j \geq 1$, since when $q_0$ is not re-entrant nothing enters it, and
when it is, the rename is off; the left context $u$ chosen at $j = 1$ is never changed. If $X$ is empty, no state is
live and $C_k$ is empty, and no word of $X^*$ contains a non-empty window, so both sides are empty.
\end{proof}

\begin{corollary}[What certification decides]
\label{cor:decides}
Let $W$ be a factor of a word of $X^*$. The model certifies $(W, o)$ exactly when, in every factorization of every word
of $X^*$ containing $W$, the factor covering the occurrence's last byte begins $o$ bytes into it.
\end{corollary}

The corollary explains the vacuous window of Section~\ref{sec:preliminaries}. Over $\{\code{0}, \code{00}, \code{01}\}$
the word \code{01001} factors as $\code{01} \mid \code{0} \mid \code{01}$, a factorization containing \code{1001} at
offset $1$, with the last factor beginning $2$ bytes into the window. Every factorization of a word containing
\code{1001} agrees, since a \code{1} is only ever the second byte of \code{01}, so the window's two \code{1}'s end
factors, and the factor covering the last one is the \code{01} at offset $2$. The model therefore certifies
$(\code{1001}, 2)$, and rightly so for $X^*$. The scan of \code{01001}, however, takes \code{01} and \code{00} and fails
on the last \code{1}, so the word is in $X^*$ and not in $\mathcal{G}$. Non-emptiness of the cloud is an occurrence
theorem for $X^*$, not for $\mathcal{G}$, and that is the whole gap between the model's certificates and the witnessed
ones.

\begin{corollary}[Token kinds are immaterial]
\label{cor:kinds}
Two automata with the same token language $X$ give, on every window, clouds with the same set of origins. Hence they
give the same certificates, the same origins and the same shortest certified length. Token kinds may be erased, the
positive-width equivalent taken, and the Boolean recognizer of $X$ minimized before any search of this paper.
Re-entrancy is recomputed on the result.
\end{corollary}

\begin{proof}
The set of origins of $F(W)$ depends only on $X$ and $W$, and certification, the origin and the shortest length are read
from that set.
\end{proof}

The reduction can be large. If every binary word of length $k$ is given its own token kind, the kind-preserving minimal
automaton keeps every prefix apart and has $2^{k+1} - 1$ live states, while the Boolean recognizer of
$\{\code{a}, \code{b}\}^k$ has $k + 1$, and both give the same certificates. The evaluation runs on the lexer's own
automata, as shipped, so the corollary offers a reduction these searches do not use.

\subsection{Where the model is exact}
\label{sec:exact}

By Corollary~\ref{cor:decides} the model quantifies over factorizations, where Definition~\ref{def:window} quantifies
over scans, and the two coincide whenever every factorization of a word of $X^*$ is the successful scan of that word.

\begin{corollary}[Exactness]
\label{cor:exact}
Suppose every word of $X^*$ has exactly one factorization and $\mathcal{G} = X^*$. Then for every window the origins of
the cloud are exactly the covering origins over the completely tokenizable occurrences of $W$, with $\before$ standing
for every origin outside the window. The model certifies every witnessed certificate, and refuses every occurring window
that is not one. Every prefix code satisfies the hypotheses, and so do $\{\code{a}, \code{ab}\}$ and $\code{ab}^*$,
which are not prefix codes.
\end{corollary}

\begin{proof}
Every word of $X^*$ is completely tokenizable, and its scan is a factorization, hence its only one. The factorizations
of words of $X^*$ containing $W$ are exactly the scans of the completely tokenizable inputs containing $W$, and the
origins of Theorem~\ref{thm:factorization} are the covering origins. An occurring window with the same in-window
covering origin at every occurrence is certified by the model, and any other occurring window is refused, since its
covering origins, $\before$ or two distinct ones, appear in $F(W)$. A prefix code admits at most one codeword at any
position, so a word of $X^*$ has one factorization, and the scan computes it by taking the one codeword that begins.
Over $\{\code{a}, \code{ab}\}$ every word of $X^*$ is a sequence of \code{a}'s, each followed by at most one \code{b};
its factors begin at the \code{a}'s, so the factorization is unique, and the scan computes it, taking \code{ab} where a
\code{b} follows and \code{a} otherwise. Over $\code{ab}^*$ the same holds with runs of \code{b}.
\end{proof}

Unique factorization alone is not enough; the scan must also succeed on all of $X^*$. Over
$\{\code{a}, \code{ab}, \code{bb}\}$ every word of $X^*$ has one factorization, since the reversed set
$\{\code{a}, \code{ba}, \code{bb}\}$ is a prefix code, so a word is factored uniquely from its right end. The window
\code{ab} is semantically certified at origin $0$ and witnessed by the input \code{ab}, because an \code{a} is only ever
a token's first byte and the scan prefers \code{ab} to \code{a} where a \code{b} follows. Yet \code{abb} factors as
$\code{a} \mid \code{bb}$, covering the window's last byte with a factor beginning at offset $1$, so the model refuses
the window. The scan of \code{abb} takes \code{ab} and fails on the final \code{b}, so
$\code{abb} \in X^* \setminus \mathcal{G}$, and that is the hypothesis that fails. The paper's example one,
$\{\code{a}, \code{ab}, \code{b}\}$, fails the other hypothesis, since \code{ab} has two factorizations.

The hypotheses of Corollary~\ref{cor:exact} are decidable in polynomial time. Write $\mathrm{Pref}_+(X^*)$ for the
non-empty prefixes of words of $X^*$.

\begin{theorem}[The exact class]
\label{thm:agreement}
The following are equivalent for a token language $X$.
\begin{enumerate}
\item Every factorization of every word of $X^*$ is the scan of that word, which succeeds.
\item Every word of $X^*$ has exactly one factorization and $\mathcal{G} = X^*$.
\item $X \cap X\, \mathrm{Pref}_+(X^*) = \varnothing$: no word of $X$ is a word of $X$ followed by a non-empty prefix of
a word of $X^*$.
\end{enumerate}
Condition (3) is decided on the live automaton in $O(256 \cdot n^2)$ time. When it fails, the decision returns
$u \in X$, a non-empty $r$ with $u r \in X$, and $z$ with $r z \in X^*$, of total length at most $n^2 + 2n - 2$. The
word $u r z$ is then in $X^*$, and its factorization beginning with $u$ is not its scan.
\end{theorem}

\begin{proof}
First, (1) implies (2). A word of $X^*$ has a factorization, so its scan succeeds, and every factorization is that scan,
so there is one. Next, (2) implies (1). The scan of a word of $X^*$ succeeds and is a factorization, hence the only one.
Suppose (3) fails, so that $u \in X$, $r$ is non-empty, $u r \in X$ and $r z \in X^*$. Then $u r z \in X^*$ has a
factorization beginning with $u$, while its scan, if it succeeds, begins with a token of length at least $|u r| > |u|$;
either the scan fails or it is a different factorization, so (1) fails. Suppose (3) holds and take any factorization
$x = x_1 \cdots x_m$. If $m = 0$ the scan of the empty word succeeds. Otherwise $x_1 \in X$ is a candidate first token,
and a longer accepted prefix $x_1 r$ of $x$ would have $r$ a non-empty prefix of $x_2 \cdots x_m \in X^*$, against (3).
The scan thus commits $x_1$ and restarts, and by induction on $m$ it commits $x_2, \ldots, x_m$ and succeeds.

For the decision, $\mathrm{Pref}_+(X^*)$ is recognized by a nondeterministic automaton $P$ over the live states, which
starts at $q_0$, follows $\delta$, optionally resets to $q_0$ before any byte read from an accepting state, and accepts
at every state reached after at least one byte. A path of $P$ spells a sequence of accepted factors followed by a prefix
of a live run, which completes to a factor, so the word is in $\mathrm{Pref}_+(X^*)$. Conversely, a non-empty prefix of
$x_1 \cdots x_m$ is spelled by resetting at each factor's end. Now run one breadth-first search over pairs $(p, p')$,
started at every pair $(q, q_0)$ with $q$ a reachable accepting state, where $p$ continues $\delta$ from $q$ and $p'$
runs $P$ from $q_0$, both over the same bytes, and a pair reached after at least one byte is marked as started. A
started pair with $p$ accepting is an obstruction. Let $u$ be a shortest word reaching the $q$ its search path began at,
of length at most $n - 1$ since $q_0$ does not accept, let $r$ be the bytes read, and let $z$ be a shortest accepting
completion from $p'$, of length at most $n - 1$; then $u \in X$, $u r \in X$ and $r z \in X^*$. Conversely an
obstruction $u$, $r$ drives the search from $\delta^*(q_0, u)$ along $r$ to such a pair. There are at most $n^2$ started
pairs, and a shortest path to one repeats none, so $|r| \leq n^2$, and the search visits each pair's successors, at most
two per byte, once.
\end{proof}

The figure program runs the decision on the examples above. It places $\{\code{a}, \code{ab}\}$ and $\code{ab}^*$ in the
class, and $\{\code{a}, \code{ab}, \code{bb}\}$ outside it with the obstruction \code{abb},
$\{\code{a}, \code{ab}, \code{b}\}$ with \code{ab}, and $\{\code{0}, \code{00}, \code{01}\}$ with \code{00}, whose
factorization $\code{0} \mid \code{0}$ the scan does not take. The conventional C-like row is out with \code{//}
followed by a tab, a comment and a prefix of a whitespace run, which the scan reads as one comment. The class is a
sufficient condition for the model to be exact, not a necessary one, since exactness compares unanimous origins, a
weaker observable than equality of every history.

\subsection{A shortest witnessed window}
\label{sec:witnessed}

The evaluation separates model certificates from witnessed ones by a bounded search for an occurrence, and a certificate
that search cannot witness is reported as unresolved. The separation can be decided exactly. A shortest witnessed window
is found, with its occurrence, by running the chosen-origin automaton of Section~\ref{sec:chosen} in step with a
recognizer of the factors of $\mathcal{G}$. That recognizer uses a recognizer of $\mathcal{G}$ as a black box, the byte
projection of the splitting companion's armed-run automaton~\cite{nidhogg2026splitting} over one compiled token
automaton. An input is completely tokenizable exactly when it segments into accepted words such that no committed word
extends, within the input, to a longer accepted word from the same start. The automaton guesses the boundaries, runs
each committed word's state onward as an armed run, and rejects if an armed run accepts again. The automaton's states
are a current token state or a fresh boundary, together with the set of armed states, at most $(n + 1)\, 2^n$ of them,
and it recognizes exactly $\mathcal{G}$.

Let $Z$ be that automaton, trimmed to the states reachable from its start and co-accessible to its accepting states, and
let $\mathrm{Fact}(Z)$ be $Z$ with every retained state made initial and accepting.

\begin{lemma}[Factors]
\label{lem:factors}
$\mathrm{Fact}(Z)$ accepts exactly the factors of completely tokenizable inputs.
\end{lemma}

\begin{proof}
Take a path of $Z$ between retained states. The trim gives a stem from the start to its first state and a suffix from
its last state to an accepting state, and the three together are an accepting run of $Z$, so the path's label is a
factor of a word of $\mathcal{G}$. Conversely a factor of a word of $\mathcal{G}$ is the label of the middle of an
accepting run of $Z$, whose states are all retained.
\end{proof}

\begin{theorem}[Witnessed synthesis]
\label{thm:witnessed}
The product of the chosen-origin automaton with $\mathrm{Fact}(Z)$ accepts exactly the witnessed model certificates.
Breadth-first search over it decides whether one exists, returning a shortest one with its origin and a completely
tokenizable input containing it, and its exhaustion proves that no witnessed model certificate exists at any length. The
product has at most $K (n + 1)\, 2^n = O(n^2 4^n)$ states. Existence of a witnessed model certificate is in PSPACE.
\end{theorem}

\begin{proof}
The product accepts the intersection of the two languages, the model-certified windows by Theorem~\ref{thm:chosen} and
the factors of $\mathcal{G}$ by Lemma~\ref{lem:factors}. That intersection is the set of witnessed model certificates,
by Definition~\ref{def:witnessed}. Every edge reads one byte, so breadth-first search from the initial pairs finds a
shortest accepted window, and the chosen-origin component gives its origin. For the occurrence, take the $Z$ states at
the ends of the accepting path and attach the stem $u$ and the suffix $v$ the trim guarantees; then $u W v$ is accepted
by $Z$, so it is completely tokenizable, and it contains $W$ at offset $|u|$. By Theorem~\ref{thm:soundness} the token
covering that occurrence's last byte begins at the certified origin. For the space bound, a product state is written in
$O(n)$ bits, a path of length below the product's size is guessed with a counter of $O(n)$ bits, and the initial and
final $Z$ states are checked reachable and co-accessible by implicit searches in $Z$ within the same space, so the
problem is in nondeterministic polynomial space.
\end{proof}

Exhaustion here means no witnessed \emph{model} certificate, and it still says nothing about a semantic certificate the
model refuses. The returned input is a witness of occurrence, not a shortest such input, since the search minimizes the
window. The four model-positive grammars of the random sweep that the bounded witness search leaves unresolved are a
question this search settles either way, though this paper does not rerun the sweep under it. With
Lemma~\ref{lem:reduction-forward}, whose window is witnessed, the witnessed half of Theorem~\ref{thm:pspace} follows.

\subsection{Whether the model misses a certificate}
\label{sec:completeness}

Section~\ref{sec:strictness} shows the model refusing a witnessed certificate over two small token sets, and whether a
given token set has any such refusal is decidable from the regularity noted in Section~\ref{sec:related}. This section
decides it explicitly, with a shortest missed window and a bound. Call the model \emph{globally complete} for a token
set when it certifies every witnessed certificate of Definition~\ref{def:witnessed}, which holds vacuously when there is
none.

Two marked graphs carry the decision. In both, an edge is labeled by a byte and marked when a token begins just before
that byte, so the last marked edge of a path locates the start of the token covering the last byte it spells. The first
is the trimmed $Z$, whose paths are the greedy histories. The second, $\Phi$, has the live states and a fresh boundary
$\bot$. It has an unmarked edge $p \xrightarrow{b} \delta(p, b)$ for every live transition and, for every byte $b$ with
$\delta(q_0, b)$ live, a marked edge $p \xrightarrow{b} \delta(q_0, b)$ from every accepting $p$ and from $\bot$. Its
initial state is $\bot$, its final states are the accepting live states, and it is trimmed likewise. Its accepting paths
are exactly the factorizations of non-empty words of $X^*$, and by the argument of Lemma~\ref{lem:factors} its paths
between retained states are exactly the factors of words of $X^*$, with the factor starts they contain marked. Write $N$
and $f$ for the two graphs' state counts, with $N \leq (n + 1)\, 2^n$ and $f \leq n + 1$. For a graph $J$ and a
candidate $(W, o)$, a path spelling $W$ \emph{has origin $o$} when its edge at byte $o$ is marked and no later edge is,
so a path with no marked edge has none. Let $\mathrm{Occ}_J$ be the set of candidates some path of $J$ spells, and
$\mathrm{Bad}_J$ the set of candidates $(W, o)$ that some path of $J$ spells without having origin $o$.

\begin{theorem}[Global completeness]
\label{thm:complete}
The model misses $(W, o)$, a witnessed certificate it refuses, exactly when
$(W, o) \in \mathrm{Occ}_Z \cap \mathrm{Bad}_\Phi \setminus \mathrm{Bad}_Z$. Whether the model is globally complete is
decidable. When it is not, a shortest missed window is returned with its origin, a completely tokenizable input
containing it, and a word of $X^*$ with a factorization that covers the occurrence's last byte from another origin. The
decision runs over an observer of at most $2^{N + f} + 4^{N + f}$ states. So it is in EXPSPACE in $n$, and takes doubly
exponential time.
\end{theorem}

\begin{proof}
A witnessed certificate $(W, o)$ is a candidate in $\mathrm{Occ}_Z$, since $W$ occurs in some completely tokenizable
input, and it is outside $\mathrm{Bad}_Z$, since every greedy occurrence has origin $o$. The model refuses it exactly
when $C_k$ carries an origin other than $o$, because the pair of the greedy occurrence is in $C_k$ by
Lemma~\ref{lem:representation} and carries $o$, so the cloud is non-empty and certification fails only by a second
origin. By Theorem~\ref{thm:factorization} a second origin is a factorization of a word of $X^*$ covering the
occurrence's last byte from elsewhere, which is membership in $\mathrm{Bad}_\Phi$, so the three conditions are exactly
the missed certificates.

The observer reads $W$ with its byte $o$ marked, and for each graph $J$ it holds the set of ends of the paths spelling
the prefix read so far. Before the mark it holds one set per graph, initially every retained state; after the mark that
set is split into the \emph{right} set $R_J$ of ends whose path has origin $o$ so far and the \emph{wrong} set $B_J$ of
the rest. At the marked byte, an edge from the set goes to $R_J$ when marked and to $B_J$ otherwise. After it, $R_J'$ is
the set of ends of unmarked edges from $R_J$, and $B_J'$ the set of ends of all edges from $B_J$, together with the ends
of marked edges from $R_J$. The two sets may overlap, since different paths reaching one state may differ in origin, and
they are held apart. At the window's end, $(W, o)$ is missed exactly when $R_Z \neq \varnothing$, $B_Z = \varnothing$
and $B_\Phi \neq \varnothing$, which restates the three conditions. Given the position of the mark the observer is
deterministic, and the mark is one extra choice per window, so breadth-first search over observer states, marking at
each depth or not, reaches a missed window of minimum length or exhausts the finitely many states. From a missed window,
a path through $Z$ into $R_Z$ and a path through $\Phi$ into $B_\Phi$ are read back, and stems and suffixes complete
them into the two inputs claimed. A state of the observer is one or two subsets of each graph, hence the bound, and it
is written in $O(N + f)$ bits, which is $O(n 2^n)$; a nondeterministic search with a counter runs in that space, and the
explicit search runs in time polynomial in the observer's size.
\end{proof}

Over $\{\code{a}, \code{ab}, \code{b}\}$ the shortest missed window is \code{ab} at origin $0$, with the completely
tokenizable input \code{ab} and the factorization $\code{a} \mid \code{b}$ of the same word. Over
$\{\code{a}, \code{ab}, \code{bb}\}$ it is the same window and input, with the factorization $\code{a} \mid \code{bb}$
of \code{abb}. The model is globally complete over $\{\code{a}\}$ and over $\{\code{a}, \code{ab}\}$, the latter by
Corollary~\ref{cor:exact}. The bound is an upper bound only, and whether a polynomial criterion for global completeness
exists is open; Theorem~\ref{thm:agreement} decides a sufficient condition in polynomial time. Completeness is a
property of the token set, settled once from its tables, and where it holds every negative of the model is a negative of
Definition~\ref{def:witnessed}, so the three-layer reading of the results collapses to two.

\section{From a boundary to a parallel cut}
\label{sec:cut}

At every occurrence at offset $t$ in a completely tokenizable input, a certified window $(W, o)$ yields a true boundary
$t + o$ of the final segmentation. Turning a boundary into a parallel cut is the predecessor's territory, since its
prefix-stability result licenses a worker to scan from a known boundary and agree with the serial stream around the
cut~\cite{nidhogg2026splitpoints}. The division of labor is exact, in that this paper establishes that $t + o$ is a
boundary and the predecessor establishes what a scan starting at a boundary preserves. The v1.3.3 artifact evaluated
here plans with single-byte certificates only. Release 1.4.0 added the window decision as a public contract, which the
sweep at the later commit cross-checks against its model, and a window-planning sibling, which lies outside this paper's
evaluation and is not a claim of this paper.

\section{Boundary profiles}
\label{sec:profile}

Section~\ref{sec:preliminaries} kept three guarantees apart, and this paper certifies the strongest. This section asks,
of every gap of a window at once, whether the weakest guarantee, a fixed safe boundary, holds there. Release 2.1.0 of
munch decides that property exactly within its search cap, and the answer places the certificate among the anchors a
window can soundly supply, the positions it makes a boundary at every occurrence.

\begin{definition}[Boundary profile]
\label{def:profile}
Let $W = w_0 \cdots w_{k-1}$ with $k \geq 1$. A \emph{gap} of $W$ is an index $g \in \{0, \ldots, k\}$: the gap before
byte $w_g$ when $g < k$, and the gap after the final byte when $g = k$. Take an occurrence of $W$ at offset $t$ in a
completely tokenizable input $x$. The gap $g$ \emph{is a boundary} there when a token of the maximal-munch tokenization
of $x$ begins at $t + g$ or $t + g = |x|$, and a token \emph{crosses} it otherwise. Over every occurrence of $W$ in
every completely tokenizable input, the gap is \emph{forced} when it is a boundary at each of them, \emph{crossed} when
it is a boundary at none, and \emph{open} when both happen. The \emph{boundary profile} of $W$ assigns every gap its
verdict.
\end{definition}

A gap here is a position between adjacent bytes or at either end of the window, not a distance between anchors. A window
no completely tokenizable input contains is forced and crossed at every gap, both vacuously, exactly as
Definition~\ref{def:window} holds vacuously there, while for an occurring window each gap has one verdict. The library
names the three verdicts \code{must}, \code{never} and \code{may}, and a window occurring nowhere \code{absent}.

\begin{remark}
\label{rem:profile}
Three earlier notions are profile patterns by definition. The pair $(W, o)$ is a certified split window exactly when gap
$o$ is forced and gaps $o + 1, \ldots, k - 1$ are crossed. At an occurrence, the token covering byte $t + k - 1$ begins
at $t + o$ exactly when a token begins at $t + o$ and none begins at $t + o + 1, \ldots, t + k - 1$, and both conditions
are quantified over the same occurrences, so the conjunction splits gap by gap. A certified split point of the
predecessor~\cite{nidhogg2026splitpoints} is the case $k = 1$ with gap $0$ forced, and the fixed safe boundary of
Section~\ref{sec:preliminaries} is a forced gap.
\end{remark}

The profile names anchors the certificate cannot. Over $\{\code{0}, \code{1}, \code{x}, \code{001x}, \code{011x}\}$ the
window \code{0011} has gaps $0$ and $1$ forced, since \code{0} lies inside a token only as the second byte of
\code{001x}, and neither of the window's first two bytes is followed by \code{1x}. Gaps $2$ to $4$ are open, crossed in
the input \code{0011x}, which tokenizes as \code{0} and \code{011x}, and boundaries in \code{0011}, which tokenizes byte
by byte. The covering token's origin is therefore $1$ in one input and $3$ in the other, while the profile gives the
window two anchors. No substring of the window is a certified split window at any origin, since the last byte of each
begins a token in \code{0011} and lies inside an earlier token in \code{0011x}, or in \code{001x} when it is a \code{0}.

\paragraph{Local maximality.} The weakest guarantee is the one a cut consumes. Fix a width $H \geq 1$. A
\emph{cut rule of width $H$} is a set $R$ of pairs $(W, g)$ with $1 \leq |W| \leq H$ and $g$ a gap of $W$, and on an
input $x$ the rule cuts at every position $t + g$ such that $(W, g) \in R$ and $W$ occurs at $t$. Each piece between
consecutive cuts is scanned from $q_0$ on its own, and the token streams are concatenated with nothing reconciled. The
rule is \emph{sound} when on every completely tokenizable input the concatenation is the input's maximal-munch
tokenization. It reads the window and nothing else. No state crosses a cut, and no piece is rescanned.

\begin{proposition}[Local maximality]
\label{prop:maximal}
A cut rule is sound if and only if every gap it cuts at is forced.
\end{proposition}

\begin{proof}
Suppose every $(W, g) \in R$ has gap $g$ forced, and fix a completely tokenizable $x$. Every cut falls at a forced gap
of an occurrence, so at a boundary of the tokenization of $x$ or at $|x|$, and a cut at a boundary is exact, the pieces
between boundaries, scanned from $q_0$ and concatenated, reproducing the tokenization. That is the splitting companion's
lemma that cuts at boundaries are exact~\cite{nidhogg2026splitting}, stated there for literal token sets and carried to
regular ones by its corollary on splitting at the regular generality, and Section~\ref{sec:cut} composes with the same
step. Conversely, suppose some $(W, g) \in R$ has gap $g$ not forced. Some completely tokenizable $x$ then contains $W$
at an offset $t$ where a token of the tokenization crosses $t + g$, beginning before it and ending after it, and the
rule cuts $x$ at $t + g$. The piece ending there emits no token reaching past $t + g$, and the piece beginning there
emits none beginning before it, so the crossing token is emitted by neither and the concatenation differs from the
tokenization.
\end{proof}

A certificate at origin $o$ is the case in which gap $o$ is forced and the gaps $o + 1, \ldots, k - 1$ are crossed, the
earlier gaps and gap $k$ being left free, and it supplies the anchor $t + o$. The proposition says what else a window
supplies. A sound rule of width $H$ cuts at no position that the forced gaps of windows of width at most $H$ fail to
name, and it may cut at every position they name. The profile at width $H$ is the whole anchor supply of that width, of
which a certificate names one member, and an occurring window has at most one certified origin. The gap $k$ after the
window is the member no origin can name, since a forced gap $k$ of $W$ is the certificate $(Wc, k)$ for every byte $c$
that can follow, one byte later.

\paragraph{Decidability.} Gap $g$ is not forced exactly when some completely tokenizable input has an occurrence of $W$
that a token crosses at $g$, and not crossed exactly when some has an occurrence cut there. Each is a question of the
form the splitting companion's offline verifiability theorem decides~\cite{nidhogg2026splitting}, since the boundary
markings of completely tokenizable inputs form a regular language, recognized by its armed-run automaton, and its
intersection with the markings holding an occurrence of $W$ with, or without, a mark at gap $g$, the end-of-input
closure counting as one, has decidable emptiness with a shortest counterexample. The profile is therefore decidable for
every regular token set, and Theorems~\ref{thm:gap-pspace} and~\ref{thm:fixed-nl} below give the complexity over a
compiled one. Release 2.1.0 of munch decides both questions per gap, as \code{boundary\_counterexample()} and
\code{crossing\_counterexample()}, which both run the one search of its exact window decision
\code{window\_counterexample()}, building the input byte by byte, breadth first, with its token boundaries guessed and
every guess held to maximal munch. Each returns the shortest refuting input or, once it has exhausted its states under a
cap on the states held at once, the proof that none exists; a search the cap stops settles nothing, and the gap is
reported undetermined. The function \code{boundary\_profile()} asks both questions at every gap and reads the verdicts
off, giving \code{must} when only the first search exhausts without a counterexample, \code{never} when only the second
does, and \code{may} when both find one. Whether the window occurs at all is decided once per window by
\code{window\_occurrence()} under the same cap, and \code{absent} is given at every gap or at none.

The two existential questions behind a verdict differ in difficulty. Whether a gap is a boundary at some occurrence is
PSPACE-complete, as is the verdict open, by the next theorem, while whether a token crosses it at some occurrence is
NL-complete, as are the verdict forced and the certificate of a given pair, by Theorem~\ref{thm:fixed-nl}. Both theorems
measure the size of the compiled automaton and the window. The upper bounds of the first run the armed-run automaton of
Section~\ref{sec:witnessed} with the marks of Section~\ref{sec:completeness}, where an edge is marked when a token
begins just before its byte.

\begin{theorem}[Gap verdicts]
\label{thm:gap-pspace}
Given a compiled token set, a window $W$ and a gap $g$ of $W$, deciding whether $g$ is a boundary at some occurrence of
$W$ in a completely tokenizable input is PSPACE-complete, and so is deciding its complement, the verdict crossed of
Definition~\ref{def:profile}. Deciding the verdict open is PSPACE-complete as well. All three remain PSPACE-hard when
the token set has one token kind over five symbols and a start state that neither accepts nor is re-entered, and the
window occurs in some completely tokenizable input, so that the verdict is not a vacuous one.
\end{theorem}

\begin{proof}
\emph{Membership.} The upper bound runs the armed-run automaton of Section~\ref{sec:witnessed} as it stands before the
trim that gives $Z$. The trim removes no accepting run, and a successor in the untrimmed automaton is computed from the
tables of the token automaton alone. An accepting run on an input $x$ guesses a segmentation of $x$ into accepted words,
none of which extends within $x$ to a longer accepted word from the same start. Every word of such a segmentation is the
longest accepted word that begins where it does, so by induction from the left the segmentation is the maximal-munch
tokenization of $x$, and the marked edges of the run read the first bytes of its tokens. A nondeterministic machine
guesses an input byte by byte along a run from the start, holding the run's current state, a token state or the fresh
boundary with a set of armed states, in $O(n)$ bits for a token automaton of $n$ states. It holds a counter as well. The
value $0$ says that no occurrence has been chosen, a value $c$ with $1 \leq c < k$ that the last $c$ bytes guessed are
the first $c$ bytes of the chosen occurrence, and the value $k$, which the counter then keeps, that the occurrence has
been read. At a byte of its choice the machine moves the counter from $0$ to $1$, checking that byte against $w_0$, and
it checks each following byte against $W$ until the counter is $k$. It keeps one bit for the gap. For $g < k$ the bit is
set when the edge reading $w_g$ is marked and clear when it is not. For $g = k$ the bit waits for the edge after the
window, which sets it when marked and clears it otherwise, or for the end of the input, which sets it. The machine
guesses the rest of the input and accepts when the counter is $k$, the run ends in an accepting state and the bit has
the value asked for. Asked for a set bit, it accepts if and only if $g$ is a boundary at some occurrence of $W$ in a
completely tokenizable input, and asked for a clear bit if and only if a token crosses $g$ at some occurrence. A
configuration is a state of the automaton, a counter value from $0$ to $k$ and a bit that is unset, set or clear, so
there are at most $3 (k + 1)(n + 1)\, 2^n$ of them. An accepting computation need not repeat one, and a step counter of
$O(n + \log k)$ bits bounds it. Both questions are therefore in NPSPACE, which is PSPACE. Since PSPACE is closed under
complement and intersection, it holds the verdict crossed, the complement of the first question, and the verdict open,
the conjunction of the two.

\emph{Hardness.} The reduction is from intersection non-emptiness~\cite[Lemma~3.2.3, p.~261]{kozen1977lower}, the source
problem of Section~\ref{sec:complexity}, with the two letters of the sources written \code{0} and \code{1}. Let
$A_i = (Q_i, \{\code{0}, \code{1}\}, \delta_i, s_i, F_i)$ for $i = 1, \ldots, m$, with $m \geq 1$, be complete DFAs. The
token DFA has the symbols $\code{0}, \code{1}, \code{x}, \code{\#}, \code{\symbol{36}}$ and one token kind. Its states
are the root $q_0$, level states $V_1, \ldots, V_{m+1}$, a state $U$, an end state $D$, and a disjoint copy $(i, q)$ of
every state $q$ of every $A_i$. The root has the one transition $q_0 \xrightarrow{\code{x}} V_1$, and the level states
continue it by $V_i \xrightarrow{\code{x}} V_{i+1}$ for $i = 1, \ldots, m$. The first level state has
$V_1 \xrightarrow{\code{\#}} U$, where $U$ keeps itself on \code{0} and \code{1} and has
$U \xrightarrow{\code{\symbol{36}}} D$. For $i = 1, \ldots, m$ the level state $V_{i+1}$ has
$V_{i+1} \xrightarrow{\code{\#}} (i, s_i)$, and a copy $(i, q)$ has $(i, q) \xrightarrow{b} (i, \delta_i(q, b))$ for
both letters $b$ and, when $q \notin F_i$ and only then, $(i, q) \xrightarrow{\code{\symbol{36}}} D$. There are no other
transitions, and the accepting states are $V_1$ and $D$. The root does not accept and nothing enters it, the automaton
has $\sum_i |Q_i| + m + 4$ states, and it is built in polynomial time.

The tokens are the single byte \code{x}, the words $\code{x}\, \code{\#}\, y\, \code{\symbol{36}}$ for every
$y \in \{\code{0}, \code{1}\}^*$, and for $d = 2, \ldots, m + 1$ the words
$\code{x}^{d}\, \code{\#}\, y\, \code{\symbol{36}}$ with $y \notin L(A_{d-1})$. A token of the last two forms has
\emph{level} $d$ when it begins with $d$ bytes \code{x}; the tokens of level $1$ are the \emph{bottom tokens}, and those
of level $d \geq 2$ are the \emph{testers} of $A_{d-1}$. Every token begins with \code{x} and holds \code{x} only in its
leading run, of at most $m + 1$ bytes, and a token that holds \code{\#} has a level, with the \code{\#} directly after
the leading run, a word of $\{\code{0}, \code{1}\}^*$ after that, and \code{\symbol{36}} as its last byte. For a word
$y \in \{\code{0}, \code{1}\}^*$ let $\lambda(y)$ be the largest $d \in \{2, \ldots, m + 1\}$ with
$y \notin L(A_{d-1})$, and $1$ when every $A_i$ accepts $y$. The window is $W = \code{x}^{m+1}\, \code{\#}$, of
$k = m + 2$ bytes, and the gap is $g = m$, the gap before the last \code{x}. Three observations give the reduction.

First, at every occurrence a token begins at the window's first byte, and a word $y\, \code{\symbol{36}}$ follows the
window. Let $z$ be completely tokenizable with $W$ at offset $t$. The token holding the \code{\#} at $t + m + 1$ has a
level, so the bytes of $z$ after the occurrence begin with $y\, \code{\symbol{36}}$ for some
$y \in \{\code{0}, \code{1}\}^*$. The token holding the \code{x} at $t$ is the single byte, which begins at $t$, or it
has that \code{x} in its leading run. In the second case the run ends where the bytes \code{x} of $z$ from $t$ on end,
at $t + m$, and it has at most $m + 1$ bytes, so it begins at $t$ or later, hence at $t$.

Second, the scan descends the levels. Let the scan of some input stand at a boundary followed by
$\code{x}^{d}\, \code{\#}\, y\, \code{\symbol{36}}$ with $1 \leq d \leq m + 1$. From the root the bytes $\code{x}^{d}$
lead to $V_d$. If $d = 1$, the \code{\#} leads to $U$, the word $y$ keeps $U$ and \code{\symbol{36}} leads to $D$, which
no transition leaves, so the scan commits the bottom token. If $d \geq 2$, the \code{\#} leads to $(d - 1, s_{d-1})$ and
$y$ to $(d - 1, \delta_{d-1}(s_{d-1}, y))$, a state with a transition on \code{\symbol{36}} exactly when
$y \notin L(A_{d-1})$. With that transition the run ends in $D$ and the scan commits the tester
$\code{x}^{d}\, \code{\#}\, y\, \code{\symbol{36}}$. Without it the run stops at the \code{\symbol{36}}, having passed
no accepting state but $V_1$, since $D$ is entered on \code{\symbol{36}} alone, so the scan commits the single \code{x}
and stands at a boundary followed by $\code{x}^{d-1}\, \code{\#}\, y\, \code{\symbol{36}}$. By induction on $d$, a scan
standing at a boundary followed by $\code{x}^{m+1}\, \code{\#}\, y\, \code{\symbol{36}}$ commits $m + 1 - \lambda(y)$
single bytes \code{x} and then the token of level $\lambda(y)$, because $A_{d-1}$ accepts $y$ for every
$d > \lambda(y)$, and $A_{\lambda(y) - 1}$ rejects it when $\lambda(y) \geq 2$. In particular
$\code{x}^{m+1}\, \code{\#}\, y\, \code{\symbol{36}}$ is completely tokenizable for every $y$, and
$\code{x}^{m+1}\, \code{\#}\, \code{\symbol{36}}$ is a completely tokenizable input containing $W$ whatever the sources
are.

Third, gap $m$ is a boundary at an occurrence exactly when every source accepts the word that follows the window there.
At an occurrence at offset $t$ in a completely tokenizable $z$, the first observation puts the scan of $z$ at the
boundary $t$ with $\code{x}^{m+1}\, \code{\#}\, y\, \code{\symbol{36}}$ ahead of it, and by the second the token holding
the \code{\#} begins at $t + m + 1 - \lambda(y)$. If $\lambda(y) = 1$ that token begins at $t + m$, and gap $m$ is a
boundary. If $\lambda(y) \geq 2$ the token begins before $t + m$ and ends after it, and it crosses gap $m$. Now
$\lambda(y) = 1$ says that every $A_i$ accepts $y$. Hence gap $m$ is a boundary at some occurrence of $W$ in a
completely tokenizable input if and only if some word is accepted by every $A_i$, the input
$\code{x}^{m+1}\, \code{\#}\, y\, \code{\symbol{36}}$ serving for such a word $y$.

The construction is polynomial, so deciding whether a gap is a boundary at some occurrence is PSPACE-hard, and since
PSPACE is closed under complement, so is deciding the verdict crossed.

The verdict open takes one more step. By the third observation a token crosses gap $m$ at some occurrence exactly when
some word $y$ has $\lambda(y) \geq 2$, the input $\code{x}^{m+1}\, \code{\#}\, y\, \code{\symbol{36}}$ serving for such
a word, and so exactly when some $A_i$ rejects some word. Gap $m$ is therefore open exactly when some word is accepted
by every $A_i$ and some $A_i$ rejects some word. Whether every $A_i$ accepts every word is decided in polynomial time,
by a search of each source for a reachable rejecting state, and on such instances intersection non-emptiness has the
answer yes. The reduction maps them to the token set, window and gap built from one fixed family of sources with a
common word and a source rejecting some word, such as the single source accepting only the empty word, where that gap is
open, and every other instance to its own construction above, where gap $m$ is open exactly when some word is accepted
by every $A_i$. It stays polynomial, so deciding the verdict open is PSPACE-hard as well. For each of the three
questions the token set has one token kind over five symbols, its start state neither accepts nor is re-entered, and its
window occurs whatever the instance is.
\end{proof}

This reduction establishes hardness for boundary existence, but not for crossing existence. At an occurrence followed by
$y\, \code{\symbol{36}}$ the token that holds the \code{\#} begins at gap $m + 1 - \lambda(y)$ of the window by the
second observation, and a token of that kind beginning at gap $s$ crosses exactly the gaps $s + 1, \ldots, m + 2$, while
the single bytes before it cross none. For $1 \leq j \leq m$ a token therefore crosses gap $j$ at some occurrence
exactly when $\lambda(y) \geq m + 2 - j$ for some word $y$, which holds exactly when one of the sources
$A_{m+1-j}, \ldots, A_m$ rejects some word and is decided in polynomial time. Gap $0$ is forced and gaps $m + 1$ and
$m + 2$ are crossed whatever the sources are, so over these token sets the forced verdict of every gap of the window is
easy. Theorem~\ref{thm:fixed-nl} below shows that it is easy over every token set.

The reduction also fixes the length of a shortest refutation, as Corollary~\ref{cor:reduction-length} fixes that of a
shortest window.

\begin{corollary}[Long refutations]
\label{cor:gap-long}
In the token set of the proof above, suppose some word is accepted by every $A_i$, and let $\ell$ be the length of a
shortest one. A shortest completely tokenizable input with an occurrence of $W$ at which gap $m$ is a boundary then has
length $\ell + m + 3$. For the cycle sources of Section~\ref{sec:long}, their letters written \code{0} and \code{1}, the
token automaton has $(m^2 + 3m + 8)/2$ states, all but one of them live when $m \geq 2$, and the window $m + 2$ bytes,
and that length is $R_m = L_m + m + 2 \geq 2^m / (m + 1) + m + 2$, which is $2^{\Omega(\sqrt n)}$ over the $n$ live
states of the automaton. Other sources give more. For infinitely many $N$ there is a token set with one token kind over
five symbols whose automaton has $N$ states, all live, with a window of at most $N$ bytes and a gap of it, such that a
shortest input refuting the crossed verdict of that gap has $2^{\Omega(N)}$ bytes. No larger order is possible for such
windows, since a gap of a window of $k$ bytes that is a boundary at some occurrence, over a token automaton of $n$
states, is a boundary at an occurrence in an input of fewer than $3 (k + 1)(n + 1)\, 2^n$ bytes.
\end{corollary}

\begin{proof}
By the first and third observations an input of that kind contains $\code{x}^{m+1}\, \code{\#}\, y\, \code{\symbol{36}}$
for a word $y$ every $A_i$ accepts, so its length is at least $\ell + m + 3$, and by the second observation
$\code{x}^{m+1}\, \code{\#}\, y\, \code{\symbol{36}}$ with $y$ a shortest such word is one. The cycle sources have
$\sum_{i=1}^{m} i$ states together, which gives the count, and a shortest word they all accept has length $L_m - 1$ by
Section~\ref{sec:long}, so $\ell + m + 3 = L_m + m + 2$. The bound $L_m \geq 2^m / (m + 1)$ is from the proof of
Theorem~\ref{thm:long}, and the live states are at most the states, whose number is quadratic in $m$. The copy of $A_1$,
a single state that accepts every word and so has no transition on \code{\symbol{36}}, is not live, and for $m \geq 2$
every other state is, since every other source has a rejecting state, which each of its states reaches. Gap $m$ is a
boundary at an occurrence in that input, which therefore refutes the verdict crossed, and no shorter input does.

For the larger order, Vorel~\cite[Theorem~8, p.~378]{vorel2014subset} gives, for infinitely many sizes $|Q|$, a complete
deterministic automaton $B = (Q, \{\code{0}, \code{1}\}, \delta)$ that is strongly connected, and a set $S \subseteq Q$
that some word sends to a single state, such that a shortest word doing so has length $2^{\Omega(|Q|)}$. Let $z$ be a
shortest one, let $\delta(S, z) = \{r\}$ and $S = \{s_1, \ldots, s_m\}$, and take the sources
$A_i = (Q, \{\code{0}, \code{1}\}, \delta, s_i, \{r\})$. Every source accepts $z$, and a word every source accepts sends
$S$ to $r$ and is no shorter than $z$, so $\ell = |z|$. The sources have the same transitions and the same accepting
state, so the token automaton can hold $Q$ once in place of its $m$ copies. Its states are then $q_0$,
$V_1, \ldots, V_{m+1}$, $U$, $D$ and the states of $Q$, with $V_{i+1} \xrightarrow{\code{\#}} s_i$ for
$i = 1, \ldots, m$, with $q \xrightarrow{b} \delta(q, b)$ for both letters $b$, and with
$q \xrightarrow{\code{\symbol{36}}} D$ for every $q \neq r$. It accepts the words the automaton of the proof accepts
over these sources, and the scan depends on the accepted words alone, a token being the longest accepted word that
begins where it does. The three observations and the first part of this corollary therefore hold for it, and a shortest
input refuting the crossed verdict of gap $m$ has $|z| + m + 3$ bytes. The automaton has
$N = |Q| + m + 4 \leq 2 |Q| + 4$ states and the window $m + 2 \leq N$ bytes, so that length is $2^{\Omega(N)}$. Every
state is live. The states of $Q$ are reached from $s_1$, those other than $r$ lead to $D$ on \code{\symbol{36}}, and $r$
leads to one of them, $B$ being strongly connected with more than one state. For the last claim, the machine of the
membership proof of Theorem~\ref{thm:gap-pspace}, asked for a set bit, has an accepting computation that repeats no
configuration whenever it has one, and it reads one byte in each step.
\end{proof}

For $m = 1, \ldots, 9$ the lengths over the cycle sources are the entries of the shortest-window column of
Table~\ref{tab:family}, reached here by a property of the scan where Theorem~\ref{thm:long} concerns the model. Only the
lengths are shared. The count $(m^2 + 3m + 8)/2$ is of all the states of the automaton built here, while
Table~\ref{tab:family} counts the live states $n_m$ of the automaton of Section~\ref{sec:complexity}, a different one. A
counterexample that \code{crossing\_counterexample()} returns is a shortest input of this kind, so no polynomial bounds
its length either.

The questions that Theorem~\ref{thm:gap-pspace} leaves aside are easier. Whether a window occurs, whether a token
crosses a given gap at some occurrence, and whether at some occurrence the token containing the window's last byte
begins elsewhere than at a proposed origin are each shown by an input that puts the window between two short words,
taken from two lists that do not depend on the window. For a state $q$ of the token automaton and a word $w$ write
$\delta(q, w)$ for the state that $w$ leads to from $q$, and $\delta(q, w) = \bot$ when the run is undefined, where
$\bot$ is not a state and does not accept, and $\delta(\bot, w) = \bot$. Let $n$ be the number of states. A token
\emph{begins at gap $g$} of an occurrence at offset $t$ when it begins at $t + g$.

\begin{lemma}[Short contexts]
\label{lem:contexts}
For every state $q$ that some nonempty word leads to from $q_0$, fix a shortest such word $u_q$. For every state $p$ and
every $b$ that is a state or $\bot$, fix a shortest nonempty word $v_{p,b}$ such that $\delta(p, v_{p,b})$ accepts and
$\delta(b, v')$ accepts for no nonempty prefix $v'$ of $v_{p,b}$, when such a word exists. Let $\mathcal{U}$ hold the
empty word and the words $u_q$, each of at most $n$ bytes, and let $\mathcal{V}$ hold the empty word and the words
$v_{p,b}$, each of at most $n^2$ bytes and those with $b = \bot$ of at most $n$. Let $W$ be a window of $k$ bytes, $g$ a
gap of it and $0 \leq o < k$. If $W$ occurs in a completely tokenizable input, then $u W v$ is completely tokenizable
for some $u \in \mathcal{U}$ and some $v$ that is empty or one of the words $v_{p,\bot}$. If a token crosses gap $g$ at
some occurrence of $W$ in a completely tokenizable input, then such $u$ and $v$ can be chosen so that a token crosses
gap $g$ at the occurrence at offset $|u|$ in $u W v$. If at some such occurrence the token containing the last byte of
$W$ does not begin at gap $o$, then the same holds at the occurrence at offset $|u|$ in $u W v$ for some
$u \in \mathcal{U}$ and $v \in \mathcal{V}$.
\end{lemma}

\begin{proof}
A word $u_q$ has at most $n$ bytes because a shortest word visits no state twice after its first byte. Along a word
$v_{p,b}$ the runs from $p$ and from $b$ are never in the same state, since they would stay together from there and
accept together at its end. The word therefore leads $(p, b)$ through pairs of a state and a different state or $\bot$,
of which there are $n^2$, and a shortest one visits no pair twice after its first byte, so it has at most $n^2$ bytes.
For $b = \bot$ the second component stays $\bot$, which leaves $n$ pairs. Let $x$ be completely tokenizable with $W$ at
offset $t$, and keep of $x$ the bytes from the start of the token containing the byte at $t$ to the end of the token
containing the byte at $t + k - 1$. The scan restarts in $q_0$ at every boundary, and a token that is the longest
accepted word beginning where it does in $x$ is the longest in the shorter input too, so the bytes kept, written
$u W v$, are completely tokenizable with the tokens $T_0, \ldots, T_j$ that $x$ has there. The token $T_0$ begins with
$u$ and contains the first byte of $W$, every later token begins at a gap of $W$, and $T_j$ contains the last byte of
$W$ and ends with $v$. A gap below $k$ is a boundary of this occurrence exactly when a token begins at it, and gap $k$
exactly when $v$ is empty.

If $u$ is not empty, replace it by $u_q$ for $q = \delta(q_0, u)$. From the first byte of $W$ on, the run of the first
token passes the states it passed before, and $T_0$ ended after that byte, so the scan commits a token ending where
$T_0$ ended and continues as before. Assume therefore that $u \in \mathcal{U}$. If $v$ is empty the input has the
required form, so let $v$ be nonempty. For each $i \leq j$ let $p_i$ be the state that the run of $T_i$ from its start
is in after the last byte of $W$, or $\bot$ when that run is undefined there, so that $p_j$ is a state and
$\delta(p_j, v)$ accepts. For the first two claims let $b = \bot$. For the third, if a token $T_i$ with $i < j$ begins
at gap $o$, let $b = p_i$. That token is the longest accepted word beginning at gap $o$, and it ends inside $W$, so
$\delta(b, v')$ accepts for no nonempty prefix $v'$ of $v$. In every other case let $b = \bot$. The word $v_{p_j, b}$
then exists, $v$ being a word of its kind.

Among the indices $i \leq j$ for which $v_{p_i, b}$ exists, take those whose word is shortest and of them the least $i$,
and write $h = v_{p_i, b}$. Scan $y = u W h$. Up to the last byte of $W$ it holds the bytes of $u W v$, so for each
$l < i$ the accepted words from the start of $T_l$ that end inside $W$ are the same in both, and $T_l$ is the longest of
them. No accepted word from that start ends after $W$ in $y$. Such a word would give a nonempty prefix $h'$ of $h$ with
$\delta(p_l, h')$ accepting, and since the nonempty prefixes of $h'$ are prefixes of $h$, the word $v_{p_l, b}$ would
exist with at most $|h'| \leq |h|$ bytes, against the choice of $i$. The scan therefore commits $T_0, \ldots, T_{i-1}$
and reaches the start of $T_i$, whose run is in $p_i$ after $W$ and accepts at the end of $y$. Then $y$ is completely
tokenizable, and its last token begins where $T_i$ does and contains the last byte of $W$. It contains $W$ at offset
$|u|$, and no token of it begins at a gap of that occurrence where no token of $u W v$ began. A gap that was crossed is
therefore crossed in it, gap $k$ because $h$ is not empty. For the third claim, the last token of $y$ does not begin at
gap $o$. When $i = j$ this is the hypothesis, and when $i < j$ a token $T_i$ beginning at gap $o$ would have $p_i = b$,
while $\delta(b, h)$ does not accept and $\delta(p_i, h)$ does. For the first two claims $b = \bot$, so $h$ is one of
the words $v_{p,\bot}$.
\end{proof}

\begin{theorem}[A given window]
\label{thm:fixed-nl}
Given a compiled token set and a window $W$, the following are NL-complete: whether $W$ occurs in a completely
tokenizable input; for a gap $g$ of $W$, whether a token crosses $g$ at some occurrence, and the verdict forced at $g$;
and for an origin $o$, whether $(W, o)$ is a certified split window in the sense of Definition~\ref{def:window}, and
whether it is a witnessed one. Each remains NL-hard when the token set has one token kind over two symbols and a start
state that neither accepts nor is re-entered, and the window is a single byte. These reductions do not guarantee that
every state is live. Over automata whose states are all live, the first three questions are decided in deterministic
logarithmic space, and so are the two questions about a pair when the window has one byte, or two bytes over two
symbols. The two questions about a pair remain NL-hard over such automata with three symbols and a window of two bytes,
and with two symbols and a window of three bytes. Over a token automaton of $n$ states and a window of $k$ bytes, an
occurrence of the window and a token crossing a given gap at an occurrence are each found in an input of at most
$k + 2n$ bytes, and an occurrence at which the token containing the window's last byte does not begin at a given origin
in one of at most $k + n^2 + n$ bytes, when some completely tokenizable input has one. The quadratic term has the right
order. For all $r, s \geq 1$ there is a token set with one token kind over four symbols and $r + s + 3$ states, all
live, with a window of two bytes and an origin of it, such that a shortest input with an occurrence of that last kind
has $r s + 2$ bytes, which for every number $n \geq 5$ of states is $\lfloor (n - 3)^2 / 4 \rfloor + 2$ at a suitable
choice of $r$ and $s$. Coded over three symbols the token set has $3 (r + s) + 3$ states, all live, and a shortest such
input of $2 + (2 r + 1) s$ bytes, and coded over two symbols, with a window of three bytes, it has $6 (r + s) + 5$
states, all live, and one of $3 + (4 r + 1) s$ bytes. Over two symbols a window of two bytes needs only linearly many,
some input with an occurrence of that kind having at most $2n + 1$ bytes whenever one exists, and for every $n \geq 2$
some token set of $n$ states over two symbols needs that many.
\end{theorem}

\begin{proof}
\emph{Membership.} Consider three questions about a completely tokenizable input: whether some has an occurrence of $W$,
whether some has an occurrence at which a token crosses gap $g$, and whether some has an occurrence at which the token
containing the last byte of $W$ does not begin at gap $o$. The verdict forced is the complement of the second, the
certified pair is the complement of the third, and the witnessed pair is the first together with the complement of the
third. NL is closed under complement~\cite{immerman1988nondeterministic,szelepcsenyi1988method} and under intersection,
so it suffices to place the three questions in NL. By Lemma~\ref{lem:contexts} each, when it has the answer yes, has an
input $u W v$ showing it at offset $|u|$ with $|u| \leq n$ and $|v| \leq n^2$, and with $|v| \leq n$ for the first two,
which are the length bounds of the statement. Cut down as in the proof of the lemma, that input begins with the token
containing the first byte of $W$ and ends with the token containing the last.

A nondeterministic machine guesses such an input without writing it down, holding a fixed number of states of the token
automaton and of positions in $W$ and a counter up to $n^2$. It first chooses whether the first token begins at gap $0$,
in $q_0$, or before the window, in which case it guesses a nonempty $u$ of at most $n$ bytes one byte at a time and
keeps the state $\delta(q_0, u)$. It then follows the tokens through the window. Standing at the start of a token, at a
gap $c$ and in the state the token has there, it runs the automaton over $w_c \cdots w_{k-1}$, noting the last gap
$e > c$ at which the run is in an accepting state and the state $p$ after $w_{k-1}$ when the run gets that far. It may
\emph{advance} when $e$ exists and $e < k$, taking the token to end at gap $e$ and the next to begin there in $q_0$. It
may \emph{select} the token when the run gets as far as $p$, taking it to contain the last byte of $W$. After selecting,
it guesses $v$, which is empty, allowed when $p$ accepts, or has at most $n^2$ bytes, read one at a time from $p$ and
ending in an accepting state. For the first question the machine then accepts. For the second it accepts when none of
the tokens it followed began at gap $g$ and, if $g = k$, when $v$ is not empty. For the third it accepts when the
selected token does not begin at gap $o$ and, if an earlier token began there, when $\delta(b, v')$ accepts for no
nonempty prefix $v'$ of $v$, where $b = \delta(q_0, w_o \cdots w_{k-1})$ is computed by one more pass over $W$ and the
condition is checked while $v$ is guessed. The work tape holds $O(\log n + \log k)$ bits.

Suppose the machine accepts, and scan $u W v$ for the words it guessed. A token it advanced over is the longest accepted
word from its start that ends inside $W$. As in the proof of Lemma~\ref{lem:contexts}, either the scan commits those
tokens and then the selected one to the end of the input, or the first of them that it does not commit extends to the
end of $u W v'$ for a nonempty prefix $v'$ of $v$, and $u W v'$ is completely tokenizable with its tokens beginning
where the machine's tokens up to that one began. Either input contains $W$, and no token of it begins at a gap of the
occurrence where the machine followed none. Gap $k$ is crossed in the first input when $v$ is not empty, and in the
second, so the condition of the second question gives a crossing of gap $g$. Under the condition of the third, a token
the machine followed from gap $o$ has no accepted word that ends after $W$, so it cannot be the first token to extend
past $W$, and the token containing the last byte of $W$ begins where another of the machine's tokens began, and not at
gap $o$. Conversely, the machine accepts an input $u W v$ of the kind above by following its tokens. Each token before
the last ends inside $W$ at the last gap where its run accepts, which permits the advance, the run of the last gets past
$W$ and accepts after $v$, and the conditions hold because the tokens are those of the scan.

\emph{Hardness.} Reachability in a directed graph is NL-complete~\cite{jones1975space}. Let the graph have the vertices
$1, \ldots, N$, and ask whether a path leads from $s$ to $t$. The token DFA has the symbols \code{0} and \code{1} and
one token kind. Its states are the root $q_0$, a state $Z$, an end state $D$, which alone accepts, and a state $(v, j)$
for every vertex $v$ and every $j$ from $0$ to $N$. The root has $q_0 \xrightarrow{\code{0}} Z$ and
$q_0 \xrightarrow{\code{1}} D$, and $Z \xrightarrow{\code{0}} (s, 0)$. A state $(v, j)$ has
$(v, j) \xrightarrow{\code{0}} (v, j + 1)$ when $j < N$ and $(v, j) \xrightarrow{\code{1}} (j, 0)$ when $j \geq 1$ and
the graph has an edge from $v$ to $j$, and $(t, 0) \xrightarrow{\code{1}} D$. There are no other transitions. The tokens
are \code{1} and the words $\code{00}\, y\, \code{1}$ for $y \in \{\code{0}, \code{1}\}^*$ leading from $(s, 0)$ to
$(t, 0)$, and such a $y$ exists exactly when a path leads from $s$ to $t$, a path through the vertices
$s, v_1, \ldots, v_r = t$ giving $y = \code{0}^{v_1} \code{1} \cdots \code{0}^{v_r} \code{1}$ and each \code{1} of a
word $y$ following an edge. Let $W = \code{1}$, which occurs, the input \code{1} being a token. For a word $y$ of that
kind the input $\code{00}\, y\, \code{1}$ is one token, which crosses gap $0$ of the occurrence at its end. Conversely a
token that crosses gap $0$ at an occurrence of \code{1} contains that byte and begins before it, so it is a word
$\code{00}\, y\, \code{1}$. A token therefore crosses gap $0$ at some occurrence exactly when a path leads from $s$ to
$t$. Gap $0$ is forced exactly when none does, and since the window has one byte and occurs, the pair $(W, 0)$ is then
certified and witnessed, and it is neither otherwise. The absence of a path is NL-complete too, NL being closed under
complement, so the four questions are NL-hard. For occurrence, remove the transition $q_0 \xrightarrow{\code{1}} D$ and
let $Z$ accept as well. The tokens are then \code{0} and the words $\code{00}\, y\, \code{1}$, and \code{1} occurs in a
completely tokenizable input exactly when a path leads from $s$ to $t$. Both automata are written in logarithmic space,
the root does not accept, and nothing enters it.

\emph{Live automata.} In these automata the states $(v, j)$ from which $(t, 0)$ is not reached are not live, and finding
the live states is the reachability question. Suppose every state is live. The machine above then needs no guess for
occurrence or for the crossing of a gap. The first token begins at gap $0$ in $q_0$, or before the window in a state
that a nonempty word leads to from $q_0$, which is any state that a transition enters, every state being reachable, and
each of these cases is tried. The tokens it follows through the window are determined by the first, so each of them is
tried as the selected one. A selected token in the state $p$ after $W$ has a completion, $p$ being live, and it has a
nonempty one exactly when a transition leaves $p$. A deterministic machine runs through these cases in logarithmic space
and decides occurrence, the crossing of a gap and so the verdict forced. Origin refutation keeps its condition on $b$,
and the two questions about a pair stay NL-hard. Take the symbols
$\code{0}, \code{1}, \code{a}, \code{c}, \code{e}, \code{\#}, \code{!}$, the root $q_0$, the end state $D$, states $A$,
$R$, $R'$ and $E_1, \ldots, E_N$, and a state $(v, j)$ for every vertex $v$ and every $j$ from $1$ to $N$, with
$(v, j) \xrightarrow{\code{0}} (v, j + 1)$ when $j < N$ and $(v, j) \xrightarrow{\code{1}} (j, 1)$ when the graph has an
edge from $v$ to $j$. Let $q_0 \xrightarrow{\code{a}} A$ and $A \xrightarrow{\code{c}} R$, let $R$ keep itself on
\code{0} and \code{1}, and let $R \xrightarrow{\code{!}} D$, $R \xrightarrow{\code{\#}} R'$ and
$R' \xrightarrow{\code{!}} D$. Let $q_0 \xrightarrow{\code{c}} (s, 1)$ and $(t, 1) \xrightarrow{\code{\#}} D$, and let
every $(v, j)$ have $(v, j) \xrightarrow{\code{!}} D$. Let $q_0 \xrightarrow{\code{e}} E_1$,
$E_v \xrightarrow{\code{0}} E_{v+1}$ for $v < N$ and $E_v \xrightarrow{\code{1}} (v, 1)$. The accepting states are $A$
and $D$. Every state is reached from the root, through $E_1, \ldots, E_N$ for the states $(v, j)$, and every state other
than $D$ leads to $D$, so every state is live. Let $W = \code{ac}$, which occurs in the token $\code{ac!}$. The byte
\code{a} is read from the root alone, so it begins a token, and \code{c} is read from $A$ or from the root. At an
occurrence of $W$ the token containing \code{c} therefore begins at gap $0$ or at gap $1$. If it begins at gap $1$, the
scan has committed the single byte \code{a}, and the token is $\code{c}\, y\, \code{!}$ or $\code{c}\, y\, \code{\#}$
for a word $y \in \{\code{0}, \code{1}\}^*$. The first is excluded, since $\code{ac}\, y\, \code{!}$ is then an accepted
word from the start of \code{a}, through $A$, $R$ and $D$. In the second, $y$ leads from $(s, 1)$ to $(t, 1)$.
Conversely, for such a $y$ the input $\code{ac}\, y\, \code{\#}$ scans as \code{a} and $\code{c}\, y\, \code{\#}$, the
run from \code{a} ending in $R'$ with no acceptance after $A$. The pair $(W, 0)$ is therefore certified, and witnessed,
exactly when no path leads from $s$ to $t$.

Fewer symbols suffice when each of the seven is written as a codeword and every state reads the codewords of its symbols
through a trie of fresh states that do not accept, all of which are live. The codewords below form a prefix code, so the
accepted words of the new automaton are the codings of the accepted words of the old, the completely tokenizable inputs
are the codings of the completely tokenizable inputs, and the tokens of a coding are the codings of the tokens. Over
three symbols write \code{a} as \code{2}, \code{c} as \code{1} and the other five symbols as the words
$\code{0}\, \code{1}^i\, \code{0}$ for $i = 0, \ldots, 4$. The byte \code{2} is then the codeword of \code{a} and
nothing else, and a byte \code{1} after it begins a codeword, which is that of \code{c}. The window \code{21} therefore
occurs in a coding only as the coding of an occurrence of \code{ac}, and the pair $(\code{21}, 0)$ is certified, and
witnessed, exactly when no path leads from $s$ to $t$. Over two symbols write \code{a} as \code{1}, \code{c} as
\code{01} and the other five as $\code{00}\, \code{1}^i\, \code{00}$, and take the window \code{101}. At an occurrence
in the coding of a completely tokenizable input, its first byte is the codeword of \code{a}, the last byte of the
codeword of \code{c}, or a byte of a word $\code{00}\, \code{1}^i\, \code{00}$. In the first case the bytes \code{01}
after it begin a codeword, which is that of \code{c}. In the second they would begin a second \code{c}, and no \code{c}
follows a \code{c} in a completely tokenizable input of the old automaton, since \code{c} leads to $R$ or to $(s, 1)$,
and neither accepts nor reads \code{c}. In the third the next bytes are \code{1} or \code{00}. In the coding of a
completely tokenizable input the window therefore occurs only as the coding of an occurrence of \code{ac}, and the pair
$(\code{101}, 0)$ is certified, and witnessed, exactly when no path leads from $s$ to $t$.

Shorter windows are easier over live states. For a window of one byte the pair $(W, 0)$ is certified exactly when gap
$0$ is forced, and for a window $xy$ of two bytes the pair $(xy, 1)$ exactly when gap $1$ is, so these are decided in
deterministic logarithmic space, together with occurrence for the witnessed pair. Let the alphabet have two symbols, and
consider $(xy, 0)$. It is not certified unless gap $0$ is forced, so let gap $0$ be forced. An accepted word is a
completely tokenizable input of one token, so no accepted word contains $xy$ after its first byte. An occurrence refutes
the origin exactly when a token begins at its gap $1$, the scan having committed the single byte $x$ before it, and the
input cut down to that byte and that token is still scanned as the two. If $x \neq y$, the token begins with $y$ and
does not contain $xy$, so it is $v = y^{h} x^{i}$ with $h \geq 1$. Here $i = 0$. Otherwise the input $v v$ would be
scanned as $v$ and $v$, since an accepted prefix of $v v$ longer than $v$ would contain $xy$ after its first byte, and
the occurrence of $xy$ across the two tokens would then have its gap $0$ inside the first of them. The token is thus
$y^{h}$. Let $j \geq 1$ be least such that $y^{j}$ is accepted, so that $j \leq h$, and $j \leq n$ because every state
that a word $y^{l}$ with $l \geq 1$ leads to from $q_0$ is reached by one with $l \leq n$. No word $x\, y^{l}$ with
$1 \leq l \leq h$ is accepted, the scan having committed $x$, so the input $x\, y^{j}$ is scanned as $x$ and $y^{j}$ as
well. The origin is therefore refuted exactly when $\delta(q_0, x)$ accepts, some word $y^{j}$ with $j \geq 1$ is
accepted, and for the least such $j$ no word $x\, y^{l}$ with $1 \leq l \leq j$ is accepted. A deterministic machine
checks this by following the runs from $q_0$ and from $\delta(q_0, x)$ over at most $n$ bytes $y$. If $x = y$, let
$A = \delta(q_0, x)$ and $R = \delta(A, x)$. A refuting occurrence needs $A$ to accept and $R$ not to, the scan having
committed $x$ before a second $x$. Conversely, when $A$ accepts and $R$ does not, whether $R$ is a state or undefined,
the input $xx$ itself is scanned as $x$ and $x$. The pair $(xx, 0)$ is therefore certified exactly when gap $0$ is
forced and either $A$ does not accept or $R$ does. The cases $x \neq y$ and $x = y$ do not use that every state is live,
and they give part of the bound for a window $xy$ of two bytes over two symbols, whether or not every state is live. An
occurrence refutes the origin $1$ exactly when a token contains $xy$. With its bytes before $x$ replaced by a shortest
word to the same state and its bytes after $y$ by a shortest word to an accepting state, that token is an accepted word
of at most $2n$ bytes, and as an input it is one token. When gap $0$ is forced and the origin $0$ is refuted, the input
$x\, y^{j}$ above, or $xx$, refutes it with at most $n + 1$ bytes. When gap $0$ is not forced, a token crosses it at
some occurrence, and the bytes of that token before $x$ may be taken to be a word $u_q$ of Lemma~\ref{lem:contexts}, of
at most $n$ bytes. If the token contains $y$, a shortest word to an accepting state in the place of its bytes after $y$
leaves an accepted word of at most $2n + 1$ bytes, one token that begins before $x$. If it ends with $x$, a token that
begins with $y$ follows it, so some accepted word begins with $y$, and a shortest such word $y\, v$ has at most $n$
bytes. The scan of $u_q\, x\, y\, v$ commits $u_q\, x$ and then $y\, v$, a token that begins at gap $1$, unless a longer
prefix of the input is accepted, and the longest accepted prefix is then an input of one token that contains $y$ and
begins before $x$. Each of these inputs refutes the origin $0$ with at most $2n + 1$ bytes. That many are needed. For
two different bytes $x$ and $y$ and $n \geq 2$ let the states $0, \ldots, n - 1$ lie on a cycle that reads $x\, y^{n-1}$
from the root $0$, and let the state $1$ alone accept. The tokens are the words $(x\, y^{n-1})^i\, x$ with $i \geq 0$,
each beginning and ending with $x$, so in a completely tokenizable input a $y$ lies in the token of the $x$ before it.
At an occurrence of $xy$ that refutes the origin $0$, the $x$ is therefore not the first byte of its token, which has
$x\, y^{n-1}$ before that $x$ and $y^{n-1}\, x$ after it, and a shortest such input is the token $(x\, y^{n-1})^2\, x$
of $2n + 1$ bytes.

\emph{The quadratic order.} For $r, s \geq 1$ take the symbols \code{x}, \code{y}, \code{a} and \code{b}, the root
$q_0$, states $A$ and $D$, which accept, and states $P_0, \ldots, P_{s-1}$ and $C_0, \ldots, C_{r-1}$, which do not. Let
$q_0 \xrightarrow{\code{x}} A$, $q_0 \xrightarrow{\code{y}} P_0$ and $A \xrightarrow{\code{y}} C_0$. Every $P_i$ keeps
itself on \code{a} and has $P_i \xrightarrow{\code{b}} P_{i+1}$, where $P_s$ stands for $D$. Every $C_j$ with
$j < r - 1$ has $C_j \xrightarrow{\code{a}} C_{j+1}$ and $C_j \xrightarrow{\code{b}} D$, and the last has
$C_{r-1} \xrightarrow{\code{a}} D$ and $C_{r-1} \xrightarrow{\code{b}} C_0$. There are no other transitions, so every
state is live, the root does not accept, and nothing enters it. Let $W = \code{xy}$ with the origin $0$. The byte
\code{x} is read from the root alone and \code{y} from the root or from $A$, so at an occurrence the token containing
\code{y} begins at gap $0$ or at gap $1$. Suppose it begins at gap $1$. It is then $\code{y}\, v$ for a word $v$ that
leads $P_0$ to $D$, which holds the byte \code{b} exactly $s$ times. The scan has committed the single byte \code{x}, so
the run from that byte, which is in $C_0$ after \code{y}, accepts after no nonempty prefix of $v$. From a state $C_j$
with $j < r - 1$ the next byte of $v$ is therefore \code{a}, and from $C_{r-1}$ it is \code{b}, which makes
$v = (\code{a}^{r-1}\, \code{b})^s$. Conversely the input $\code{xy}\, (\code{a}^{r-1}\, \code{b})^s$ is scanned as
\code{x} and $\code{y}\, (\code{a}^{r-1}\, \code{b})^s$. Every input with an occurrence that refutes the origin thus
contains that word, which has $r s + 2$ bytes, and its last $r s$ bytes are the word $v_{P_0, C_0}$ of
Lemma~\ref{lem:contexts}. The automaton has $n = r + s + 3$ states, so for every $n \geq 5$ the choice of
$r = \lfloor (n - 3) / 2 \rfloor$ and $s = \lceil (n - 3) / 2 \rceil$ gives these lengths as
$\lfloor (n - 3)^2 / 4 \rfloor + 2$ and $\lfloor (n - 3)^2 / 4 \rfloor$.

The order survives the coding of the four symbols, each state reading the codewords of its symbols through a trie of
fresh states as above, all of them live. Over three symbols write \code{x} as \code{2}, \code{y} as \code{1}, \code{a}
as \code{00} and \code{b} as \code{010}, which adds two states under every $P_i$ and every $C_j$ and gives
$3 (r + s) + 3$ states. The byte \code{2} is the codeword of \code{x} and nothing else, and a byte \code{1} after it
begins a codeword, which is that of \code{y}. The window \code{21} therefore occurs in a coding only as the coding of an
occurrence of \code{xy}, and a shortest input refuting its origin $0$ is the coding of
$\code{xy}\, (\code{a}^{r-1}\, \code{b})^s$, of $2 + (2 r + 1) s$ bytes. Over two symbols write \code{x} as \code{1},
\code{y} as \code{01}, \code{a} as \code{0000} and \code{b} as \code{00100}, which adds one state under the root, one
under $A$ and five under every $P_i$ and every $C_j$ and gives $6 (r + s) + 5$ states, and take the window \code{101}.
At an occurrence in the coding of a completely tokenizable input, its first byte is the codeword of \code{x}, the last
byte of the codeword of \code{y}, or the byte \code{1} of the codeword of \code{b}. In the first case the bytes
\code{01} after it begin a codeword, which is that of \code{y}. In the second they would begin a second \code{y}, and no
\code{y} follows a \code{y} in a completely tokenizable input, since \code{y} leads to $P_0$ or to $C_0$, and neither
accepts nor reads \code{y}. In the third the next bytes are \code{00}. The window therefore occurs there only as the
coding of an occurrence of \code{xy}, and a shortest input refuting its origin $0$ is the coding of the same word, of
$3 + (4 r + 1) s$ bytes.
\end{proof}

The lists also decide the three questions directly, by scanning $u W v$ for every $u \in \mathcal{U}$ and
$v \in \mathcal{V}$, at most $n + 1$ words and at most $n^2 + 1$. The two halves of the open verdict thus differ. A
crossing is found in nondeterministic logarithmic space, or deterministic over live states, and in a short input, while
a boundary is PSPACE-complete to find and can need an input of exponential length. A boundary at a gap requires each
earlier token of the occurrence to end where it does, which is a condition on the rest of the input, whereas a crossing,
or an origin other than the proposed one, survives when an earlier token extends past the window.

The lemma also bounds the question that has no window given, and that question is as hard as the model's.

\begin{proposition}[Existence]
\label{prop:existence}
Whether a compiled token set has a witnessed certified split window, in the sense of Definitions~\ref{def:window}
and~\ref{def:witnessed}, is PSPACE-complete. It remains PSPACE-hard when the token set has one token kind over two
symbols, a start state that neither accepts nor is re-entered, and every state live.
\end{proposition}

\begin{proof}
\emph{Membership.} By Lemma~\ref{lem:contexts} a pair $(W, o)$ is a witnessed certified split window exactly when
$u W v$ is completely tokenizable for some $u \in \mathcal{U}$ and $v \in \mathcal{V}$, and in every completely
tokenizable $u W v$ with $u \in \mathcal{U}$ and $v \in \mathcal{V}$ the token containing the last byte of the
occurrence at offset $|u|$ begins at gap $o$. The first condition says that $W$ occurs. The second follows from the
certificate, and when it holds no occurrence refutes the origin, since by the lemma one at offset $|u|$ in such an input
would. The lists are computed in polynomial time and hold polynomially many words of polynomial length.

A nondeterministic machine guesses $W$ one byte at a time, marking one byte as the origin, and scans all the inputs
$u W v$ together, each through its $u$, then all of them through the guessed bytes, then each through its $v$. Maximal
munch reads ahead, so each scan is told what it needs of the bytes it has not read. For a word $z$ let $S(z)$ be the set
of states $q$ such that $\delta(q, z')$ accepts for some nonempty prefix $z'$ of $z$. The set of the empty word is
empty, and for a byte $a$ the set $S(a z)$ holds the states $q$ such that $\delta(q, a)$ accepts or lies in $S(z)$. For
each input the machine guesses the set of the unread suffix before the first byte and after every byte, checks that
equation as it reads the byte, and requires the empty set at the end, and by induction from the end only the true sets
pass. With them the scan is deterministic. When bytes remain at a boundary, a token begins in $q_0$, which must lie in
the set of the unread suffix. After each byte the run is in a state $q$. It continues when $q$ lies in the set of the
unread suffix, a longer accepted word lying ahead. It commits the token when $q$ accepts and does not lie in that set,
the token being the longest accepted word from its start. Otherwise the scan fails, and the input is not completely
tokenizable. For each input the machine keeps a state, a set, whether the scan has failed, and whether the last token to
begin at a byte of $W$ so far began at the marked byte, which is false while none has. That flag is frozen when $W$
ends, and the equations of the sets go on being checked for a scan that has already failed, so that every failure is one
of the true scan. It accepts when it has ended $W$ with a byte marked, some scan has reached the end of its input
without failing, and in every such scan the token containing the last byte of $W$ began at the marked byte. This takes
polynomial space, a computation that repeats no configuration is bounded by a step counter of polynomially many bits,
and NPSPACE is PSPACE.

\emph{Hardness.} A partial automaton $B = (P, \Sigma, \gamma)$ has a finite set $P$ of states and a partial transition
function $\gamma$, with no initial or accepting states. A word $w$ has \emph{rank one} in $B$ when $\gamma(p, w)$ is
defined for some state $p$ and is the same state wherever it is defined, and deciding whether a partial automaton has a
word of rank one is PSPACE-complete~\cite[Theorem~1, p.~2]{ryzhikov2019birecurrent}. Given $B$, take as symbols its
letters and one \emph{marker} $m_p$ for every state $p$, number them from $1$ with the letters first, and write symbol
number $i$ as the codeword $\code{1}^i\, \code{0}$. For a word $y$ of symbols let $c(y)$ be its coding and $\pi(y)$ the
word of its letters, the markers deleted. The token automaton over \code{0} and \code{1} has the root $q_0$, the states
of $P$, all of which accept, and under the root and under each state of $P$ a chain of states that do not accept, which
reads the bytes \code{1} of a codeword, the byte \code{0} then leading to the target of the symbol. From the root the
codeword of $m_p$ leads to $p$. From a state $p$ the codeword of a letter $a$ leads to $\gamma(p, a)$ when that is
defined, and the codeword of a marker $m_r$ with $r \neq p$ leads back to $p$. Nothing else is defined, and each chain
ends at the highest symbol that its state reads, so every state is live, the root does not accept, and nothing enters
it. The automaton is written in polynomial time.

Accepting states are entered only at the end of a codeword, so every token is the coding of a word of symbols, and so is
every completely tokenizable input. A token begins with a marker $m_p$ and enters $p$. In a state $p$ the automaton
reads the codeword of every symbol except $m_p$ and the letters at which $\gamma(p, \cdot)$ is undefined, a letter
moving it as $\gamma$ does and a marker leaving it in $p$. Every state of $P$ accepts, so a token ends before the first
codeword that the state $p$ it has reached does not read, or at the end of the input. The scan goes on only if that
codeword is one the root reads, which makes it $m_p$, and the next token then begins in $p$ again. The state of $P$ that
the scan is in is thus moved by the letters alone. An input $c(m_p\, y)$ is completely tokenizable exactly when
$\gamma(p, \pi(y))$ is defined, and a token begins at a marker of $y$ exactly when it is the marker of the state reached
by the letters before it.

Let $w$ have rank one, with $r$ the state it leads to, and let $W = \code{0}\, c(w\, m_r)$ with its origin at the first
byte of $c(m_r)$. In a completely tokenizable input every \code{0} ends a codeword, so an occurrence of $W$ begins at
the last byte of a codeword and continues with the codewords of $w$ and of $m_r$. The letters of $w$ are read from the
state $p$ that the scan is in after that first codeword, so $\gamma(p, w)$ is defined and is $r$, and a token begins at
$m_r$, with which the window ends. The pair is therefore certified. It is witnessed, since for a state $p$ at which
$\gamma(p, w)$ is defined the input $c(m_p\, w\, m_r)$ is completely tokenizable and contains $W$ from the last byte of
$c(m_p)$ on.

Conversely let $(W, o)$ be a witnessed certified split window, and take an occurrence of $W$ in a completely tokenizable
input. A token begins at gap $o$ of it, so a codeword of a marker $m_r$ begins there. Let $y$ be the block of symbols of
the input whose codewords hold a byte of the occurrence, and write $y = y_1\, m_r\, y_2$, where this $m_r$ is the marker
that begins at gap $o$. The first and the last codeword of the block may lie only partly inside the occurrence. Let
$v = \pi(y_1 y_2)$. For every state $p$ at which $\gamma(p, v)$ is defined, the input $c(m_p\, y_1\, m_r\, y_2)$ is
completely tokenizable and contains $W$ at the same place in $c(y_1\, m_r\, y_2)$, with gap $o$ at the first byte of
$c(m_r)$. A token begins there by the certificate, so $\gamma(p, \pi(y_1)) = r$ and
$\gamma(p, v) = \gamma(r, \pi(y_2))$. The word $v$ is defined at the state that the scan of the input at hand is in when
it reads the first letter of $y_1 y_2$, or at every state when there is no letter, and it leads every state at which it
is defined to one state, so it has rank one. A witnessed certified split window therefore exists exactly when $B$ has a
word of rank one.
\end{proof}

The proposition concerns Definition~\ref{def:window} itself, where Theorem~\ref{thm:pspace} concerns the model, and each
is proved by its own reduction. Reading a window from the right bounds how long a shortest one can be.

\begin{theorem}[Shortest windows of the definition]
\label{thm:semantic-length}
Let a token automaton have $n$ states. If its token set has a certified split window, in the sense of
Definition~\ref{def:window}, then a shortest one has at most $(n + 1)\, 7^n + 4^n + 1$ bytes, and so has a shortest
witnessed one, in the sense of Definition~\ref{def:witnessed}, when there is one. With the lower bound of
Theorem~\ref{thm:plain}, both lengths are $2^{\Theta(n)}$ in the worst case.
\end{theorem}

\begin{proof}
Removing the states that are not live changes no scan, since a run passes no accepting state after one of them, so $n$
may be read as the number of live states when $q_0$ is live. Otherwise only the empty input is completely tokenizable,
and every window of one byte is certified. Let $(W, o)$ be a pair with $W$ of $k$ bytes, and let
$b = \delta(q_0, w_o \cdots w_{k-1})$. By the proof of Theorem~\ref{thm:fixed-nl}, $W$ occurs in a completely
tokenizable input exactly when the machine there has a run that selects a token and completes it, and $(W, o)$ is
certified exactly when no run of it accepts for the third question. The bounds on the lengths of the guessed $u$ and $v$
may be dropped, shortest words being within them. For a state $p$ write $H(p)$ when some word, possibly empty, leads $p$
to an accepting state, and $G(p)$ when some such word $v$ has no nonempty prefix $v'$ with $\delta(b, v')$ accepting,
and let both be false of $\bot$. For a gap $c \leq k$ let $W_c = w_c \cdots w_{k-1}$. Standing at a gap $c < k$ at a
token that is in the state $q$ there, the machine can select the token when $p = \delta(q, W_c)$ is a state. The
completion it then needs exists exactly when $H(p)$ holds, and one that meets the condition on $b$ exactly when $G(p)$
does. The machine can advance when some nonempty prefix of $W_c$ leads $q$ to an accepting state and the longest such
prefix is shorter than $W_c$, and the next token then begins in $q_0$ at the gap $e$ where that prefix ends.

Three questions are asked of a token the machine stands at. \emph{Occurrence} asks whether the machine can go on to
select a token and complete it. \emph{Refutation} asks whether it can go on to accept for the third question when no
earlier token began at gap $o$, and \emph{guarded refutation} asks the same when an earlier one did, so that the
completion has to meet the condition on $b$. For a gap $c < k$ let $O_c$ be the answer to occurrence for a token that
begins at $c$, and let $R_c$ be its answer to refutation when $c \leq o$ and to guarded refutation when $c > o$. Let
$O_k$ and $R_k$ be false, and write $D_c$ for $R_c$ when $c \leq o$ and for $O_c$ when $c > o$. This $D_c$ is the answer
to refutation for a token that begins at any gap $c$, since to the right of $o$, when no earlier token began at $o$, no
later one does either, and the machine then accepts for the third question exactly when it selects a token and completes
it.

A token is followed through the window as well. Let a token that began before a gap $c \leq k$ be in the state $q$ at
$c$, and let $p = \delta(q, W_c)$. The machine can select it when $p$ is a state. Let $e$ be the last gap after $c$, gap
$k$ included, at which the run from $q$ over $W_c$ accepts, when there is one. If $e < k$ the machine can advance from
the token to $e$. If $e = k$ it cannot advance, the longest accepted word of the run inside the window ending with the
window's last byte, and the answers $O_k$, $R_k$ and $D_k$, all false, stand for that. Only when there is no such gap
can it advance to the last gap up to $c$ at which the run of the token accepted after its start, and after gap $0$ if it
began before the window. In this last case the answer of a token that begins at that gap, to the question then asked of
it, is the \emph{fallback} of the token, which is false when there is no such gap either. For a truth value $x$ let
\[ F_c(q, x) = H(p) \vee \bar{O}, \qquad G_c(q, x) = G(p) \vee \bar{R}, \qquad A_c(q, x) = \bar{R}, \qquad
B_c(q, x) = H(p) \vee \bar{D}, \] where $\bar{O}$, $\bar{R}$ and $\bar{D}$ are $O_e$, $R_e$ and $D_e$ when the gap $e$
exists, and all three are $x$ when it does not. The map $A_c$ is used for $c \geq o$, $G_c$ for $c > o$ and $B_c$ for
$c \leq o$. With $x$ the fallback of the token, $F_c(q, x)$ is its answer to occurrence. Its answer to guarded
refutation is $G_c(q, x)$ when it began to the right of $o$. Its answer to refutation is $A_c(q, x)$ when it began at
$o$, where it cannot be selected and every later token is guarded, and $B_c(q, x)$ when it began to the left of $o$. A
token that begins at a gap has no fallback, so $O_c = F_c(q_0, \text{false})$ for every $c < k$, and $R_c$ is
$G_c(q_0, \text{false})$ for $o < c < k$, $A_o(q_0, \text{false})$ for $c = o$ and $B_c(q_0, \text{false})$ for $c < o$.
Each of these values is defined from the answers at later gaps, so all of them are defined from the right.

The maps of a gap are determined by those of the next. Let $1 \leq c \leq k$, let $a = w_{c-1}$ and $r = \delta(q, a)$,
and let $\Phi$ be one of the four maps that is used at both gaps, with $y$ standing for $O_c$ when $\Phi$ is $F$, for
$R_c$ when it is $G$ or $A$, and for $D_c$ when it is $B$. Then
\[ \Phi_{c-1}(q, x) \;=\; \begin{cases}
  x & \text{when $r = \bot$,} \\
  \Phi_c(r, x) & \text{when $r$ is a state that does not accept,} \\
  \Phi_c(r, y) & \text{when $r$ accepts.}
\end{cases} \]
In the first case the run from $q$ over $W_{c-1}$ is undefined after one byte, so $p = \bot$ and the run accepts
nowhere. In the second it has the state $p$ and the gap $e$ of the run from $r$ over $W_c$. In the third it has the same
$p$ and accepts at gap $c$, so its last accepting gap is that of the run from $r$ when there is one and $c$ otherwise,
and the map of $r$ gives both cases when $y$, the answer of a token that begins at $c$, is put in the place of the
fallback. At gap $k$ the maps are $F_k(q, x) = H(q) \vee x$, $G_k(q, x) = G(q) \vee x$ and $A_k(q, x) = x$, and at gap
$o$ the map $B_o$ is $F_o$, every later gap lying to the right of $o$.

The first token of a run is at gap $0$. It begins there, and its answers to occurrence and to refutation are $O_0$ and
$R_0$, or it began before the window and is in a state $q$ that a nonempty word leads to from $q_0$, and its answers are
$F_0(q, \text{false})$ and $B_0(q, \text{false})$, a token that began before the window having no fallback because the
machine advances only to gaps after gap $0$. So $W$ occurs exactly when $O_0$ or some such $F_0(q, \text{false})$ holds,
and $(W, o)$ is certified exactly when neither $R_0$ nor any such $B_0(q, \text{false})$ does.

The \emph{label} of a gap $c$ with $o < c < k$ consists of the three maps $F_c(q, \cdot)$, $G_c(q, \cdot)$ and
$A_c(q, \cdot)$ of every state $q$, together with $\delta(q_0, w_o \cdots w_{c-1})$, the state that the run from the
origin is in at $c$, or $\bot$. The label of a gap $c < o$ consists of the two maps $F_c(q, \cdot)$ and $B_c(q, \cdot)$
of every state $q$. Suppose two gaps $c < c'$ have the same label, both of them between $o$ and $k$ or both to the left
of $o$, and let $W'$ be $W$ without the bytes $w_c \cdots w_{c'-1}$, with its origin at the byte where $W$ has it. The
state $b$ of $W'$ is that of $W$. Between $o$ and $k$ this holds because the run from the origin is in the same state at
$c$ and at $c'$, or undefined at both, and to the left of $o$ because the bytes from the origin on are unchanged. The
maps of a gap depend only on the bytes after it, on $b$, and on where the origin lies among those bytes or that it lies
to their left. In $W'$ the gap before $W_{c'}$ therefore has the maps that gap $c'$ has in $W$, which are those of gap
$c$. The answers $O$ and $R$ of a gap other than $o$ and $k$ are values of its maps at $q_0$, so the gaps to the left
get the maps and the answers they have in $W$, one byte at a time. So $W'$ occurs exactly when $W$ does, and it is
certified at its origin exactly when $W$ is.

A shortest certified split window therefore has different labels at the gaps between $o$ and $k$ and different labels at
the gaps to the left of $o$, and so has a shortest witnessed one. Between $o$ and $k$ a state $q$ has one of seven
triples of maps. When the run from $q$ over $W_c$ accepts at no gap after $c$, the map $A_c(q, \cdot)$ is the identity,
and $F_c(q, \cdot)$ and $G_c(q, \cdot)$ are each the identity or constantly true, the second only if the first, since
$G(p)$ implies $H(p)$. When it accepts at one, the three maps are constant, and $A_c(q, \cdot)$ implies $G_c(q, \cdot)$,
which implies $F_c(q, \cdot)$, since $G(p)$ implies $H(p)$ and $R_e$ implies $O_e$, a run that accepts for the third
question having selected a token and completed it. Writing $\iota$ for the identity and $0$ and $1$ for the constant
maps, that makes the triples $(\iota, \iota, \iota)$, $(1, \iota, \iota)$ and $(1, 1, \iota)$ of $F$, $G$ and $A$, and
$(0, 0, 0)$, $(1, 0, 0)$, $(1, 1, 0)$ and $(1, 1, 1)$. With the $n + 1$ values of the state of the run from the origin,
there are at most $(n + 1)\, 7^n$ labels. To the left of $o$ the two maps of a state are both the identity, or they are
constant with the second implying the first, since $D_e$ implies $O_e$. That makes the pairs $(\iota, \iota)$, $(0, 0)$,
$(1, 0)$ and $(1, 1)$ of $F$ and $B$, and at most $4^n$ labels. Counting the gaps $o$ and $k$ as well,
$k + 1 \leq (n + 1)\, 7^n + 4^n + 2$.
\end{proof}

A certified split window that need not be witnessed is no easier to decide and no shorter in the worst case, although a
token set has one as soon as some word occurs in none of its completely tokenizable inputs.

\begin{theorem}[Windows that need no witness]
\label{thm:plain}
Whether a compiled token set has a certified split window in the sense of Definition~\ref{def:window} is
PSPACE-complete. It remains PSPACE-hard when the token set has one token kind over two symbols, a start state that
neither accepts nor is re-entered, and every state live. For infinitely many $N$ there is a token set of that kind with
$N$ states that has a witnessed certified split window and whose certified split windows, witnessed or not, all have
$2^{\Omega(N)}$ bytes. With the upper bound of Theorem~\ref{thm:semantic-length}, the length of a shortest certified
split window, witnessed or not, is $2^{\Theta(n)}$ in the worst case.
\end{theorem}

The proof is in Appendix~\ref{app:plain}. Its token sets code six symbols in blocks of three bits, and a window too
short to hold a certain pattern occurs in completely tokenizable inputs at two offsets that differ by one, while tokens
begin only at offsets divisible by three, so that no origin of it is a boundary at both.

\paragraph{Anchor propagation.} Section~\ref{sec:preliminaries} noted that the weakest guarantee is already usable, a
worker cutting at the known boundary and rescanning the window's suffix. Read as a window, that rescan forces further
gaps. The statement needs a longest token. Write $L$ for the maximum token length of a literal token set, a finite set
of byte strings, or of a regular token set whose token languages are all finite; a regular token set with an infinite
token language has no $L$, and the lemma is not stated for it. Two facts precede it. First, extension is monotone. A
forced or crossed gap of $W$ stays forced or crossed, shifted, in every window that contains $W$, since every occurrence
of the larger window contains an occurrence of $W$ at the shifted offset. Second, the scan restarts in $q_0$ at every
boundary, which is the splitting companion's memorylessness lemma~\cite{nidhogg2026splitting}, so from a boundary $p$ of
a completely tokenizable $x$ it commits the longest token that is a prefix of $x[p \,..)$.

\begin{lemma}[Anchor propagation]
\label{lem:propagation}
Let the token set have maximum token length $L$, let $W$ occur in some completely tokenizable input, and let gap $g$ of
$W$ be forced. Define offsets $g = c_0 < c_1 < \cdots < c_m$ by $c_{i+1} = c_i + |\tau_i|$ for as long as
$c_i + L \leq k$, where $\tau_i$ is the longest token that is a prefix of $W[c_i \,..\, k)$. Then every $c_i$ is a
forced gap of $W$, and every gap strictly between $c_i$ and $c_{i+1}$ is crossed.
\end{lemma}

\begin{proof}
Take an occurrence of $W$ at offset $t$ in a completely tokenizable $x$, and suppose $t + c_i$ is a boundary of the
tokenization of $x$, as it is for $i = 0$ by hypothesis. The scan commits at $t + c_i$ the longest token prefixing
$x[t + c_i \,..)$, and since every token has length at most $L$, that token is determined by
$x[t + c_i \,..\, t + c_i + L)$, which is $W[c_i \,..\, c_i + L)$ when $c_i + L \leq k$. The committed token is
$\tau_i$, which exists because the occurrence is tokenized, and it is the same at every occurrence. Its end
$t + c_{i+1}$ is a boundary and no boundary lies strictly inside it, so induction on $i$ gives both claims.
\end{proof}

\begin{corollary}[Resolving width]
\label{cor:width}
Let the token set have maximum token length $L$, let $x$ be completely tokenizable, and let $b$ be a boundary of its
tokenization with $b + L \leq |x|$. Let $a < b$ be a boundary that some window $W_a$ of width $h_a$ forces: $W_a$ occurs
at $q$, and $a = q + g_a$ for a forced gap $g_a$. Then the window $x[q \,..\, e)$ with $e = \max(b + L,\, q + h_a)$
forces $b$ at its gap $b - q$. A window of width at most $(b - a) + h_a + L$ therefore forces $b$.
\end{corollary}

\begin{proof}
The window contains $W_a$ at offset $0$, so its gap $g_a$ is forced by monotonicity, and it occurs in $x$. Let
$a = p_0 < p_1 < \cdots < p_m = b$ be the boundaries of $x$ from $a$ to $b$. Each is the end of the token the scan
commits at the one before, and each $p_i$ with $i < m$ satisfies $p_i + L < b + L \leq e$. In the window's offsets these
are the $c_i$ of Lemma~\ref{lem:propagation}, each with $c_i + L \leq e - q$, so $c_m = b - q$ is forced. The width is
$e - q = \max((b - a) + g_a + L,\, h_a) \leq (b - a) + h_a + L$.
\end{proof}

In words, a boundary that lies $d$ bytes past an anchor forced by a window of width $h$, and at least $L$ bytes before
the end of the input, is itself forced by a window of width at most $d + h + L$ that reaches back to the anchor. This
gives the guarantee of Section~\ref{sec:preliminaries} in window form, where a worker can cut at the known boundary and
rescan the suffix. The window carries the anchor, and the rescan is the lemma's induction. The bound is not the only
route to a forced gap, and the measurements below show windows forcing a boundary with no anchor in reach, from local
structure the lemma does not describe.

\paragraph{Reach and its limits.} Fix a non-empty completely tokenizable input $x$ and a width $H$. The
\emph{anchors of $x$ at width $H$} are the positions $t + g$ with $0 < t + g < |x|$ such that some window
$x[t \,..\, t + h)$ with $h \leq H$ has gap $g$ forced. The \emph{reach} $R_H(x)$ is the length of the longest interval
of $x$ between consecutive anchors, with $0$ and $|x|$ counting as anchors, which we call its longest
\emph{anchor spacing}. The \emph{all-width reach} $R_\infty(x)$ admits every finite width, and the \emph{floor} $F(x)$
is the length of the longest token of the tokenization of $x$. The reach is a distance on one input, while the splitting
companion's anchor-free span~\cite{nidhogg2026splitting} is the supremum, over all completely tokenizable inputs, of the
lengths of their runs of positions carrying no anchor from a fixed finite inventory of certificates. With the same
anchors on one input, an interval $[a, b)$ holds $b - a$ positions, unanchored but for $a$, so a distance $d$ is a run
of $d - 1$, or of $d$ from $0$ when position $0$ carries none, and over $\{\code{a}\}$ the input \code{aa} has reach $1$
and no unanchored position. Then
\[ R_H(x) \;\geq\; R_\infty(x) \;\geq\; F(x). \]
The first inequality holds because a wider budget adds anchors and removes none, and the second because every anchor is
a boundary, so none lies strictly inside the longest token, and some interval contains that token whole. The floor is
shared by every method that cuts only at boundaries; below it only a cut inside a token, resumed from scanner state, can
go. Two facts locate the middle term.

(a) With a maximum token length $L$, Corollary~\ref{cor:width} makes every boundary $b$ past an anchor $a$, with
$b + L \leq |x|$, an anchor at some finite width, so past its first anchor and up to $L$ bytes before its end, the
all-width intervals of such an input are single tokens. The interval before the first anchor need not be one. Over
$\{\code{a}, \code{aa}, \code{bbb}\}$ the input \code{aaaaabbba} has boundaries $0, 2, 4, 5, 8, 9$ and the all-width
anchors $5$ and $8$ alone, since a longer run of \code{a} may precede any occurrence, so the all-width reach is $5$
against a floor of $3$. The corollary gives the bound
\[ R_\infty(x) \;\leq\; \max(a_1,\, F(x) + L - 1), \]
with $a_1$ the first anchor of $x$, taken as $|x|$ when there is none. The head interval has length $a_1$, and every
other interval begins at an anchor, where the scan commits a token ending at a boundary $b$; either $b$ is an anchor and
the interval is that token, or $b + L > |x|$ by the corollary and the interval ends within $L - 1$ bytes past $b$. A
brute-force check over every literal set of one to three tokens of length at most $3$ over $\{\code{a}, \code{b}\}$, and
every completely tokenizable input of length at most $10$, found no violation in 82{,}336 inputs, 12{,}526 of them above
the floor. Past the head, the middle term measures how far apart the anchors of the input lie, not an ambiguity wider
windows cannot resolve.

(b) No uniform radius exists for the JSON row. For $H \geq 1$ let $v_H$ be the prefix of length $H$ of the periodic byte
string $\code{,1,1,1}\cdots$. The input $v_H$ tokenizes into one-byte tokens, a comma or the number \code{1} each, so
every gap of that occurrence is a boundary. The input $\code{"}v_H\code{"}$ is one string token, since neither byte is a
quote, a backslash or a control byte, so a token crosses every gap of that occurrence. Both inputs are completely
tokenizable, so every gap of $v_H$ is open, at every width $H$; the two occurrences differ in whether an unmatched quote
lies to their left, which no bounded window sees. The conventional and split-friendly rows hold the same family between
the quotes of a string literal, since outside a string a comma is a punctuation token and \code{1} a number, and the
string form excludes only a quote and a newline. The block-comment rows hold it between \code{/*} and \code{*/}, which
the family never closes, having no \code{*}, and the bare row has neither strings nor comments, so the family says
nothing there.

In the terms of constrained coding, read the token set as a tagged graph, the splitting companion's armed-run automaton
with each edge labeled by its byte and tagged by its boundary bit. A uniform radius $r$ is the tagged analog of the
$(m, a)$-sliding-block decodability of Marcus, Roth and Siegel~\cite[Section~4.3]{marcus2001constrained}, defined there
for tagged encoders, which the armed-run automaton is not. In that definition any two paths
$e_{-m} \cdots e_0 \cdots e_a$ that generate the same word agree on the tag of $e_0$, and the radius maps to
$(m, a) = (r, r - 1)$ when $e_0$ is the byte after the gap, with $r$ bytes preceding it and $r - 1$ following it. The
family above is its failure for every look-behind and look-ahead, since the quoted and the bare reading are two paths
over the same bytes with opposite tags throughout. Locality in the sense of B\'eal, Berstel, Eilers and
Perrin~\cite[Section~6.1]{beal2021symbolic} pins the state that a long enough label passes at a fixed position, where a
uniform radius pins only the boundary tag, and for a local automaton with $N$ states the synchronization delay, the
least window length that pins that state, is quadratic in $N$ in the worst case~\cite{beal1998synchronization}. Deciding
that ceiling for a token set is an analysis of the cycles of the pair graph of the splitting companion's verifier, two
copies run over equal bytes, which this paper does not carry out.

\paragraph{Measurements.} The campaign examined every distinct substring of width $1$ to $H$ of each corpus, with
nothing sampled. It decided the certificate exactly using \code{window\_counterexample()} at every origin and the
profile exactly using \code{boundary\_profile()}. Every search ran under the library's default cap of $2^{20}$ states,
and none reached it. An anchor is a position strictly inside the corpus that a certified origin, or a forced gap, of an
occurring window names, counted once. Every anchor was checked against the sequential scan to be a token boundary, and
every crossed gap was checked at every occurrence. The record is the campaign note
\path{notes/cut-contexts-campaign-2026-09-28.md} and its outputs \path{explorations/cut-contexts/} in the author's
research repository, which is private. The figures are exploratory until an archived record exists, and the probe of
Section~\ref{sec:evaluation} asserts none of them.

The library is munch, commit \code{df818b6a224221cf964095dd8e2c20735f506701}, one commit before the absence verdict
moved to one decision per window, a move that changed nothing here, since no run reported an absent gap. The JSON row
runs over \code{twitter.json} of simdjson v3.10.1 under the JSON lexer of hopper, the author's parser
library~\cite{hopper}, \path{libs/json/src/lexer.cpp} at commit \code{caff92b}, unchanged since release v1.0.0 and
compiled into the probe as it stands. The shortest input on which this lexer differs from the recovery companion's JSON
lexer is a string literal containing an ill-formed UTF-8 byte, found by the library's segmentation decision. The C-like
rows run over the 512 KiB generated corpora of the recovery companion's campaign~\cite{nidhogg2026panicmode}, each under
a token set the same decision proves equivalent to the corpus's generator. Two literal vocabularies of the splitting
companion's supply study run over a 16 KiB slice of \code{twitter.json}: a byte-pair vocabulary of 384 merges trained on
that corpus, 640 tokens, and the first 4{,}096 merges of the GPT-2 vocabulary, 4{,}352 tokens. The width $H$ is $4$ for
JSON and C-like and $3$ for the vocabularies, and Table~\ref{tab:profile} gives the anchor density and the longest
anchor spacing of both detectors.

\begin{table}
\centering
\caption{Anchors from the certificate and from forced gaps at the campaign width $H$: distinct positions per KiB of
corpus, the gain of the profile over the certificate, and the longest anchor spacing in bytes under each, the corpus's
ends counting as anchors. The block-comment row is the campaign's token set proved equivalent to its corpus's generator,
and munch's own block-comment grammar gives the same figures on that corpus.}
\label{tab:profile}
\setlength{\tabcolsep}{4pt}
\begin{tabular}{lcrrrrr}
\toprule
 & & \multicolumn{2}{c}{Anchors per KiB} & & \multicolumn{2}{c}{Longest anchor spacing} \\
\cmidrule(lr){3-4} \cmidrule(lr){6-7}
Token set & $H$ & certificate & forced gaps & Gain & certificate & forced gaps \\
\midrule
JSON, \code{twitter.json} & 4 & 94.64 & 123.73 & 30.7\% & 465 & 465 \\
C-like, conventional & 4 & 32.01 & 32.01 & 0 & 104 & 104 \\
C-like, split-friendly & 4 & 47.33 & 47.33 & 0 & 104 & 104 \\
C-like, block comments & 4 & 6.71 & 6.71 & 0 & 965 & 965 \\
C-like, bare & 4 & 324.42 & 324.42 & 0 & 10 & 10 \\
Byte-pair, 384 merges & 3 & 230.69 & 240.31 & 4.2\% & 61 & 61 \\
GPT-2, 4{,}096 merges & 3 & 620.88 & 630.75 & 1.6\% & 13 & 13 \\
\bottomrule
\end{tabular}
\end{table}

The gain has two sources. A forced gap $k$ is the certificate of the window's one-byte extensions, so at the next width
the certificate can name that gap. The C-like rows gain only below the campaign width, and only this way. The
conventional row gains 2{,}140 positions at width $2$, the split-friendly row 10{,}506 at width $1$ and 2{,}140 at width
$2$, the bare row 26{,}499 at width $1$, and the block-comment row none, every one from the gap after the window. The
gain is gone from width $3$ on the first two rows and from width $2$ on bare, and the profile shortens no C-like
interval at the campaign width. The other source is several boundaries per window, which no certificate names, and
16{,}517 of JSON's 17{,}940 profile-only positions come from a gap inside the window. One such window is two spaces, a
closing bracket and a newline. No string holds a raw newline, so the occurrence lies outside every string, the bracket
is a token and the newline begins a whitespace run; gaps $2$ and $3$ are forced, while the certificate names $3$ alone.

The 465-byte interval on JSON is one string literal, the floor itself, which nothing at any width can shorten. On the
vocabularies the gain is 4.2 and 1.6 percent, and at width $4$ the 384-merge vocabulary's longest anchor spacing under
forced gaps falls from 61 to 40 bytes, the floor, while under the certificate it stays at 61. The model of
Section~\ref{sec:model} equals the exact certificate on every JSON and C-like row and falls short only on the
vocabularies, at 224.88 against 230.69 anchors per KiB on the 384-merge vocabulary and 615.38 against 620.88 on the
GPT-2 prefix.

A second run asked how wide a window must be to force the boundaries inside each corpus's largest intervals between
consecutive anchors. Up to 40 boundaries per row, evenly spaced, were asked of \code{boundary\_counterexample()} over a
ladder of widths $4$ to $128$ with three splits each, and of a counterexample-guided growth beside it, every answer
checked against the scan. JSON and bare are at the floor, with no interval longer than the longest token, 465 bytes of
one string on JSON and 10 on bare, so nothing was searched there, but the other C-like rows are not. The conventional
row's longest interval is 104 bytes holding 17 token starts, where the longest token is 15 bytes, and the block-comment
row's is 965 bytes holding 283, where the longest token is 135 bytes. The conventional and split-friendly corpora are
byte-identical, and both rows have a longest anchor spacing of 104 bytes. Every searched boundary resolves, all 40 on
the conventional and split-friendly rows by width 128, and on block comments 7 of 40 by 128 and all 40 by 1{,}024. Every
resolving window the search found on a C-like row reaches back to within a byte of the interval's left anchor, or past
it, 40 of 40 on the conventional row over distances of 5 to 95 bytes, 40 of 40 on split-friendly over 5 to 96, and 40 of
40 on block comments over 1 to 912. A boundary 95 bytes past its anchor resolves with a guided window of 106 bytes, and
one 912 bytes past with 933. These resolving windows resemble propagation from an anchor in Corollary~\ref{cor:width},
which proves it only for a token set with a maximum token length, a hypothesis the C-like rows do not meet. Inside the
conventional row's searched intervals the longest goes 104, 101, 91, 82, 44 and 16 bytes at widths $4$ to $128$, and
inside the block-comment row's it goes from 965 to 857 at 128 and to 127 at 1{,}024. In these C-like samples, the
windows that brought the longest anchor spacing near the floor were about as wide as the original intervals.

The vocabularies differ from each other. The 384-merge vocabulary's 12 searched boundaries resolve by width 16, 11 of
them by windows that do not reach back to within a byte of the anchor, while the GPT-2 prefix's 6 resolve by width 16,
every one by a window reaching back to within a byte of its anchor, over 3 to 12 bytes.

Three neighbors frame the profile. Kaplan's pinch-points~\cite{kaplan2005tokenizing}, the concatenation check of Lester,
Ong and Sch\"afer~\cite{lester2016flow} and ZipLex's separability predicate~\cite{chassot2026ziplex} each certify one
join, of a text at a position or of two code strings, where a forced gap quantifies over every occurrence in every
completely tokenizable input. Resolved context provides reach past the floor where boundaries cannot, as in the
data-parallel simulation of Mytkowicz, Musuvathi and Schulte~\cite{mytkowicz2014dpfsm} and the bounded alternative
states of Barenghi et al.~\cite{barenghi2015parallel}, which enter a chunk in every live state and pay a scan per state
kept.

\section{Evaluation}
\label{sec:evaluation}

The artifact runs the evaluation in the default test target and CI; the figures of this section are asserted, so a
drifted number fails the test suite.

The random sweep is the output of \code{tools/probes/src/window\_gate.cpp} at munch commit
\code{707a79941472885a260c0bc96e615dd2fb5e99d2}, the commit at which the positive-width equivalent entered the library.
The v1.3.3 release the rest of this evaluation describes excludes the nullable sets and asserts the earlier counts. Over
400 random token sets on a three-symbol alphabet, generated by the probe with a pinned seed and draw order, 266 are
nullable and are decided through their positive-width equivalent, as the artifact compiles them, and 63 certify at least
one byte exactly. Of the 337 that certify no byte,
\textbf{326 gain a certified window under the model, and 322 of those are witnessed}. For each of the 322, the bounded
search finds a completely tokenizable input containing a certified window, with the covering token beginning at the
reported origin, and each input is verified as it is constructed, with the 322 aggregate asserted. Another 4
model-positive grammars have no occurrence within the bounded witness search and are reported as unresolved, never as
rescues, and the remaining 11 exhaust the quotient with no window under the model. No model search reaches its threshold
of 200{,}000 retained keys, so none is inconclusive in the sense of Section~\ref{sec:quotient}, and those 4 grammars are
the only ones the bounded witness search leaves undecided. Occurrence is a property of the concrete word and is not
determined by its quotient key, so the witness search continues past the shortest certified length instead of stopping
at the first certifying word.

Separate rewind-stress rows exercised 1{,}079{,}392 generated executions that scanned through the window and contained
at least one rewind, with zero disagreements against the shipped scanner; 418{,}466 of those executions tokenize their
whole input completely and the remaining 660{,}926 have malformed suffixes past the window, both counts asserted. This
is an implementation stress check, since the generated inputs were required to scan through the window, with no
requirement of complete tokenization. The random sweep additionally checked every certified two-byte window over the
probe's generated contexts. The length-one case reproduces the published certificate on all 400 grammars, the 266
nullable ones included.

Among the named token sets, all six exact-empty rows of the predecessor's applicability table, five C-like variants and
JSON, gain witnessed windows, and one new cumulative C-like variant joins them as a seventh positive row.
Table~\ref{tab:windows} shows the concrete windows, with the retained key counts, one indicator of the search footprint,
in place of the $6^{|Q^+|}$ bound, and beside the gate's six-status keys it reports the keys of the three-status
quotient and the nodes of the chosen-origin search of Section~\ref{sec:smaller}. Those come from the figure program
\code{paper/figures/window\_search.cpp}, added after release v1.3.3, which runs both over the same token sets, asserts
that their shortest lengths equal the gate's on every row, and pins the counts.

\begin{table}[!htbp]
\centering
\caption{Certified windows for the named token sets, from the gate's asserted rows and the figure program's: the
shortest model-certified length, one example window with its origin, and the quotient keys the search retained before
shortest-window stopping, under the six-status key of the gate and the three-status key of Section~\ref{sec:three}. The
last column is the nodes the chosen-origin search of Section~\ref{sec:chosen} holds when its shortest window appears, or
when the search exhausts. Every positive row's displayed window carries an asserted occurrence witness. Among the
positive rows the cumulative one is new to this study, and every other positive row is an exact-empty row of the
predecessor's table.}
\label{tab:windows}
\small
\setlength{\tabcolsep}{3pt}
\begin{tabular}{lccrrr}
\toprule
 & & & \multicolumn{2}{c}{Retained keys} & \\
\cmidrule(lr){4-5}
Token set & \shortstack{Shortest\\model-\\certified $k$} & \shortstack{Example window\\at origin} & six-status &
three-status & \shortstack{Chosen-origin\\nodes} \\
\midrule
C-like, string literals & 2 & \code{\char`\\n!} at 1 & 24 & 20 & 17 \\
C-like, line comments & 2 & \code{\char`\\n!} at 1 & 18 & 13 & 14 \\
C-like, block comments & 4 & \code{\char`\\t*/\char`\\t} at 3 & 53 & 30 & 30 \\
C-like, conventional & 2 & \code{\char`\\n!} at 1 & 27 & 22 & 16 \\
split-friendly, block comments & 4 & \code{\char`\\n*/\char`\\t} at 3 & 188 & 103 & 81 \\
C-like, cumulative (new here) & 4 & \code{\char`\\n*/\char`\\t} at 3 & 189 & 102 & 79 \\
JSON, RFC 8259 & 2 & \code{\char`\\t"} at 1 & 69 & 52 & 29 \\
$\code{a}^+$ (negative row) & none & search exhausted & 3 & 2 & 2 \\
\bottomrule
\end{tabular}
\end{table}

The example windows show the recovery anchors. The string and line-comment rows resynchronize at a newline followed by a
byte that must begin a token, and the block-comment rows at the byte pair \code{*/} followed by whitespace, in comment
context the closer, with the origin immediately after the pair. The certificate is occurrence-universal, so it also
covers occurrences where \code{*/} reads as two operator tokens. The pinned witnesses exercise that reading. The
conventional row certifies at length two because its whitespace runs include the newline, so \code{!} must begin a token
there as well. The JSON row uses the RFC 8259~\cite{bray2017json} lexical forms over bytes and assumes UTF-8 validity;
it is not a conforming JSON processor. Section~\ref{sec:profile} measures the anchor density supplied by certified
windows, in an exploratory campaign the artifact does not assert, and under hopper's JSON lexer those of width at most
$4$ name 94.64 anchors per KiB of \code{twitter.json}. The run $\code{a}^+$ is included as the negative row, and
$\code{a}^*$ would be decided as the same automaton, in which every byte continues a run as readily as it begins one.
The model certifies no window, since absent-byte windows certify only vacuously under Definition~\ref{def:window}, and
the search exhausts its quotient, which is the shape of a model-negative.

\section{Related work}
\label{sec:related}

The window generalizes the certified split point of~\cite{nidhogg2026splitpoints}, and inherits its relation to the
parallel lexing families: composition carries every state and pays for it~\cite{mytkowicz2014dpfsm}; speculation
predicts an entry state, validates, and re-executes on a miss~\cite{prabhu2010speculative}; bounded alternative-state
scans disambiguate without single-state speculation~\cite{barenghi2015parallel}; prescanning pays a pass over the
input~\cite{li2021plex}. The window, like the byte, is derived from the grammar before any input exists. A close
compiler-style neighbor in the maximal-munch setting is the streaming analysis of Li, Yang, and
Mamouras~\cite{li2026streaming}, which statically computes a grammar's maximum token-neighbor distance and uses the
resulting bounded lookahead to emit a sequential maximal-munch stream without backtracking. Their scan advances from a
known token boundary, so the window decides maximality, and parallelization is left there as future work. On the formal
side of the same setting, ZipLex formally verifies linear-time invertible maximal-munch lexing, with an abstraction
capturing the separability of tokens in a sequence~\cite{chassot2026ziplex}, and Li and Mamouras formalize the uniform
tokenization problem and give $O(mn)$-time algorithms, linear in text length $n$ for a fixed grammar of size $m$,
precomputing what each suffix admits in a right-to-left pass before tokenizing from the input's
start~\cite{li2025uniform}. Neither line of work derives an occurrence-universal raw-window certificate or recovers the
covering token's origin from an arbitrary occurrence without scanning from a known boundary. Lester's boxing check
anticipates the flavor at a known join. Every lexer state possible after the first analyzed string must either already
emit a token, making the following character irrelevant to that string's lexing, or emit the same token immediately on
every character that may begin the second, so the join is a lexeme boundary and concatenation preserves token
boundaries. The check is stated compactly in the position paper and in full in the journal
treatment~\cite{lester2013boxing, lester2016flow}. It judges one join of two analyzed fragments at a time, and it
neither certifies every occurrence of a grammar-derived word nor recovers a covering origin.

Recent LLM-tokenizer work addresses input-specific seams instead. LoPT validates position-aligned tokenizations of
overlapping chunks and adjusts chunk length where needed~\cite{shao2026lopt}, and for BPE, recent work bounds the
streaming delay of ordered merge rules from a known beginning~\cite{mamouras2026bpedelay} or maintains the tokenization
incrementally over every prefix~\cite{jiang2026incremental}. Hayase, Liu, Smith, and Oh enumerate, for a byte prefix of
BPE-tokenized text, the tokenizations whose last token straddles the prefix end, a valid covering tree over one concrete
input~\cite{hayase2025bytesampler}. The certificate here is grammar-derived and occurrence-universal, where that work is
input-relative and enumerative, and the certificate fixes one origin for every occurrence. The parsing side of that
pipeline is active again. Cyclic operator precedence grammars~\cite{chiari2025cyclic} admit parallel-suitable chunking
of flat unbounded terminal-string substructures, with precedence relations defined between adjacent terminals. In a
conventional pipeline those terminals are lexer output, so that application presupposes tokenization, though the
formalism is not restricted to pre-tokenized input.

A classical antecedent of window-determines-state is the definite automaton, studied in depth by Perles, Rabin, and
Shamir~\cite{perles1963definite}. Its operational form for a fixed $k$ is $k$-locality, in which any $k$ consecutive
symbols force a unique state, every word of length $k$ being synchronizing~\cite{holub2009parallel}. Both quantify
uniformly over all windows and speak of states; the certificate here is per-window, speaks of token boundaries under
maximal munch, and recovers an origin. Raw-state synchronization alone does not identify the start of the covering
maximal-munch token.

The classical special case of boundary recovery is code synchronization, and Proposition~\ref{prop:codes} states the
relation exactly for prefix codes, where left-to-right tokenization realizes the code factorization. Over a prefix code
the certified windows that occur are the synchronizing splits whose right half sits inside one codeword, and a
synchronizing pair~\cite[p.~375]{berstel2010codes} with a nonempty right component yields, on its concatenation, a
certified window whose origin is the start of the right component's last codeword. The pair notion is older than the
book. Lassez defines a pair of words of $X^*$ as synchronizing when every occurrence of it in a word of $X^*$ forces a
factorization cut between its halves, and proves that whether a regular code has a bounded synchronization delay is
decidable, through the syntactic monoid of $X^*$~\cite[Th\'eor\`eme~2]{lassez1976messages}. Carpi and D'Alessandro
define synchronizing pairs for finite sets that need not be codes and, for a regular synchronizing code, characterize
them on an unambiguous automaton for $X^*$~\cite[Definition~2 and Lemma~2(4)]{carpi2017synchronizing}. An early member
of that family is the comma-free code of Golomb, Gordon and Welch~\cite{golomb1958commafree}, a block code no codeword
of which occurs across the boundary of two adjacent codewords, so that every codeword is a certified window with origin
$0$ for the code's own factorization. Finite synchronization delay~\cite{restivo1975synchronization} bounds how many
codewords a synchronizing pair needs, hence a byte bound for a finite code. The resumption argument, in which an
occurrence of $ww$ in a coded message, for a synchronizing word $w$, lets decoding resume independently from the
position after the first $w$, survives into the partial-DFA treatment~\cite{berlinkov2021partial}. The explicit modern
bounded-window form of that special case is the synchronizing morphism, for which a window of bounded length suffices to
detect boundaries between codewords~\cite{fici2025morphisms}, in a morphic code-factorization setting. Uniquely
decipherable codes may share prefixes; what they guarantee is a unique factorization, with no maximal-munch priority
resolving overlaps between competing token languages, and the origin machinery here addresses that difference.

The relationship to reset words is one-way and stops at length one. A useful certified byte induces a reset-like action
on the partial live automaton with domain $\{q_0\}$, under the stated re-entrancy qualification, a correspondence that
fails under the classical complete-DFA reading~\cite{volkov2008synchronizing}. The correspondence does not extend. A
certified window need not be a reset word of the token DFA at all, since over $\{\code{a}, \code{b}\}$ the window
\code{ab} certifies at origin 1 while the action of \code{ab} on the partial automaton is empty, and a rank-one letter
need not certify either, since over the token language $\code{b}^*\code{a}$ the letter \code{b} can act with rank one
while $q_0$ is re-entrant and \code{b} occurs inside tokens. Certified windows synchronize token-origin information in
an enriched cloud model; they need not synchronize the raw token DFA.

The complexity of the neighborhood is known. Checking careful synchronizability of a partial automaton, and finding a
shortest carefully synchronizing word, are PSPACE-complete already over two-letter
alphabets~\cite[Theorem~1]{martyugin2010careful}, careful meaning the word stays defined from every state and maps all
states to one. Section~\ref{sec:complexity} proves the model's existence problem PSPACE-complete from the same source,
the intersection non-emptiness of DFAs~\cite{kozen1977lower}, by a reduction of the shape Martyugin uses for careful
synchronization of a partial automaton with a zero~\cite[Proposition~2]{martyugin2010careful}. In Martyugin's
construction, a first letter sends every state to a source's start or to the zero, the common word runs the sources, and
a final letter is defined exactly where the sources accept and at the zero. Here the prefix $\code{r}^{m+1}\code{t}$
gives each source's start one fresh origin and parks the pairs from before the window at the guard $G$, and the final
\code{\#} is undefined exactly at the accepting copies and at $G$, since the model needs competing histories to die
rather than to merge. A careful reset word maps every state of the automaton to one, where a certified window need not
synchronize the token DFA's states, only make the cloud's survivors agree on an origin. The per-grammar retained-key
counts, one footprint indicator, stayed far below the worst case throughout.

Two practices are the practical counterparts. Compiler panic-mode recovery discards input to a recovery
set~\cite[Section~4.1]{dragon1986}, which may itself be grammar-derived; the distinction is post-error parser recovery
against a pre-input lexical boundary guarantee. Incremental lexing in the style of Wagner and Graham restarts from
per-token scanner-state snapshots in a versioned token stream, with dynamically maintained lookahead dependencies, exact
for the text they were cached against~\cite{wagner1997lexing}; a certified window is grammar-universal instead, valid in
every completely tokenizable context and known before any input exists, at the price of existing only where the token
set admits one.

Very recent concurrent work establishes nearby but different local guarantees. TokTier certifies input-relative splice
junctions for selected frozen pre-tokenizer and BPE pipelines. It matches cached against fresh runs, uses
family-specific synchronizing character-class transitions proved to reset the pre-tokenizer under every left context,
and proves exact recombination of its overlapping windows~\cite{zhang2026toktier}. ReTokSync monitors the receiver-view
tokenization of one concrete generated stream and triggers a corrective reset when ambiguity
occurs~\cite{wang2026retoksync}. Borsotti, Crespi Reghizzi, and Pradella define precedence relations over synthesized
attributes for a deterministic parser using one-symbol left and right neighborhoods~\cite{borsotti2026attribute}.
Gigatoken certifies cuts before selected punctuation in SentencePiece input where no vocabulary piece contains the two
symbols at the cut in that order, a sufficient condition argued in a source comment and checked by
tests~\cite{rod2026gigatoken}, which is a certified window of width two in subword form. None decides
occurrence-universal bounded words for competing prioritized token languages under maximal munch, nor recovers the
origin of the token covering an arbitrary cut.

Two formal treatments of tokenization address a different question. Kaplan models a language's tokenizing conventions as
a finite-state transduction and emits output at input-relative pinch-points, positions where all live analysis paths
converge to a single transducer state, so settled prefixes stream out as a particular text is
processed~\cite{kaplan2005tokenizing}. Cognetta and Okazaki encode the tokenizations of a regular language as a
finite-state transduction and represent canonical MaxMatch and BPE subword tokenizers as
transducers~\cite{cognetta2025fst}. Neither derives occurrence-universal windows from the grammar alone, and neither
recovers the origin of the covering maximal-munch token at an arbitrary occurrence.

Symbolic dynamics comes closest in shape. A resolving block is a block of a factor subshift all of whose admissible
preimages agree at a selected coordinate~\cite{adler1983sliding, marcus1985sofic}, an occurrence-universal lift of a
local observation to hidden state, which Adler, Coppersmith and Hassner define in their Definition~3.3 and use as a
means of resetting an encoding automaton in sliding-block code construction~\cite[Definition~3.3]{adler1983sliding}. A
certified window shares that shape, reading the window as the observable block and the covering token's start as the
hidden coordinate. The classical theory, however, is stated for factor maps of subshifts, without token priorities,
maximal munch, backup, or any tokenization semantics, and it neither defines nor decides the lexical instantiation.

The areas surveyed above are definite and local automata, synchronizing words and codes under both the complete and the
partial reading, resolving blocks, careful synchronization, streaming, verified, and uniform maximal-munch tokenization,
incremental relexing, the parallel lexing families, and the parallel parsing line. Two special cases lie outside the
claim that follows. In codes, synchronizing pairs and bounded windows already recover codeword boundaries, and bounded
delay is decidable for a regular code~\cite{lassez1976messages, restivo1975synchronization, fici2025morphisms}; in
single formats, a proved two-character pattern fixes the state of a CSV automaton~\cite{ge2019csv}. Beyond those cases
we have not found prior work in these areas that defines in general, for competing prioritized token languages under
maximal munch, a grammar-derived bounded word whose every occurrence identifies the origin of the covering token at an
arbitrary cut. Nor have we found, in these areas, prior work instantiating resolving-block or local-decoding machinery
for prioritized maximal-munch token origins, nor a derivation of this conservative compiled-table cloud model; the
closest results decide maximality from a known boundary, bound retained memory, assume the boundaries they schedule,
recover a state without an origin, or validate a concrete overlap or splice after local retokenization.

We claim no novelty for the decidability of the exact certificate of a proposed window. Longest-match tokenization is a
finite-state transduction~\cite{karttunen1996directed}, so the property is regular and a proposed window is decided by
standard closure, and what this paper supplies is the conservative model, its soundness, the finite quotient that makes
the search terminate and the specializations, together with the smaller searches, the complexity of the model's
existence problem and of the verdicts and the certificate of a given window, the factorization semantics of the cloud,
explicit constructions for the witnessed and global-completeness questions, and the boundary profile's cut results. We
are likewise not aware of an implementation that derives certified windows from a compiled maximal-munch token set and
checks them against a running scanner, other than the instruments this paper reports and their supported successor in
release 1.4.0. For a subword vocabulary, Gigatoken derives two-symbol cuts before selected punctuation from a
SentencePiece vocabulary and checks the chunked output against its own serial encoding in its
tests~\cite{rod2026gigatoken}.

\section{Limitations}
\label{sec:limitations}

All deliberate. The soundness proof speaks only of completely tokenizable inputs, so malformed input is outside the
proof, and no consumed-prefix analog is claimed. Negatives are model-relative. The model refuses windows a greedy
scanner would allow, so an exhausted search means no window \emph{under this model}, never that none exists. The worst
case is exponential, and Section~\ref{sec:complexity} shows that no polynomial algorithm for the model's existence
question exists unless $\mathrm{P} = \mathrm{PSPACE}$, and Section~\ref{sec:long} that shortest windows themselves have
length $2^{\Omega(n)}$ on a family of token sets. The probe and the shipped search are budgeted, though the retained
keys stayed below 200 on every named row, at most 32 with mean 9.9 among the 337 no-byte grammars. The complexity
results concern the compiled automaton, and bounds in the size of a regular expression are not claimed. Those of
Section~\ref{sec:complexity} concern the model. For the property the model approximates, Theorem~\ref{thm:gap-pspace}
proves two questions about a given window PSPACE-complete, whether a given gap is a boundary at some occurrence and
whether it is open, and Theorem~\ref{thm:fixed-nl} proves five more NL-complete, whether the window occurs, whether a
token crosses a given gap at some occurrence, the verdict forced, Definition~\ref{def:window} for a given pair, and the
same for a witnessed pair. Its hardness for the first three uses states that are not live, and over an automaton whose
states are all live those three are decided in deterministic logarithmic space. Whether a token set has a witnessed
certified split window at all is PSPACE-complete by Proposition~\ref{prop:existence}, and a shortest one has
$2^{\Theta(n)}$ bytes in the worst case by Theorems~\ref{thm:semantic-length} and~\ref{thm:plain}, where the exponential
lengths of Theorem~\ref{thm:long} are those of model certificates. Theorem~\ref{thm:plain} proves both for a certified
split window that need not be witnessed. The decision of global completeness has an EXPSPACE upper bound and no matching
lower bound, and the exact witnessed synthesis was not run over the random sweep.

The evaluated v1.3.3 artifact ships no window-planning API, and its probe excludes the nullable sets, so the sweep of
Section~\ref{sec:evaluation} is reproduced at the later commit named there; its certificate machinery is a probe, on the
principle that the certificate's formulation should freeze in this paper before becoming a public contract. Release
1.4.0 made the formulation a public contract, \code{is\_split\_window()}, which the sweep at the later commit asks
beside its own model at the coverage the sweep's source states: every one-byte answer of the named rows and of the 400
random token sets, every model-positive random two-byte answer, and the named certificates and refusals, with the count
of those checks pinned. The window planner that release added beside the decision lies outside this paper's evaluation.
The real-corpus evidence is the exploratory campaign of Section~\ref{sec:profile}, recorded outside the artifact and
asserted by no test, and window occurrence frequency on representative corpora remains a question for a controlled
campaign.

\section{Conclusion}
\label{sec:conclusion}

A certified split window $(W, o)$ of a compiled token set is a byte string such that in every completely tokenizable
input containing $W$, the token covering the occurrence's final byte begins exactly $o$ bytes into it. Model
certification is decided from the compiled tables alone by running a conservative cloud of token-prefix hypotheses
across the window, seeding new trajectories only where the automaton accepts, and demanding unanimity on an in-window
origin. The soundness argument is a representation invariant. The final segmentation's actual token-prefix history is
always among the hypotheses, so a unanimous origin is always the actual one, and maximal-munch backup never needs
simulating, because the model tracks where tokens begin.

The search is a decision procedure for the model. A quotient of the reachable clouds, exact under the stated
single-occupancy invariant, makes breadth-first search terminate by exhaustion, so every answer is a certified window
with its origin or a proof that the model admits none at any length. Two smaller searches decide the same question, a
three-status quotient and an unambiguous automaton that follows one chosen origin, the latter shipped as
\code{shortest\_split\_window()} by release 2.2.0 of munch. Existence of a model certificate is PSPACE-complete, plain
or witnessed, so no polynomial-time algorithm decides it unless $\mathrm{P} = \mathrm{PSPACE}$, and separately a
shortest window can be $2^{\Omega(n)}$ bytes long over $n$ states. The cloud's hypotheses are exactly the histories of
arbitrary token factorizations of words of $X^*$, the concatenations of tokens, which is why the vacuous window is
certified, why token kinds do not matter to the search, why the model is exact on the polynomially decidable class where
every factorization is the scan's, and how a shortest witnessed window and the model's global completeness are decided.

The three-layer reading of the results is intentional and permanent. Model-positive answers are semantically certified
within the stated flat, completely-tokenizable scope. Negatives are model-relative. The model refuses windows a greedy
scanner would allow, the two asserted examples exhibit the over-approximation, and the strictness corollary locates its
minimum nonvacuous length at two, since at length one the model provably coincides with the published predicate, made
exact for occurring bytes by the predecessor's necessity theorem. Over a prefix code the certificate is code
synchronization under another name, the occurring certified windows being the synchronizing splits whose right half sits
inside one codeword, and the named token sets studied here are not codes. A token set in which some token matches the
empty string is decided through its positive-width equivalent, which unrolls the accepting start state and changes no
scan, so nothing is set aside on that account.

Of the 337 random token sets certifying no byte, 322 gain a witnessed certified window and four more a model certificate
for which the bounded search finds no occurrence, with no model search reaching its threshold, and the predecessor's six
exact-empty rows all gain witnessed windows of two to four bytes, with the concrete windows in Table~\ref{tab:windows}.
Separate rewind-stress rows exercised 1{,}079{,}392 generated executions that scanned through the window and contained
at least one rewind, with zero disagreements against the shipped scanner, and the figures are printed and asserted in
the artifact, so they fail the test suite if they move. By choice, the evaluation measures no parallel scan cut at
windows, although Section~\ref{sec:cut} proves when such a cut is exact. A controlled evaluation on representative
corpora, with planning times, occurrence frequencies and end-to-end comparisons against byte certificates, is the
natural next step.

Read gap by gap, a certificate at origin $o$ is one pattern of the boundary profile of a window of $k$ bytes, gap $o$
forced and the gaps $o + 1, \ldots, k - 1$ crossed. A rule that cuts at a window's gap and reconciles nothing is sound
exactly at forced gaps, so the profile is the whole anchor supply a window of its width can give such a rule. Whether a
given gap is a boundary at some occurrence is PSPACE-complete to decide, and so is whether it is open, while occurrence,
the verdict forced and the certificate of a given pair are NL-complete, with an occurrence, a crossing and an occurrence
refuting an origin each found in an input of polynomial length when one exists, where a shortest input refuting a
crossed verdict can need $2^{\Omega(n)}$ bytes over $n$ states. Whether a token set has a certified split window at all
is PSPACE-complete, witnessed or not, and a shortest one has $2^{\Theta(n)}$ bytes in the worst case, for
Definition~\ref{def:window} as for the model. Over a token set with a maximum token length $L$ a forced gap propagates,
the scan from it forcing every boundary it commits with its lookahead inside the window, so a boundary lying after an
anchor and at least $L$ bytes before the end of the input resolves at a width of at most its distance from the last
preceding anchor plus that anchor's window width and $L$. No single width decides every gap of the JSON row, or of the
C-like rows with strings or comments, since no bounded window sees whether an unmatched quote or an open comment lies to
its left. At the campaign width the profile adds anchors on JSON and the byte-pair vocabularies, none on the C-like
rows, and shortens no longest anchor spacing, whose floor is the longest token.

Like its predecessor, this paper contributes the certification step and no new way to recover context. It extends that
step from single bytes to bounded windows, with the origin recovered and the conservatism exhibited. Certified windows
synchronize token-origin information in an enriched cloud model, and need not synchronize the token DFA itself; the
predecessor's one-way reset-word connection applies at length one, and no general implication extends to longer windows.

\section*{Tools}

Large language models assisted with drafting, with checking citations against primary records and with reviewing the
implementation; the author verified each suggestion by derivation, against primary sources, or against the
implementation, its tests and the archived artifacts, and is responsible for all content.

\appendix
\section{Proof of Theorem~\ref{thm:plain}}
\label{app:plain}

\emph{A token set from a set of states.} Let $B = (Q, \{\code{a}, \code{b}\}, \delta_B)$ be a strongly connected
complete deterministic automaton with at least two states, let $S$ be a set of $m \geq 1$ states $s_1, \ldots, s_m$ of
it, and let $f$ be a state. The token automaton has the symbols \code{H}, \code{t}, \code{j}, \code{a}, \code{b} and
\code{E}, and $\Gamma$ stands for the five other than \code{H}. Its states are the root $q_0$, states $h_0, \ldots, h_m$
and $g$, states $c_1, \ldots, c_m$, the states of $Q$, an end state $T$ and a state $Z$, of which $h_0$, $c_1$, $T$ and
$Z$ accept. The root has $q_0 \xrightarrow{\code{H}} h_0$ and $q_0 \xrightarrow{\code{t}} c_1$ and leads to $T$ on
\code{j}, \code{a} and \code{b}. A state $h_i$ with $i < m$ has $h_i \xrightarrow{\code{t}} h_{i+1}$, the last has
$h_m \xrightarrow{\code{j}} g$, and $g$ keeps itself on \code{a} and \code{b}. Every other symbol of $\Gamma$ leads
these states to $Z$, except that $g$ does not read \code{E}. A state $c_i$ has $c_i \xrightarrow{\code{j}} s_i$, and
$c_i \xrightarrow{\code{t}} c_{i+1}$ when $i < m$. The states of $Q$ read \code{a} and \code{b} as in $B$, each of them
but $f$ leads to $Z$ on \code{E}, and $Z$ keeps itself on every symbol of $\Gamma$. There are no other transitions. The
automaton has $|Q| + 2m + 5$ states, its root does not accept, and nothing enters the root. Every state is reached from
the root. From every state of $Q$ a state other than $f$ is reached, which reads \code{E}, the states $c_i$ lead into
$Q$, and the states $h_i$ and $g$ lead to $Z$, so every state is live.

Call a word \emph{bad} when it is $\code{H}\, \code{t}^m\, \code{j}\, w\, \code{E}$ for a word $w$ of letters \code{a}
and \code{b} with $\delta_B(S, w) = \{f\}$. A word over the six symbols occurs in a completely tokenizable input exactly
when it has no bad factor. To see this, note that \code{H} is read from the root alone, so every \code{H} of a
completely tokenizable input begins a token and no token reaches past the next \code{H}. Call a \emph{segment} an
\code{H} with the bytes $y \in \Gamma^*$ up to the next \code{H} or the end of the input. The scan from the \code{H} of
a segment is the scan of the input $\code{H}\, y$. From $h_0$ the run over $y$ follows the form
$\code{t}^m\, \code{j}\, w\, \code{E}$ through the states $h_i$ and $g$ without accepting, enters $Z$ at the first byte
that departs from the form, and is undefined at an \code{E} that completes it. If $y$ departs from the form before
completing it, the run enters $Z$ and the segment is one token. If $y$ is a proper prefix of a word of the form, the
scan commits \code{H} and then every byte of $y$ singly, since a token that begins with \code{t} and has more bytes
contains \code{E}, and \code{j}, \code{a} and \code{b} lead the root to $T$, which reads nothing. In the remaining case
$y = \code{t}^m\, \code{j}\, w\, \code{E}\, v$, and the scan commits \code{H}. When it then stands at a \code{t} with
$i$ bytes \code{t} from there to the \code{j}, the run from the root is in $s_i$ after the \code{j} and in
$\delta_B(s_i, w)$ before the \code{E}. If that state is not $f$, the run enters $Z$ and the rest of the segment is one
token. If it is $f$, the run has accepted after its first byte only, and the scan commits that \code{t}. The scan comes
to these bytes \code{t} in turn, $i$ running from $m$ down to $1$ for as long as it commits them singly, and after the
last it commits \code{j} and the letters of $w$ singly and fails at the \code{E}, which begins no token. A segment thus
fails to be completely tokenizable exactly when it begins with such a word in which $\delta_B(s_i, w) = f$ for every
$i$, which is a bad word. In a completely tokenizable input the \code{H} of a bad factor would begin a segment that
begins with the factor, so there is none. Conversely, for a word $x$ without a bad factor the input
$\code{H}\code{E}\, x$ is completely tokenizable. Its first segment departs from the form at \code{E}, as $m \geq 1$,
and each of the others begins at an \code{H} of $x$ and with no bad word.

When no word sends $S$ to $\{f\}$, no word is bad and every word occurs, and $L$ stands for $\infty$. Otherwise let $z$
be a shortest word with $\delta_B(S, z) = \{f\}$, let $d = |z|$ and $L = m + d + 3$. The word
$X = \code{H}\, \code{t}^m\, \code{j}\, z\, \code{E}$ is bad and has $L$ symbols, and no bad word is shorter, so $X$
does not occur and every word of fewer than $L$ symbols does.

\emph{Three bits for a symbol.} Write \code{H} as \code{000}, \code{t} as \code{101}, \code{j} as \code{010}, \code{a}
as \code{011}, \code{b} as \code{110}, and \code{E} as any of \code{111}, \code{001} and \code{100}, so that every block
of three bits is a code of exactly one symbol. In the token automaton over \code{0} and \code{1}, whose bytes are called
bits here, every state reads the blocks of its symbols through a trie of at most six fresh states that do not accept,
all of them live. It has $N \leq 7 (|Q| + 2m + 5)$ states, all live, and one token kind, its root does not accept, and
nothing enters the root. Its accepted words are the words of whole blocks that decode to accepted words. An input is
therefore completely tokenizable exactly when it consists of whole blocks and decodes to a completely tokenizable input,
and its tokens begin at offsets divisible by three.

Let $W$ be a word of $k \leq 3L - 4$ bits. Let $\rho$ be $0$ or $1$, and put $\rho$ bits before $W$ and at most two
after it to make whole blocks, which are at most $L - 1$, since $\rho + k \leq 3L - 3$. They decode to a word $x$ that
occurs, so the input of the blocks \code{000} and \code{111} followed by these, which decodes to $\code{H}\code{E}\, x$,
is completely tokenizable and contains $W$ at offset $6 + \rho$. For an origin $o$ of $W$ the offsets $6 + o$ and
$7 + o$ are not both divisible by three, so at one of the two occurrences no token begins at gap $o$, and $(W, o)$ is
not certified. A certified split window therefore has at least $3L - 3$ bits, and there is none when no word sends $S$
to $\{f\}$.

Suppose now that $z$ exists. Let $K$ be the coding of $X$ with \code{E} written \code{111}, of $3L$ bits, and let
$Y = \code{1}\, K\, \code{1}\, K\, \code{1}\, \code{000}$, of $6L + 6$ bits, with its origin $o$ at the first of its
last three bits. Let $Y$ occur at offset $t$ in a completely tokenizable input. The copies of $K$ begin at the offsets
$t + 1$ and $t + 3L + 2$. If $t + 1$ were divisible by three, the first copy would consist of whole blocks of the input,
which decode to $X$, and if $t + 2$ were, the second would. Since $X$ does not occur, $t$ is divisible by three, the
last three bits of $Y$ are a block, the code of \code{H}, and a token begins at it. The pair $(Y, o)$ is therefore
certified. It is witnessed as well. In $K$ the block \code{000} comes first and is followed by \code{101}, as
$m \geq 1$, and no other block of $K$ begins or ends with \code{00}, so three bits \code{0} in a row begin in $Y$ only
at the offsets $1$ and $3L + 2$, where the copies of $K$ begin, and at $o$. Cut into blocks from its first bit, $Y$ thus
decodes to a word whose only \code{H} is its last symbol, which has no bad factor, and the blocks \code{000} and
\code{111} followed by $Y$ are a completely tokenizable input.

\emph{The two claims.} For the length, take the automata $B$ and sets $S$ that the proof of Vorel's
theorem~\cite[Theorem~8, p.~378]{vorel2014subset} builds, one for every power of two $\ell \geq 2$, with
$|Q| = 288 \ell + 48 \log_2 \ell + 168$ and a shortest word sending $S$ to one state of at least $2^\ell$ letters, and
let $f$ be the state that a shortest such word sends $S$ to. Then $d \geq 2^\ell$ and $N \leq M(\ell) = 7 (3 |Q| + 5)$,
where $M(\ell) = \Theta(\ell)$. The token set has the witnessed certified split window $(Y, o)$, and its certified split
windows all have at least $3L - 3 > 2^\ell$ bits, which is $2^{\Omega(N)}$. For every sufficiently large $n$ the largest
power of two $\ell$ with $M(\ell) \leq n$ has $n < M(2 \ell) < 2 M(\ell)$, so $\ell = \Omega(n)$ and the bound holds
over token sets with at most $n$ states as well. With this $\ell$, the token set built from $B$ and $S$ in the proof of
Theorem~\ref{thm:long} has $|Q| + 2m + 6 \leq M(\ell) \leq n$ states, all live, and its shortest model-certified
windows, plain or witnessed, have $m + d + 3 > 2^\ell$ bytes, so the model's lower bound holds over at most $n$ live
states too. For the hardness, whether some word sends a given set of states of a strongly connected complete automaton
over two letters to one state is PSPACE-complete~\cite[Theorem~11, p.~379]{vorel2014subset}. An instance with one state,
or with an empty set, is decided at once. For any other let $f$ be any state. A word that sends $S$ to one state can be
continued to $f$, the automaton being strongly connected, so some word sends $S$ to one state exactly when some word
sends it to $\{f\}$, and by the last two paragraphs exactly when the token set over two symbols has a certified split
window. That token set is written in polynomial time. For membership, by Lemma~\ref{lem:contexts} a pair $(W, o)$ is a
certified split window exactly when in every completely tokenizable $u W v$ with $u \in \mathcal{U}$ and
$v \in \mathcal{V}$, for the lists of the lemma, the token containing the last byte of the occurrence at offset $|u|$
begins at gap $o$. The machine in the proof of Proposition~\ref{prop:existence} decides whether some pair is one when it
no longer requires a scan that reaches the end of its input without failing, in the same space.

\bibliographystyle{plainnat}
\bibliography{refs}

\end{document}